\documentclass{amsproc}

\usepackage{bbm}
\usepackage{amsmath,amssymb}
\usepackage{mathrsfs}  
\usepackage{subcaption}

\usepackage{pgfplots}
\usetikzlibrary{decorations.pathreplacing, positioning}

\newtheorem{theorem}{Theorem}[section]
\newtheorem{lemma}[theorem]{Lemma}
\newtheorem{paradox}[theorem]{Paradox}
\newtheorem{proposition}[theorem]{Proposition}
\newtheorem{corollary}[theorem]{Corollary}

\theoremstyle{definition}
\newtheorem{definition}[theorem]{Definition}
\newtheorem{example}[theorem]{Example}
\newtheorem{xca}[theorem]{Exercise}

\theoremstyle{remark}
\newtheorem{remark}[theorem]{Remark}
\usepackage{tikz}
\usetikzlibrary{patterns, decorations.markings,decorations.pathmorphing,arrows, calc}

\numberwithin{equation}{section}

\newcommand{\abs}[1]{\lvert#1\rvert}

\newcommand{\blankbox}[2]{%
  \parbox{\columnwidth}{\centering
	   \setlength{\fboxsep}{0pt}%
    \fbox{\raisebox{0pt}[#2]{\hspace{#1}}}%
  }%
}

\begin{document}

\title[The role of integrability in FPUT-like models]{The role of integrability in Fermi-Pasta-Ulam-Tsingou-like models}

\author{Matteo Gallone}
\address{Mathematics Area, SISSA, Via Bonomea 265, Trieste (Italy)}
\curraddr{Dipartimento di Matematica ``F.~Enriques'',
Università Statale di Milano, Milano, Italia, 20133}
\email{matteo.gallone@unimi.it}
%TO BE FILLEDDDDDDD!!!!    \thanks will become a 1st page footnote.
%\thanks{}

%\author{Author Two}
%\address{Mathematical Research Section, School of Mathematical Sciences,
%Australian National University, Canberra ACT 2601, Australia}

%    General info
\subjclass[2020]{Primary 35Q53  -- KdV equations, 37K60 -- Lattice dynamics; integrable lattice equations}
\date{\today}


\keywords{Thermalization, Prethermalization, Fermi-Pasta-Ulam-Tsingou, Toda lattice, Korteweg-de Vries}

\begin{abstract}
In these lecture notes, we present, contextualize, and discuss the phenomenon of metastability (or prethermalization) in the Fermi-Pasta-Ulam-Tsingou lattice and similar models from the viewpoint of perturbation theory. We provide an updated state-of-the-art description of this intermediate state in terms of the underlying integrable dynamics, in particular the Korteweg-de Vries and inviscid Burgers equations.
\end{abstract}

\maketitle

\tableofcontents

\section{Introduction}
Seventy years ago Enrico Fermi, John Pasta, Stanislaw Ulam and Mary Tsingou revolutionized the study of nonlinear dynamical systems by performing the first series of numerical simulations with the Maniac II computer at Los Alamos. In modern terms, they were interested in the dynamics of Fourier modes along the time evolution of the Hamiltonian system
\begin{equation}\label{eq:IntroFPUHam}
	H=\sum_{j=1}^N \left[\frac{p_j^2}{2}+\Phi(q_{j+1}-q_j) \right]
\end{equation}
where $p_j,q_j \in \mathbb{R}$ and $\Phi$ models a nearest neighbor interaction
\begin{equation}
	\Phi(z)=\frac{1}{2}z^2+\frac{\alpha}{3}z^3+\frac{\beta}{4}z^4+\dots
\end{equation}
in which the real numbers $\alpha,\beta,\dots$ measure the strength of the nonlinearity.\footnote{To give a precise mathematical model one shall also fix boundary conditions, e.g.~$q_N=q_0$, $p_N=p_0$ (periodic boundary conditions) or $q_N=q_0=p_N=p_0=0$.} The equations of motion associated to $H$ can be written as second order ordinary differential equations
\begin{equation}
	\begin{split}
	\ddot{q}_j&=\Phi'(q_{j+1}-q_j)-\Phi'(q_j-q_{j-1}) \\
	&\simeq q_{j+1}+q_{j-1}-2 q_j + \alpha \Big[(q_{j+1}-q_j)^2-(q_j-q_{j-1})^2 \Big]+\cdots
	\end{split}
\end{equation}
from which one explicitly recognizes a nonlinear perturbation of the harmonic chain.

Their goal was to investigate the statistical behavior of energy redistribution in one-dimensional nonlinear chains of particles, with the expectation that energy would gradually spread among the modes, leading to thermal equilibrium -- a fundamental assumption in statistical mechanics and hydrodynamics.

Surprisingly, the numerical results did not meet the expectations and it became customary to refer to this problem as to the \emph{Fermi-Pasta-Ulam (FPU) paradox} and to the model \eqref{eq:IntroFPUHam} to the \emph{FPU model} or chain. The letter `T' to the acronym was added much later, when the role of Mary Tsingou was acknowledged and thus, nowadays, one refers to the FPUT paradox and to the FPUT model (see \cite{Dauxois2008}). Instead of a straight relaxation to equilibrium, the numerical experiment showed a strikingly regular and recurrent dynamics. This unexpected behavior challenged the prevailing understanding of ergodicity and thermalization, calling for extensive further investigations. Among the most notable consequences of these findings was the pioneering work of Zabusky and Kruskal, who linked the anomalous dynamics of the FPUT model to the integrable structure of the Korteweg-de Vries (KdV) equation. This insight opened the way for the discovery of solitonic behavior and the broader field of infinite dimensional integrable systems.

At present, vicinity to integrable systems is used as a paradigm to understand the slow thermalization observed in many different physical systems. Thanks to advances in experimental techniques, phenomena similar to those observed in the FPUT model can now be studied in controlled laboratory settings, particularly in quantum many-body systems. 

In these lecture notes, we will introduce the FPUT model, outline its historical and physical motivations, and explore its deep connections with the KdV hierarchy and the Toda chain. By surveying both classical results and contemporary developments, we aim to provide a comprehensive perspective on the role of integrability in FPUT-like models and its implications for modern physics and mathematics.

Before starting, let me remark that the style of the notes is informal and an they reflect the author's point of view on the question. These lecture notes are not to be meant in any sense exhaustive of the topic. %Indeed, in 2025, in occasion of the 70th anniversary of the original Los Alamos report, a special issue on \emph{Journal of Statistical Physics} is going to appear discussing an updated version of the topic. 

I tried to do my best to clarify and simplify many questions and ideas, to make them accessible to graduate students and I also tried to highlight the modernity of the question(s) posed by Fermi and collaborators. If I had succeeded in clearly presenting the beauty, complexity, and interest of the problem, I would be pleased. If I had not, the fault would be mine alone, and I encourage the reader to consult other sources, perhaps more stimulating. A particularly valuable starting point is the volume compiled on the occasion of the problem’s fiftieth anniversary \cite{GallavottiBook2008}.

\section{The problem of thermalization}

\subsection{Macroscopic and microscopic dynamics}
One of the first physical phenomena we encounter in everyday life is macroscopic irreversibility. Take, for instance, my breakfast: coffee is a mixture of various substances that appear homogeneous in a (usually) tasty drink. During the preparation, I put some brown powder and water into the moka, and then the heat does something fascinating. From the top part of the moka, a brown, flavorful liquid emerges. In essence, the two ingredients have combined in what seems to be an irreversible process (I have never experienced the two substances separating back to their original states in my cup).

We can find countless examples where \emph{macroscopic irreversibility} plays a role in our daily lives. It is worth mentioning that, for a thermally isolated system, this irreversibility is governed by the \emph{second law of thermodynamics}, which states that the entropy of the final state of a thermodynamic transformation is always greater than or equal to that of the initial state. Macroscopic irreversibility is thus established as one of the fundamental principles of thermodynamics, with thermalization being one of its consequences. Everyday experiences suggest that a system initially driven out of thermal equilibrium will eventually relax to its thermal equilibrium state.% For instance, if we consider a gas in a room and initially concentrate all its particles in one corner, we expect that, after some time, the gas will have evenly and homogeneously filled the entire room.

Although it is a phenomenon that can be easily observed, the mechanism by which physical systems thermalize is still largely not understood, especially in solids. Thanks in part to studies on the Fermi-Pasta-Ulam-Tsingou model, it has been shown to be more complex than initially expected. To date, and quite recently, this type of analysis has led to the establishment of the \emph{prethermalization} paradigm. It has been observed that, in certain regimes, different physical systems reach thermodynamic equilibrium through a two-stage process: initially, the system relaxes to a state where \emph{macroscopic observables} appear frozen, and then, over much longer timescales, their values evolve until they reach those predicted by thermodynamic equilibrium.

To discuss the notion of thermalization, one could begin by recalling that each thermodynamic system is composed of a large number of elementary constituents, such as gas molecules or atoms in a solid. The dynamics of these constituents are governed by Newton's laws (in the classical framework) or the Schr\"odinger equation (in the quantum framework). For the purposes of these lecture notes, we will assume that the elementary constituents are classical particles, obeying Newton's equations in a conservative environment. These equations can conveniently be reformulated as an equivalent first-order system using Hamiltonian mechanics.

By general facts, that we are going to review lately in these lecture notes, the Hamiltonian flow of a system of of $N$ particles confined in a vessel lies on the surface $H(q,p)=E$ which is compact, when we assume that the inter-particle potential is lower semi-bounded. When this is the case, the solution to Hamilton equations is globally defined in time and its flow $\Phi_H^t$ is a one-parameter real group. In particular, it is \emph{reversible}, that means: an ideal gas evolving according to Hamilton equations will concentrate in a corner of a room, provided the initial datum is suitably chosen! This fact is clearly contrary to our everyday experience discussed above. Even more, due to Poincaré recurrence theorem, if initially the gas is concentrated in one of the corners of the room, there will be a future time in which all the gas will be \emph{substantially} concentrated on the same corner. Again, this fact is contrary to our experience and conceptually seems to contradict the possibility of describing \emph{thermalization} from a microscopic point of view.

We can thus summarize those two possibilities as the following two paradoxes:

\begin{paradox}[Loschmidt paradox]
	Mechanical processes are reversible, while thermodynamic processes are not. How does this irreversibility emerge?
\end{paradox}

\begin{paradox}[Zermelo paradox]\label{par:Zermelo}
	In all mechanical conservative systems with compact phase space, for almost all initial data, the dynamics returns infinite times close to the initial condition. How does recurrence disappear in macroscopic systems?
\end{paradox}

Both paradoxes originates form the fact that we didn't define yet what is a \emph{macroscopic state} from a \emph{microscopic point of view}. This latter point is very delicate and it is very difficult to formulate it in a clear, universal and mathematically precise way. Ergodic theory, at least in its origins, is an attempt to provide such a definition and it is the starting point to contextualize and analyze the Fermi-Pasta-Ulam-Tsingou problem. Before entering ergodic theory, let us discuss some heuristics.

\subsection{Macroscopic observables as time averages}\label{subsec:MacroTimeAverages} To focus our attention and develop some intuition, let us discuss a possible definition of macroscopic observables describing a perfect gas of identical particles inside a box of size $L$. A perfect gas is a collection of particles that moves freely without interacting among them and which collides elastically with the walls of the box. Experimentally it was verified that for temperatures high enough and densities low enough, the thermodynamic state of a gas can be described by only four variables: the volume $V=L^3$ of the box, the pressure $P$ that the gas exerts on the walls of the box, the number of particles of the gas $N$ and its temperature $T$. Not all these quantities are independent, as they have to satisfy a so-called \emph{equation of state} that is
\begin{equation}\label{eq:LawPerfGas1}
	P \, V \, = \, N \, k_B \, T
\end{equation}
where $k_B$ is a universal constant, called \emph{Boltzmann constant}, whose numerical value is  $k_B=1.380\,649 \times 10^{-23} \, \text{J} \, \text{K}^{-1}$. 

Among the four thermodynamic variables mentioned above, the volume $V$ and the number of particles $N$ already have a natural interpretation in terms of microscopic dynamics, while $P$ and $T$ do not.

From its very physical definition, pressure is the force exerted by the gas on one of the faces of the box divided by the surface of the face. Let us concentrate, for the sake of concreteness on the right face perpendicular to the $x$ axis as seen in Figure~\ref{fig:Gas}.

Let us consider the contribution of a single particle, whose velocity along the $x$-axis is $v_x$ and let us assume the collision to be completely elastic, so that if the collision happens at time $t_{\text{coll}}$, we have by the impulse theorem\footnote{The impulse theorem states that if to a particle of mass $m$ it is applied a force $F(t)$, then
\[
	\int_{t_0}^{t_1} F(s) \, ds \, = \, m v(t_1)-mv(t_0) \, .
\]} 
\begin{equation}
	\int_{t_{\text{coll}}^-}^{t_{\text{coll}}^+} F(s) \, ds \, = \, 2 m v_x \, ,
\end{equation}
where $v_x$ is the modulus of the initial velocity of the particle and by $t_{\mathrm{coll}}^\mp$ we mean times shortly before and shortly after $t_{\mathrm{coll}}$. For all the times where there is no collision, the force exerted by the particle on the wall is zero. The net force is the derivative of the impulse with respect to time, which is $F(t)=\sum_{n \in \mathbb{Z}} 2 m v_x \delta(t-n t_{\mathrm{coll}})$.\footnote{Here $\delta$ denotes the Dirac delta distribution.} We can reasonably assume that our measurement lasts a time $t_{\mathrm{meas}}$ and the value of the pressure is the time-average over an interval of length $t_{\mathrm{meas}}$. Then, the contribution to the pressure of the single particle labeled by $j$ is
\begin{equation}
	P_j(t)\, :=\, \frac{1}{t_{\mathrm{meas}}} \int_0^{t_{\mathrm{meas}}} \frac{F(s)}{L^2} \, ds \,=\,\frac{1}{L^2} \frac{1}{t_{\mathrm{meas}}} \sum_{\substack{n \in \mathbb{N} \\ n t_{\text{coll}} < t_{\mathrm{meas}}}} 2 m v_{j,x} \, .
\end{equation}
We want now to take into account that $t_{\mathrm{meas}}$ is much longer than $t_{\mathrm{coll}}$. This latter information can be encoded in the definition of pressure by taking the limit as $t_{\mathrm{meas}} \to +\infty$. This yields
\begin{equation}
	P_j \,=\, \frac{1}{L^2} \lim_{t_{\mathrm{meas}} \to +\infty} \frac{1}{t_{\mathrm{meas}}} 2mv_{j,x} \frac{t_{\mathrm{meas}}}{t_{\text{coll}}} = \frac{ m v_{j,x}^2}{L^3} \, .
\end{equation}

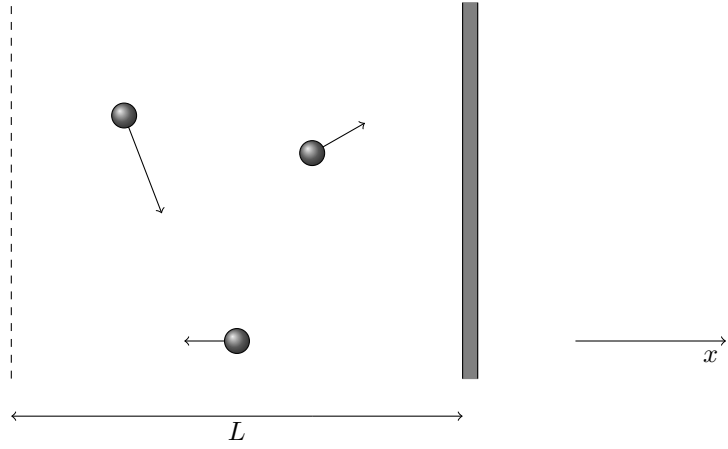
\begin{figure}[t]
	\begin{center}
		\begin{tikzpicture}
			\draw[draw=gray,fill=gray] (0,0) rectangle ++(0.2,5);
			\draw[-] (0,0) --(0,5);
			\draw[-] (0.2,0) -- (0.2,5);
			
			\draw[dashed] (-6,0) -- (-6,5);
			
			\draw[->] (-2,-0.5) -- (-6,-0.5);
			\draw[->] (-2,-0.5) -- (0,-0.5);
			\node at (-3,-0.7) {$L$};
			
			\node at (-3,0.5) [circle, draw, ball color = gray] {};
			\draw[->] (-3.16,0.5) -- (-3.7,0.5);

			\draw[->] (-2,3) -- (-1.3,3.4);
			\node at (-2,3) [circle, draw, ball color = gray] {};				
				
			\draw[->] (-4.5,3.5) -- (-4,2.2);
\node at (-4.5,3.5) [circle, draw, ball color = gray] {};

\draw[->] (1.5,0.5) -- (3.5,0.5);
	\node at (3.3,0.3) {$x$};
		\end{tikzpicture}
	\end{center}
	\caption{Pictorial representation of the gas of particles.}
	\label{fig:Gas}
\end{figure}

We can thus obtain the pressure of the gas on the wall by taking the sum over all particles, that is
\begin{equation}
	P=\sum_{j=1}^N P_j = \frac{m}{L^3} N \frac{1}{N} \sum_{j=1}^N v_{j,x}^2 = \frac{N}{L^3} m \langle v_{j,x}^2 \rangle_{N} \, .
\end{equation}
where $\langle \cdot \rangle_{N}$ denotes the average over the $N$ particles.

Interpreting now $m \langle v_{j,x}^2 \rangle_{N}=k_B T$, we obtain the law of perfect gases \eqref{eq:LawPerfGas1}.

At the end, with this heuristic argument we showed that it is reasonable that microscopic definitions of macroscopic quantities can be given as time-averages along the dynamics:
\begin{equation}\label{eq:TimeAverages1}
	F_{\text{macro}}(q_0,p_0)\,=\, \lim_{t \to +\infty} \frac{1}{t} \int_0^t F_{\text{micro}}(\Phi_H^s(q_0,p_0)) \, ds \, ,
\end{equation}
where $(q_0,p_0)$ is the microscopic initial state of the system. Actually, this idea is not so weird: think of when someone wants to measure their fever. To measure body temperature, the thermometer takes time before providing an accurate result. The procedure just described suggests that this \emph{time-averaging} is the general case. Nevertheless, this definition cancels any dependence on time and therefore, as it is, it cannot be used to detect \emph{thermalization} but it can be a definition for \emph{macroscopic observables at thermal equilibrium}. 

Leaving the thermalization process apart for one moment, we notice that \eqref{eq:TimeAverages1} poses a number of urgent questions:
\begin{itemize}
	\item[1.] Does the integral/limit in \eqref{eq:TimeAverages1} make sense?
	\item[2.] As it is, \eqref{eq:TimeAverages1} requires too much information. In fact, it requires the complete knowledge of the initial state of the system which, for a gas of $N \sim 10^{23}$ particles is much more than the bare knowledge of $P,V,N$ and $T$ suggested by experiments.
	\item[3.] Given a macroscopic observable $F_{\text{macro}}$, is it always possible to associate to it a function on the phase space $F_{\text{micro}}$? What about the contrary?
\end{itemize}
A partial answer to these questions will be given in Section \ref{sec:Ergodic}.

\section{Elements of ergodic theory and classical statistical mechanics}\label{sec:Ergodic}

In this section, we introduce some basic concepts from ergodic theory and classical statistical mechanics to address, in a mathematically well-defined framework, problems on thermalization and description of thermal equilibrium. 

\subsection{Ergodic theory}
The mathematical objects of this section are dynamical systems intended in an abstract sense. A dynamical system is a collection of four elements:
\begin{itemize}
	\item[(i)] A set $\Omega$ which is the phase space of the dynamical system and whose elements $x \in \Omega$ represent the microscopic states;
	\item[(ii)] A map $\Phi:\mathbb{M} \times \Omega \to \Omega$ (with $\mathbb{M}=\mathbb{R},\mathbb{Z}$) which is a group action of $\mathbb{M}$ on $\Omega$ and realizes the dynamics (see Definition \ref{def:OneParameterGroup} below);
	\item[(iii)] A sigma-algebra on $\Omega$ which is the collection of all measurable sets;
	\item[(iv)] A measure $\mu$ on the sigma algebra, which is normalized in the sense that
	\begin{equation}
		\mu(\Omega)\,=\,\int_{\Omega} d \mu(x) \,=\, 1 \, .
	\end{equation}
\end{itemize}
In the following, we will never mention explicitly the sigma-algebra and we will always intend a measure to be given together with a sigma-algebra.

Following the standard construction of Lebesgue integral, since $(\Omega,\mu)$ is a measure space we can construct the set of measurable function and for $1 \leq p < +\infty$ we can define the $L^p$ spaces as the (quotient) spaces of function
\begin{equation}
	L^p(\Omega, d \mu) \,:=\, \Big\{[f]\, \text{$\mu$-measurable} \, : \, \int_{\Omega} |f(x)|^p \, d \mu(x) < +\infty\Big\}
\end{equation}
with $[f]:=\{g \, \text{$\mu$-measurable} : \mu(\{g-f \neq 0 \})=0\}$. As it is customary, we will never make an explicit reference to the equivalence classes. Also, for compactness, given a measure $\mu$ and a function $f \in L^1(\Omega,d \mu)$, we will denote by
\begin{equation}
	\langle f \rangle_\mu \,:=\, \int f(x) \, d\mu(x) \, .
\end{equation}
Note that, the notation $\langle \cdot \rangle_\mu$ is reminiscent of the theory of probability: it denotes the average of $f$ with respect to the probability measure $\mu$. 

\begin{definition}\label{def:OneParameterGroup}
	The action of the one-parameter group $\mathbb{M}$ on $\Omega$ is a map $\Phi:\mathbb{M} \times \Omega \to \Omega$ such that
	\begin{itemize}
		\item[(i)] $\Phi^0 = \mathbbm{1}$;
		\item[(ii)] $\Phi^{-t} = (\Phi^t)^{-1}$ for all $t \in \mathbb{M}$;
		\item[(iii)] $\Phi^{s+t}=\Phi^s \circ \Phi^t=\Phi^t \circ \Phi^s$ for all $t,s \in \mathbb{M}$.
	\end{itemize}
\end{definition}

\begin{definition}
	A measure $\mu$ on the phase space is invariant under the action of a one-parameter group $\Phi$ if for all measurables $A$ and for all $t \in \mathbb{M}$, $\mu(\Phi^t(A))=\mu(A)$.
\end{definition}

\begin{definition}\label{def:DynSist}
	Let $\Omega$ be a set, $\Phi:\mathbb{M} \times \Omega \to \Omega$ (with $\mathbb{M}=\mathbb{Z},\mathbb{R}$) a one-parameter group on $\Omega$ and $\mu$ an invariant measure along along $\Phi$. We call $(\Omega,\Phi,\mu)$ an \emph{(abstract) dynamical system}.
\end{definition}

\begin{example} \textbf{}
	\begin{itemize}
		\item[(i)] Let $\Omega=\mathbb{T}^n=(\mathbb{R}/\mathbb{Z})^n$ be the unit $n$-dimensional torus of size one and $\omega \in \mathbb{R}^n$ be a vector. We can define the map
	\begin{equation}
		\Phi^t(\varphi)=\varphi+ \omega t \, \text{mod} \, \mathbb{Z}^n \, .
	\end{equation} Then, $(\mathbb{T}^n,\Phi, d\varphi_1 \cdots d \varphi_n)$ is a dynamical system.
		\item[(ii)] Let $\Omega=\mathbb{T}^2$ with $d \mu = d \varphi_1 d \varphi_2$ and let $A$ be the $2\times 2$ matrix
		\begin{equation}
			A\,=\, \begin{pmatrix}
				1 & 1 \\
				1 & 2
			\end{pmatrix} \, .
		\end{equation}
		Let the dynamics be defined by
		\begin{equation}
			\Phi^1(\varphi) \, = \, A \varphi \, \mathrm{mod} \, \mathbb{Z}^2 \, .
		\end{equation}
		Then $\Phi : \mathbb{Z} \times \Omega \to \Omega$ is defined as $\Phi(j,\varphi)=\Phi^j(\varphi):=\Phi^1(\Phi^{j-1}(\varphi))$. This is a dynamical system known as \emph{Arnol'd's cat}.
	\end{itemize}
\end{example}

We will now state all the results in the time-continuous case, $\mathbb{M}=\mathbb{R}$. Similar statements and properties can be stated in the discrete case.

\begin{theorem}\label{thm:ExistenceTimeAverage}
	Let $(\Omega,\Phi,\mu)$ be a dynamical system and let $f \in L^1(\Omega, d \mu)$, then 
	\begin{equation}\label{eq:DefinitionTimeAverageGeneral}
		\bar{f}(x) \,:=\, \lim_{t \to +\infty} \frac{1}{t} \int_0^t f(\Phi^s(x)) \, ds \, ,
	\end{equation}
	exists for $\mu$-a.e.~$x \in \Omega$.
\end{theorem}

\begin{corollary}\label{cor:TimeAverage}
	If $x \in \Omega$ is s.t.~$\bar{f}(x)$ exists, then $\forall t \in \mathbb{R}$, $\bar{f}(\Phi^t(x))$ exists and $\bar{f}(x)=\bar{f}(\Phi^t(x))$.
\end{corollary}

Theorem \ref{thm:ExistenceTimeAverage} and Corollary \ref{cor:TimeAverage} constitute the so-called \emph{Birkhoff-Kinchin ergodic theorem}, see Theorem 1 in \cite{Cornfeld1982} and Section ``Pointwise convergence'' in \cite{Halmos2018-jl}. 

\begin{corollary}
	Let $f \in L^1(\Omega,d\mu)$, then $\langle \bar{f} \rangle_\mu = \langle f \rangle_\mu$.
\end{corollary}
\begin{proof}
Since $f \in L^1(\Omega, d \mu)$, and the measure $\mu$ is invariant, we have
\[
	\langle f \rangle_\mu= \int_{\Omega} f(x) \, d \mu(x) = \int_{\Omega} f(\Phi^t(x)) \, d \mu(x) \, .
\]
Taking now the time-average on both sides, using the fact that by Theorem~\ref{thm:ExistenceTimeAverage}, the average $\bar{f}(x)$ exists for $\mu$-a.e.~$x \in \Omega$, we complete the proof.
\end{proof}

\begin{definition}\label{def:FrequencyOfVisit}
	Let $B \subset \Omega$ be a $\mu$-measurable set. Then 
	\begin{equation}
		\chi_B(x) \,:=\, \left\{\begin{array}{lcl}
			1 & \qquad & x \in B \\
			0 & \qquad & x \notin B
		\end{array}\right.
	\end{equation}
	is the \emph{characteristic function} of $B$. The \emph{frequency of visit} of $B$ starting at $x \in \Omega$ is
	\begin{equation}
		\nu_B(x) \,:=\, \overline{\chi_B}(x) \, ,
	\end{equation}
	where the bar denotes the time average in \eqref{eq:DefinitionTimeAverageGeneral}.
\end{definition}

\begin{theorem}\label{thm:ErgodicEquivalence}
	Let $(\Omega,\Phi,\mu)$ be a dynamical system. The following are equivalent:
	\begin{itemize}
		\item[(a)] For any $f \in L^1(\Omega,d\mu)$, $\overline{f}(x)= \langle f \rangle_\mu$ for $\mu$-a.e.~$x \in \Omega$;
		\item[(b)] For any $\mu$-measurable $B$, $\nu_B(x)=\mu(B)$ for $\mu$-a.e.~$x \in \Omega$;
		\item[(c)] Let $B$ be $\mu$-measurable, if $\Phi^t(B)=B$, then $\mu(B)=0$ or $\mu(B)=1$. 
		\item[(d)] Let $f \in L^1(\Omega,d \mu)$, if for all $t \in \mathbb{R}$ and for $\mu$-a.e.~$x \in \Omega$, $f(\Phi^t(x))=f(x)$ then $f=\mathrm{const}$ for $\mu$-a.e.~$x \in \Omega$.
	\end{itemize}
\end{theorem}
\begin{proof}
	$\mathrm{(a)} \Rightarrow \mathrm{(b)}$. If (a) holds for any $f \in L^1(\Omega, d\mu)$, then since $\chi_B \in L^1(\Omega, d \mu)$, (b) follows from Definition~\ref{def:FrequencyOfVisit}.
	
	$\mathrm{(b)} \Rightarrow \mathrm{(c)}$. From $\Phi^t(B)=B$ it follows that $\overline{\chi_B}(x)=\chi_B(x)$. Therefore, from (b) it follows that for $\mu$-a.e.~$x \in \Omega$,
	\[
		\mu(B)=\nu_B(x) =\overline{\chi_B}(x)=\chi_B(x) \, .
	\]
	Since by Definition~\ref{def:FrequencyOfVisit} $\chi_B(x) \in \{0,1\}$, (c) follows.
	
	$\mathrm{(c)} \Rightarrow \mathrm{(a)}$. Equivalently, we prove $\mathrm{not(a)} \Rightarrow \mathrm{not(c)}$. Suppose there exists $f \in L^1(\Omega, d \mu)$ and a measurable $C \subset \Omega$ such that, for $\mu$-a.e.~$x \in C$, $\overline{f}(x) \neq \langle f \rangle_\mu$. Then, at least one among the sets
	\[
		C_>\,:=\,\{x \in \Omega \, : \, \overline{f}(x) > \langle f \rangle_\mu\} \, , \qquad C_< \,:=\, \{x \in \Omega \, : \, \overline{f}(x) < \langle f \rangle_\mu\}
	\]
	has positive measure. Suppose that $\mu(C_>)>0$. Then, $0< \mu(C_>)<1$. Indeed, $\mu(C_>)=1$ would yield $\langle \overline{f} \rangle_\mu > \langle f \rangle_\mu$, contradicting Corollary~\ref{cor:TimeAverage}, thus $\mu(C_>) \neq 1$. We now prove that $C_>$ is an invariant set.
	\[
		\begin{split}
			\Phi^t(C_>) \,&=\, \{x \in \Omega \, : \, x=\Phi^t(y) \,, \, y \in C_> \} \\
			&=\,\{x \in \Omega \, : \, y=\Phi^{-t}(x) \, , \, \overline{f}(y) > \langle f \rangle_\mu \} \\
			&=\,\{ x \in \Omega \, : \, \overline{f}(\Phi^{-t}(x)) > \langle f \rangle_\mu \} \\
			&=\, \{x \in \Omega \, : \, \overline{f}(x) > \langle f \rangle_\mu \} \, = \, C_> \, ,
		\end{split}
	\]
	where in the last step we used Corollary~\ref{cor:TimeAverage}. The same argument works supposing instead that $\mu(C_<)>0$. Therefore $\mathrm{not(a)} \Rightarrow \mathrm{not(c)}$.
	
	$\mathrm{(c)} \Rightarrow \mathrm{(d)}$. Equivalently, we prove $\mathrm{not(d)} \Rightarrow  \mathrm{not(c)}$. Suppose $f \in L^1(\Omega, d \mu)$ invariant and $f \neq \text{const}$ a.e.~in $\Omega$. Let us define,
	\[
		C_>\,:=\,\{ x \in \Omega \, : \, f(x) > \langle f \rangle_\mu\} \, , \qquad C_<\,:= \, \{ x \in \Omega \, : \, f(x) < \langle f \rangle_\mu \} \, .
	\]
	Since $f$ is invariant, both sets are invariant. Since $\mathrm{not(d)}$, at least one among $C_>$ and $C_<$ has positive measure. Let us suppose that $\mu(C_>)>0$, then $\mu(C_>)<1$, otherwise Corollary~\ref{cor:TimeAverage} would be contradicted. Therefore, $C_>$ is a non-trivial invariant set. If, instead, we assume that $C_<$ has positive measure, the argument is the same. This completes the proof that $\mathrm{not(d)} \Rightarrow \mathrm{not(c)}$ and therefore $\mathrm{(c)} \Rightarrow \mathrm{(d)}$.
	
	$\mathrm{(d)} \Rightarrow \mathrm{(c)}$. Equivalently, we prove $\mathrm{not(c)} \Rightarrow \mathrm{not(d)}$. If there exists a non-trivial set $B$, then $\chi_B \in L^1(\Omega,d\mu)$ is invariant and is not a.e.~constant. This completes the proof.
\end{proof}

\begin{definition}
	A dynamical system is said to be \emph{ergodic} if any (and, therefore, all) of (a)-(d) in Theorem \ref{thm:ErgodicEquivalence} holds.
\end{definition}

\begin{remark}\label{rem:SimpleFunctions}
	Note that, requirements (a) and (d) in Theorem \ref{thm:ErgodicEquivalence} can be weakened. Indeed, item (c) states that we can limit ourselves to requiring (a) and (d) for simple functions or for any $L^p$ space. 
\end{remark}

Characterizations (a)--(d) in Theorem \ref{thm:ErgodicEquivalence} have a meaning in terms of dynamical constraints. Item (a) is the equality between time and space average, which is the definition of ergodicity in the spirit of Boltzmann. Item (b) establishes the equality between the measure of a set and the frequency of visit. It can be paraphrased by saying that an ergodic system explores uniformly (with respect to $\mu$) the whole set $\Omega$. Item (c) states \emph{metric indecomposability} of $\Omega$, that is: we cannot split $\Omega$ into two (or more) parts that are invariant for the dynamics and with positive measure: that is, in other words, apart from a zero-measure set, all the points belong to the same indecomposable invariant set. Item (d) states that an ergodic system does not have any $L^1$-constant of motion. This latter point is quite important since it relates the lack of ergodicity and the presence of integrals of motion. 

\begin{example}\label{ex:TranslationTorus}
	The dynamical system $(\mathbb{T}^1, \Phi, d\varphi)$ with $t \in \mathbb{Z}$ and
	\begin{equation}\Phi^j(\varphi)=\varphi+ \omega j \, \text{mod} \, \mathbb{Z}
	\end{equation}
 with $\omega \in \mathbb{R}$, is ergodic if and only if $\omega \in \mathbb{R} \setminus \mathbb{Q}$.
	
	If $\omega \in \mathbb{Q}$, then there are two integers $p,q \in \mathbb{Z}$ so that $\omega=p/q$. Then, $f(\varphi)=\cos( 2 \pi q \varphi)$ is a nontrivial integral of motion, therefore the system is not ergodic.
	
	If $\omega \in \mathbb{R}\setminus \mathbb{Q}$ then, we show that there is no $L^2(\mathbb{T}^1, d \varphi)$ integral of motion. This is sufficient because of Remark \ref{rem:SimpleFunctions} and Theorem \ref{thm:ErgodicEquivalence} (d). If $f$ is such a constant of motion, it admits a Fourier series
	\begin{equation}
		f(\varphi)=\sum_{k \in \mathbb{Z}} \hat{f}_k e^{ i 2 \pi k \varphi} \, ,
	\end{equation}
	where the series converge in $\ell^2(\mathbb{Z})$. Then, we have
	\begin{equation}
		f(\Phi^j(\varphi)) \,=\, \sum_{k \in \mathbb{Z}} \hat{f}_k e^{i 2 \pi k (\varphi + \omega)} \, ,
	\end{equation}
	and this implies that if $f(\Phi^j(\varphi))=f(\varphi)$, we have
	\begin{equation}
		\hat{f}_k(1-e^{2 \pi i k \omega}) \, =\, 0 \qquad \forall j \in \mathbb{Z} \, , 
	\end{equation}
	that is $\hat{f}_k=0$ for all $k \in \mathbb{Z}$ and, by Plancherel, $f=0$.
\end{example}

\subsection{Mixing and approach to thermal equilibrium}\label{subsec:Mixing}
Even if ergodicity is a good candidate as microscopic requirement to justify equilibrium statistical mechanics, it is still not enough to guarantee an approach to thermal equilibrium in an irreversible dynamical way (in a sense to be defined below).

We can define the macroscopic state of a system as an absolutely continuous measure with respect to the invariant measure $\mu$. That is, by Radon-Nikodym theorem (see Theorem 2.1 of \cite{Berezansky1996-xg}), there is a positive $L^1(\Omega,d \mu)$ function $\rho(x)$ (the \emph{density} of the measure) by which one can compute the value of macroscopic observables as
\begin{equation}
	\langle f \rangle_{\rho,\mu} \,:=\, \int_{\Omega} f(x) \rho(x) \, d \mu(x) \, .
\end{equation}
From this point of view, the thermal equilibrium corresponds to the macroscopic state of constant density $\rho(x)=1$.

One of the consequences of ergodicity, is the uniqueness of the equilibrium state in the following, very delicate sense. Fixing $\mu$, $\rho=1$ is the unique function in $L^1(\Omega,d\mu)$ invariant along the action of the one-parameter group $\Phi$. Still, fixing a sigma-algebra, it is possible to have different invariant measures.

\begin{xca}
	Let $\Omega \subset \mathbb{R}^n$ be a bounded and connected set and $\Phi:\mathbb{R} \times \Omega \to \Omega$ be the action of $\mathbb{R}$ on $\Omega$. Suppose that there exists $x_0 \in \Omega$ such that $\Phi^t(x_0)=x_0$ for all $t \in \mathbb{R}$. Consider $\mu$ the (normalized) Lebesgue measure and suppose that $(\Omega,\Phi,\mu)$ is ergodic. 
\begin{itemize}	
	\item[(i)] Find a normalized measure $\nu \neq \mu$ such that $(\Omega,\Phi,\nu)$ is ergodic.
	\item[(ii)] Show that, for any $a \in (0,1)$, the system $(\Omega,\Phi, (1-a)\mu + a \nu)$ is not ergodic.
\end{itemize}
\end{xca}

The dynamics generated by $\Phi$, induces a dynamics on the densities:
\begin{equation}
	\rho(t,x) \, := \, \rho(\Phi^{-t}(x)) \, .
\end{equation}
A meaningful approach to equilibrium, would be to require, in some sense, that $\rho(t,x) \underset{t \to +\infty}{\longrightarrow} 1$.

\begin{example}
	Let us consider the dynamical system in Example \ref{ex:TranslationTorus} with $\omega \in \mathbb{R} \setminus \mathbb{Q}$. It has been shown that it is ergodic. We want now to investigate if a property resembling $\rho(t,x) \to 1$ is possible. Indeed, starting from 
	\begin{equation}\label{eq:RhoExample}
		\rho(\varphi)=\left\{ \begin{array}{lll} 2 & \qquad & 0 \leq \varphi \leq \frac{1}{2} \, , \\
		0 & & \frac{1}{2} < \varphi < 1 \, .
		\end{array} \right.
	\end{equation}
	Then, $\rho(t,\varphi)$ is a rigid rotation of angle $\omega$ and has no chance to converge to $1$ in any strong sense (see Figure~\ref{fig:EvolutionDensity}).
	
	\begin{figure}[h]
		\begin{center}
			\begin{tikzpicture}[scale=0.9]

				\draw[draw=gray,fill=gray] (0,0) rectangle ++(1.5,1);		

				\draw[-] (0,0) -- (3,0);
				\draw[-] (0,-0.1) -- (0,0.1);
				\draw[-] (3,-0.1) -- (3,0.1);
				
				\node at (0,-0.3) {$0$};
				\node at (3,-0.3) {$1$};

			\draw[draw=gray,fill=gray] (5.3,0) rectangle ++(1.5,1);	

				\draw[->,thick] (3.5,0.5) -- (4.5,0.5);
				\node at (4,0.8) {$\Phi$};

				\draw[-] (5,0) -- (8,0);
				\draw[-] (5,-0.1) -- (5,0.1);
				\draw[-] (8,-0.1) -- (8,0.1);
				
				\node at (5,-0.3) {$0$};
				\node at (8,-0.3) {$1$};
				
			\draw[draw=gray,fill=gray] (10.6,0) rectangle ++(1.5,1);		
				
				\draw[->,thick] (8.5,0.5) -- (9.5,0.5);
				\node at (9,0.8) {$\Phi$};		
				
				\draw[-] (10,0) -- (13,0);
				\draw[-] (10,-0.1) -- (10,0.1);
				\draw[-] (13,-0.1) -- (13,0.1);
				
				\node at (10,-0.3) {$0$};
				\node at (13,-0.3) {$1$};				
				
			\end{tikzpicture}
		\end{center}
		\caption{Schematic dynamics of the density $\rho$ in Eq.~\eqref{eq:RhoExample}.}\label{fig:EvolutionDensity}
	\end{figure}
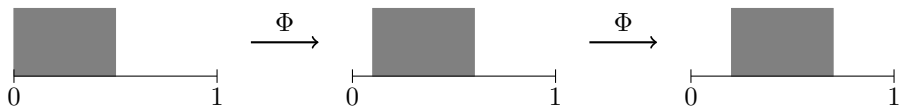
\end{example}

To take into account of the loss of memory of the initial datum, it is possible to introduce a property which is stronger than ergodicity, called \emph{mixing}. 

\begin{proposition}\label{prop:MixingEquivalence}
	Let $(\Omega,\Phi,\mu)$ be a dynamical system. The following two conditions are equivalent:
	\begin{itemize}
		\item[(a)] for any measurables $A,B \subset \Omega$,
		\begin{equation}
			\lim_{t \to +\infty} \mu(\Phi^{-t}(A) \cap B) = \mu(A) \mu(B) \, ;
		\end{equation}
		\item[(b)] for any pair of functions $f,g \in L^2(\Omega,d \mu)$,
		\begin{equation}
			\lim_{t \to +\infty} \langle (f \circ \Phi^t) g \rangle_\mu = \langle f \rangle_\mu \langle g \rangle_\mu \, .
		\end{equation}
	\end{itemize}
\end{proposition}
\begin{proof}
	(a) is a special case of (b), corresponding to the case $f=\chi_A$ and $g=\chi_B$, therefore (b) implies (a). 
	
	For the opposite inclusion, (b) implies that (a) holds for simple functions. Since simple functions are dense in $L^2(\Omega,d\mu)$, one concludes.
\end{proof}

\begin{definition}
	A dynamical system is said to be \emph{mixing} if either (a) or (b) (and then, both) in Proposition \ref{prop:MixingEquivalence} holds.
\end{definition}

We now show that mixing is stronger than ergodicity.

\begin{lemma}
	If $(\Omega,\Phi,\mu)$ is mixing, then it is ergodic.
\end{lemma}
\begin{proof}
	We show property (c) of Theorem \ref{thm:ErgodicEquivalence}. Let $A$ be a measurable invariant set. Then,
	\[
		\mu(A)=\mu(A \cap A) = \mu(\Phi^{-t}(A) \cap A)=\mu^2(A) \, ,
	\]
	the latter means that $\mu(A)=0$ or $\mu(A)=1$, thus proving (c) of Theorem \ref{thm:ErgodicEquivalence}.
\end{proof}

\begin{xca}
	Show rigorously that irrational translations on the one-dimensional torus are not mixing. Indeed, consider $\Omega=\mathbb{T}^1$, $d \mu = 
	d \varphi$, $\Phi^j(\varphi)=\varphi+ \omega j \, \mathrm{mod} \, \mathbb{Z}$. Show that if $\omega \in \mathbb{R} \setminus \mathbb{Q}$, and if
\begin{equation*}
	u_k(\varphi) \,=\, e^{i 2 \pi k \varphi}
\end{equation*}
then 
\begin{equation*}
	\lim_{j \to +\infty} \langle (u_k \circ \Phi^j) u_{-k} \rangle_\mu \, = \, \text{does not exist} \, . 
\end{equation*}
\end{xca}

The two (equivalent) conditions characterizing mixing in Proposition \ref{prop:MixingEquivalence} have a very natural interpretation. Item (a) encodes the idea that when $t \to +\infty$ the time evolution diffuses uniformly in $\Omega$ all the points in $A$. This is the intuitive idea of \emph{irreversible dynamics}: if the particles of a gas are originally placed in the corner of a room, they will eventually diffuse uniformly in the room. Item (b) is a property of \emph{decorrelation} of observables and encodes the idea of \emph{loss of memory} along time evolution.

In particular, from property (b) we can easily read out the approach to thermal equilibrium for observables:
\begin{equation}
	\begin{split}
		\langle f \rangle_{\rho(t),\mu}&=\int_{\Omega} f(x) \rho(t,x) \, d \mu(x) = \int_{\Omega} f(\Phi^t(x)) \rho(x) \, d \mu(x) \\
		&\underset{t \to +\infty}{\longrightarrow} \langle f \rangle_\mu \langle \rho \rangle_\mu = \langle f \rangle_\mu \, .
		\end{split}
\end{equation}
This latter consideration suggests that mixing implies approach to equilibrium $\rho(t,x) \to 1$ in a \emph{weak sense}.

\subsection{General comments and further results} Before moving to the analysis of Hamiltonian systems,
 we shall state one last theorem that is at the basis of the formulation of Zermelo's paradox (see Paradox \ref{par:Zermelo}). 

\begin{theorem}[Poincar\'e recurrence] Let $(\Omega,\Phi,\mu)$ be a dynamical system with $\Omega$ bounded. Let $A \subset \Omega$ be measurable and $\mu(A)>0$. Then, for $\mu$-a.e.~$x \in A$ and for any $T>0$, there exists $t > T$ so that $\Phi^t(x) \in A$.
\end{theorem}

In the previous paragraphs, we saw that a reasonable definition of macroscopic state, is given by positive $\rho \in L^1(\Omega, d \mu)$ with $\langle \rho \rangle_\mu=1$. Accepting this definition and further requiring the microscopic dynamics to be mixing, implies that the macroscopic dynamics is irreversible while the microscopic one is recurrent. This picture is not contradictory because the recurrence times \emph{depend} on the initial datum and thus, close initial data, can have different recurrence times. Recurrence of macroscopic state would be possible in the case of coherent recurrence of the microscopic states that compose it. 

Proving that a physically meaningful system is ergodic or mixing is a quite hard task, especially for systems made by large number of particles. To the best of my knowledge, it has only recently been possible to prove that a hard-sphere gas is ergodic (see \cite{Simnyi2004}).

\subsection{Basic properties of Hamiltonian systems}
In classical mechanics, the state of a system of particles is defined once all positions and momenta are known. Then, time evolution is generated by Hamilton equations. The space of all possible states is often referred to as the \emph{phase space} of the system. We denote by $\Gamma \subset \mathbb{R}^{2N}$ the phase space of the system and we denote by $(q,p) \in \Gamma$ the state of a system.

The Hamiltonian is a real-valued function on the phase space of the system. $H(q,p)$ has the meaning of the \emph{energy} of the state $(q,p) \in \Gamma$.

Denoting with a dot the time derivative, Hamilton equations are
\begin{equation}\label{eq:HamiltonEquationsGeneral}
	\begin{split}
			\dot{q}(t) \, & = \, \frac{\partial H}{\partial p}(q(t),p(t)) \, , \\	
		\dot{p}(t) \, &= \, - \frac{\partial H}{\partial q}(q(t),p(t)) \, .
	\end{split}
\end{equation}
We assume that the solution of these equations exists and is unique for all $(q,p) \in \Gamma$. The flow associated to the Hamiltonian $H$, that is the map that associates to an initial datum $(q_0,p_0) \in \Gamma$ the solution of the Hamilton equations at time $t$, is denoted by $\Phi_H^t(q_0,p_0)=(q(t),p(t))$. 

By Liouville's theorem (see Section 16 in \cite{Arnold1989}), the Lebesgue measure is invariant under $\Phi_H$ and therefore, the triple $(\Gamma, \Phi_H, dq dp)$ is a dynamical system in the sense of Definition \ref{def:DynSist}. Here, and throughout all these notes, we will denote by $dqdp=\prod_{j=1}^N dq_j dp_j$. 

Hamilton equations can be written in a more concise form by introducing the Poisson tensor
\begin{equation}
	J\,=\, \begin{pmatrix}
		\mathbb{O} & \mathbbm{1} \\
		-\mathbbm{1} & \mathbb{O}
	\end{pmatrix} \, ,
\end{equation}
and, denoting by $z=(q,p) \in \Gamma$, the system \eqref{eq:HamiltonEquationsGeneral} is equivalent to
\begin{equation}
	\dot{z} \,=\, J \nabla_z H(z(t)) \, .
\end{equation}

\begin{definition}
	Given two smooth functions $F,G: \Gamma \to \mathbb{R}$, their Poisson bracket is the smooth function defined as
	\begin{equation}
		\{F,G\}(z) \,:=\, \nabla_z F(z) \cdot J \nabla_z G(z) \, = \, \sum_{j,\ell=1}^{2N} \frac{\partial F(z)}{\partial z_j} J_{j,\ell} \frac{\partial G(z)}{\partial z_\ell} \, .
	\end{equation}
	where $\cdot$ denotes the standard Euclidean scalar product.
\end{definition}

\begin{lemma}
	Given $F_1,F_2,F_3 \in C^\infty(\Gamma)$ and $a,b \in \mathbb{R}$, the Poisson brackets satisfy the following properties:
	\begin{itemize}
		\item[(a)] $\{aF_1+bF_2,F_3\}=a\{F_1,F_3\}+b\{F_2,F_3\}$ (left linearity);
		\item[(b)] $\{F_1,F_2\}=-\{F_2,F_1\}$ (antisymmetry);
		\item[(c)] $\{F_1,F_2F_3\}=\{F_1,F_2\}F_3+F_2\{F_1,F_3\}$ (Leibnitz rule);
		\item[(d)] $\{\{F_1,F_2\},F_3\}+\{\{F_2,F_3\},F_1\}+\{\{F_3,F_1\},F_2\}=0$ (Jacobi identity).
	\end{itemize}
\end{lemma}

Poisson brackets are useful for several reasons. First, an equivalent way of writing Hamilton equations is
\begin{equation}
	\dot{z}(t)\,=\,\{z,H\}\,=\, J \nabla_z H \, ;
\end{equation}
second, the total derivative of a function along the flow of $H$ is
\begin{equation}
	\dot{F}(t) \,=\, \nabla_z F(z(t)) \cdot \dot{z}(t) \,=\, \nabla_zF(z(t)) \cdot J \nabla_z H(z(t)) \, = \, \{F,H\}(z(t)) \, .
\end{equation}
This last comment justifies the following definition of integral of motion.

\begin{definition}
	$F \in C^\infty(\Gamma)$ is said to be an integral of motion for the Hamiltonian system of Hamiltonian $H$ if $\{F,H\}=0$.
\end{definition}
It follows then from antisymmetry of the Poisson brackets that $\{H,H\}=0$ and therefore $H$ is an integral of motion. Due to the presence of an integral of motion, by Theorem \ref{thm:ErgodicEquivalence}-(c), a Hamiltonian system cannot be ergodic on the whole phase space. Nevertheless, if $E \in \mathbb{R}$ is a regular value of $H$, so that the set $\{H(q,p)=E\}$ is a manifold, we can still hope to prove ergodicity on that set. Indeed:
\begin{itemize}
	\item The set $\{H(q,p)=E\}$ is left invariant by the dynamics and under suitable hypotheses on $H$ is compact.
	\item The restriction of the Lebesgue measure on the set $\{H(q,p)=E\}$ is left invariant by the Hamiltonian flow as a consequence of Liouville theorem (see Section 16 in \cite{Arnold1989}). 
\end{itemize}
Note that, if two functions are such that $\{F,G\}=0$, then their Hamiltonian vector fields commute, that is $[J \nabla F, J \nabla G]=0$, and the associated flows commute as well. 

It is customary to refer to the Poisson brackets $\{z_j,z_\ell\}$ as to the \emph{fundamental Poisson brackets}. In general, we have
\begin{equation}
	\{q_j,q_\ell\}\,=\, 0 \,=\, \{p_j,p_\ell\} \, , \qquad \{q_j,p_\ell\}\,=\, \delta_{j,\ell} \, .
\end{equation}
The reason for this definition is that, for smooth functions, one has
\begin{equation}
	\{F,G\} \, = \, \sum_{j,\ell=1}^{2n} \frac{\partial F}{\partial z_j} \frac{\partial G}{\partial z_\ell} \{z_j, z_\ell\} \, ,
\end{equation}
and therefore the Poisson brackets among every pair of functions can be reconstructed once the fundamental ones are defined.

Last, from Poisson brackets one can introduce the concept of Lie derivative.
\begin{definition}\label{def:PoissonBracket}
Given two functions $F,G \in C^\infty(\Gamma)$, the Lie derivative of $F$ with respect to $G$ is the smooth function $\mathcal{L}_G(F)$ defined as
\begin{equation}
	\mathcal{L}_G(F) \, :=\, \{F,G\} \, .
\end{equation}
\end{definition}
Within this last definition, we can write the flow of the Hamiltonian system as
\begin{equation}\label{eq:FlowLieDer}
	\Phi_H^t(z) \, =\, e^{t \mathcal{L}_H}z \, .
\end{equation}

\subsection{Microcanonical and canonical ensembles}
In this section, for concreteness, we consider Hamiltonians of the form
\begin{equation}
	H(q,p; N, V) \, = \, \sum_{j=1}^N \frac{p_j^2}{2} + \sum_{j=1}^N \Phi_V(q_j) + \frac{1}{2} \sum_{j,\ell=1}^N \Phi_{\text{int}}(q_j-q_\ell)
\end{equation}
where 
\begin{itemize}
	\item $(q,p) \in \Gamma \subset \mathbb{R}^{6N}$ are points on the phase space of the system;
	\item $\Phi_V$ denotes a \emph{confining potential} that can be thought of as a ``smoothing'' of the potential that is $+\infty$ outside the box of volume $V$ and $0$ inside;
	\item $\Phi_{\text{int}}$ is an inter-particles interacting potential that is assumed to be bounded from below.
\end{itemize}
Under previous assumptions, for every value $E \geq 0$, the set $\{H(q,p;N,V)=E\}\subset \Gamma \subset \mathbb{R}^{6N}$ is compact. 

\begin{definition}
	If $E$ is a non-singular value of $H(p,q;N,V)$, the \emph{microcanonical measure} is defined as
	\begin{equation}
		d \mu_{\text{mc}}(E,V,N)\,:=\, \frac{1}{Z_{\text{mc}}(E,V,N)} \delta(H(p,q;V,N)-E) \, dp dq \, ,
	\end{equation}
	where $Z_{\text{mc}}(E,V,N)$ is the $6N-1$ Lebesgue measure of $\{H(p,q;V,N)=E\}$.
\end{definition}

\begin{lemma}
	The microcanonical measure $\mu_{\text{mc}}(E,V,N)$ is invariant under the flow of $H(p,q;V,N)$.
\end{lemma}
	
As discussed in Subsection~\ref{subsec:Mixing}, invariant measures on the phase space are good definitions of macroscopic states corresponding to thermal equilibrium. The macroscopic state defined from the microcanonical ensemble is referred to as the \emph{microcanonical ensemble}. 

In principle, one is free to choose his own invariant measure, does it describe in \emph{any} sense a theory compatible with thermodynamics? The question is subtle and is solved by requiring that observables computed with respect to the invariant measure describes a theory which is compatible with the principles of thermodynamics (the so-called orthodicity). This delicate point is discussed exhaustively in Chapter 2 of \cite{Gallavotti1999}. For the purposes of these notes, it is sufficient to state, as a definition, the \emph{entropy} in the microcanonical ensemble:
\begin{definition}
	The \emph{entropy} of a system at thermal equilibrium whose macroscopic state is given by the microcanonical measure is given by
	\begin{equation}
		S(E,V,N)\,=\, k_B \ln (Z_{\mathrm{mc}}(E,V,N)) \, .
	\end{equation}
	The temperature is defined as
	\begin{equation}\label{eq:TemperaturaMicro}
		\frac{1}{T(E,V,N)} \, :=\, \frac{\partial S}{\partial E}(E,V,N) \, ,
	\end{equation}
	and the pressure as 
	\begin{equation}
		P(E,V,N)\,=\,T(E,V,N) \frac{\partial S}{\partial V}(E,V,N) \, .
	\end{equation}
\end{definition}

Then, to check whether thermodynamics defined by $\mu$ is compatible with standard principle of thermodynamics we shall define the internal energy of the gas
\begin{definition}
	Let $\mu(X_1,\dots,X_n)$ be an invariant measure depending on the parameters $X_1,\dots,X_n$. The \emph{internal energy} of the gas is defined as
	\begin{equation}
		U(X_1,\dots,X_n)\,:=\, \langle H \rangle_{\mu(X_1,\cdots,X_n)} \, ,
	\end{equation}
	where the dependence of the Hamiltonian $H$ on the parameters has not been written explicitly.
\end{definition}

For example, in the case of the microcanonical measure, we have $U(E,V,N)=E$.

\begin{definition}
The \emph{thermodynamic limit} of $F(N,V,E)\,:=\, \langle f \rangle_{\mu_{\text{mc}}(N,V,E)}$ is defined as
\begin{equation}
	\mathsf{f}(v,\epsilon) \,:=\, \lim_{\substack{N\to +\infty \\ E/N=\epsilon \\ V/N=v}} F(N,V,E) \, .
\end{equation}
\end{definition}

Since for Hamiltonian systems, the microcanonical measure is the invariant measure describing the thermodynamic equilibrium of an isolated system. In experiments, this is not always the case: it is often easier to keep the total temperature of the system fixed, allowing exchanges of energy between the system and the environment outside. From a statistical point of view, we assume that the equilibrium state of a system which is kept at a fixed temperature $T=\frac{1}{k_B \beta}>0$ is described by the \emph{canonical ensemble}.

\begin{definition}\label{def:CanonicalPF}
	Let $\beta>0$ and $H(q,p; N,V)>-C N$ for some $C>0$. Then, the \emph{canonical partition function} is defined as
	\begin{equation}
		Z_{\text{can}}(\beta,V,N)\,:=\, \int_{\mathbb{R}^{6N}} e^{-\beta H(q,p;V,N)} \, dq dp
	\end{equation}
	and the \emph{canonical measure} is defined as
	\begin{equation}
		d \mu_{\text{can}}(\beta,V,N)\,:=\, \frac{1}{Z_{\mathrm{can}}(\beta,V,N)}e^{-\beta H(q,p;V,N)} dq dp \, .
	\end{equation}
\end{definition}

\begin{definition}
	The \emph{Helmholtz free energy} of a system at thermal equilibrium whose macroscopic state is given by the canonical measure at (inverse) temperature $\beta$ is given by
	\begin{equation}
		A(\beta,V,N)\,=\,-\frac{1}{\beta} \ln (Z_{\mathrm{can}}(\beta,V,N)) \, .
	\end{equation}
	The entropy is defined as
	\begin{equation}
		S(\beta,V,N)\,:=\, k_B \beta^2 \frac{\partial A}{\partial \beta}(\beta,V,N) \, ,
	\end{equation}
	and the pressure as
	\begin{equation}
		P(\beta,V,N) \, :=\, -\frac{\partial A}{\partial V}(\beta,V,N) \, .
	\end{equation}
\end{definition}

When $\Phi_{\mathrm{int}}$ is a short-range interaction, i.e.~it decays quickly enough when the positions of the two particles are far apart, then the canonical and the microcanonical ensemble are \emph{equivalent}. This means that all the thermodynamic quantities have the same values if computed \emph{in the thermodynamic limit} when in the limiting procedure the temperature of the canonical ensemble is $\beta=\frac{1}{k_B T}$ with $T$ given by \eqref{eq:TemperaturaMicro} and the energy of the microcanonical ensemble is $E=\langle H \rangle_{\mu_{\text{can}}(\beta,V,N)}$.

It can be observed that ensemble equivalence can be required for the values of observables only. Indeed, if one considers $\mu$ as a probability distribution on $\Gamma$, then the requirement of equalities is limited to the \emph{averages} and cannot be extended to higher order momenta. Indeed, it is immediate to observe that
\begin{equation}
	\langle (H-\langle H \rangle_{\mu_{\text{mc}}(E,V,N)})^2 \rangle_{\mu_{\mathrm{mc}}(E,V,N)}  \, \neq \, \langle (H-\langle H \rangle_{\mu_{\text{can}}(\beta,V,N)})^2 \rangle_{\mu_{\mathrm{can}}(\beta,V,N)} \, .
\end{equation}

A detailed discussion on the topics presented in this chapter, together with the missing proofs, can be found in Chapter 2 of \cite{Gallavotti1999}. We can now recap the content of this Section.
\begin{itemize}
	\item It is reasonable to define macroscopic quantities as time averages of functions on the phase space.
	\item If the system is ergodic, it is possible to compute time averages without knowing many details of the dynamics: computation of time averages is reduced to computation of averages over certain measures on the phase space.
	\item In the thermodynamic limit, under certain assumptions, one can choose the microcanonical or the canonical measures to describe the thermodynamics of the system.
	\item We can interpret the invariant measures as the \emph{macroscopic equilibrium states} corresponding to the thermodynamic equilibrium of the physical system which is microscopically described by $H(p,q;V,N)$.
\end{itemize}

\section{Elements of Hamiltonian perturbation theory}

In this section we revisit some general facts of Hamiltonian systems and Hamiltonian perturbation theory in both its finite and in its infinite dimensional settings. There are many references for these topics and, in these notes, I closely follow the monograph \cite{Giorgilli2022} for the finite-dimensional case and the recent review \cite{Gallone2022} for the infinite-dimensional case. One of the books that where the reader can find a beautiful presentation of perturbation theory is \cite{Arnold1988}.

\subsection{Canonical transformations}
Transformation of coordinates are one of the most powerful tools to analyze dynamical systems: properties of several physical systems become clearer when the system is described with a well-chosen set of coordinates. Among the examples, the Keplerian motion or the harmonic oscillator.

In the context of Hamiltonian systems, a special role is played by the change of coordinates that maintain the structure of Hamilton equations, those are called \emph{canonical transformations}.

\begin{definition}\label{def:CanonicalTransf}
	A diffeomorphism $\Psi : \tilde{\Gamma} \to  \Gamma$ is said to be a \emph{canonical transformation} if for any $F,G  \in C^\infty(\Gamma)$, $\{F \circ \Psi, G \circ \Psi\}=\{F,G\} \circ \Psi$. 
\end{definition}

We recall that, given a diffeomorphism $\Psi:\tilde{\Gamma} \to \Gamma$, the \emph{Jacobian matrix} is the matrix whose entries are
\begin{equation}
	\big[D\Psi(\tilde{z})\big]_{j,\ell}\,:=\, \frac{\partial \Psi_j}{\partial \tilde{z}_\ell}(\tilde{z})	\, .
\end{equation}
We state in the following proposition the important properties of canonical transformations.

\begin{lemma}[Properties of canonical transformations]\label{lem:PropertiesCanonical}
Let $\Psi$ be a diffeomorphism. The following are equivalent.
\begin{itemize}
	\item[(a)] The equations of motion of the transformed system are still Hamiltonian in the sense that if $\tilde{z} \in \tilde{\Gamma}$, 
	\begin{equation}
		\dot{\tilde{z}}(t)\,= J \nabla_{\tilde{z}} (H \circ \Psi)(\tilde{z}(t))
	\end{equation}
	\item[(b)] The transformation is symplectic, that is if $D\Psi(\tilde{z})$ is the Jacobian matrix of $\Psi$ then,
	\begin{equation}
		D\Psi^T \, J \, D \Psi\,=\,J \, .
	\end{equation}
	\item[(c)] $\Psi$ preserves the fundamental Poisson brackets, i.e.~if $(\tilde{q},\tilde{p})=\Phi(q,p)$,
	\begin{equation}
		\{\tilde{q}_j, \tilde{q}_\ell\}=0=\{\tilde{p}_j, \tilde{p}_\ell \} \, , \qquad \{\tilde{q}_j,\tilde{p}_\ell\}=\delta_{j,\ell} \, .
	\end{equation}
	\item[(d)] $\Psi$ is canonical.
\end{itemize}
\end{lemma}

For the proof, we refer to Chapter 2 in \cite{Giorgilli2022}. There, the definition of canonical transformation is different from the one given in Definition \ref{def:CanonicalTransf}, but their equivalence is stated in Proposition 2.10. Then, 
 the proof of equivalence of item (b) and (d) is the content of Proposition 2.7 of \cite{Giorgilli2022}; the equivalence between (c) and (d) is Proposition 2.12 and the equivalence between (a) and (d) is the definition of canonical transformation therein.

\begin{lemma}\label{lem:FlowCan}
	For any $t \in \mathbb{R}$, and for any $G \in C^\infty(\Gamma)$, $\Phi_G^t:\Gamma \to \tilde{\Gamma}_t=\Phi_G^t(\Gamma)$ is a canonical transformation.
\end{lemma}
The proof of this Lemma is an instructive exercise, and can be done proving the preservation of fundamental Poisson brackets (item (c) in Lemma \ref{lem:PropertiesCanonical}). Otherwise, we refer to Proposition 2.15 in \cite{Giorgilli2022}.

We shall now introduce the canonical rescalings, that is a type of very useful transformation of $\mathbb{R} \times \Gamma$ which is canonical in a broader sense. 

\begin{definition}\label{def:CanonicalRescaling}
	Let $\alpha,\beta,\gamma,\delta \in \mathbb{R} \setminus \{0\}$. A transformation 
	\begin{equation}
		(q,p,H,t) \mapsto (Q,P,K,\tau)=(\alpha q, \beta p, \gamma H, \delta t)
	\end{equation}
	is called a \emph{canonical rescaling} if $\alpha \beta=\gamma \delta$.
\end{definition}

\begin{lemma}\label{lem:CanResc}
	Let $\Psi:\Gamma \to \tilde{\Gamma}$ be a canonical rescaling. Then, if the system is Hamiltonian in $\Gamma$ with Hamiltonian $H$, it is also Hamiltonian in $\tilde{\Gamma}$ with Hamiltonian $K=\gamma H$.
\end{lemma}

\begin{example}\label{ex:CanonicalTransf} \textbf{}
	\begin{itemize}
		\item[(i)] 	Let $\alpha \in \mathbb{R} \setminus \{0\}$. The rescaling
	\begin{equation}
		(q,p) \mapsto (\tilde{q},\tilde{p}) = \Big(\frac{q}{\alpha}, \alpha p \Big) \, ,
	\end{equation}
	is a canonical transformation. If we start from the harmonic oscillator
	\begin{equation}
		H(q,p)\,=\, \frac{p^2}{2m}+ \frac{m \omega^2 q^2}{2} \, ,
	\end{equation}
	we can define $(q,p)=(\sqrt{m \omega}^{-1} \tilde{q}, \sqrt{m\omega} p)$, and the Hamiltonian becomes
	\begin{equation}
		\tilde{H}(\tilde{q},\tilde{p}) \,=\, \frac{\omega}{2}(\tilde{p}^2+\tilde{q}^2) \, .
	\end{equation}
	\item[(ii)] The change of coordinates 
	\begin{equation}
		(q,p) \,=\,(\sqrt{2 I} \cos \varphi, \sqrt{2 I} \sin \varphi) 
	\end{equation}
	is a canonical transformation from $(0,+\infty) \times \mathbb{T}^1$ to $\mathbb{R}^2$. The transformed Hamiltonian is
	\begin{equation}
		H(\varphi,I) \, = \, \omega I \, .
	\end{equation}
	\end{itemize}
\end{example}

\subsection{Liouville integrable systems and their perturbations}\label{sec:IntegrablePerturbations}
Since the (only) obstruction for a dynamical system to be ergodic is the presence of integrals of motion, it can be of interest to characterize those Hamiltonian systems which have enough integrals of motion that permits an exact solution of the equations of motion. In classical mechanics, those are the \emph{Liouville integrable systems}.

\begin{definition}
	Let $\Gamma$ be a $2n$-dimensional subset of $\mathbb{R}^{2n}$. A Hamiltonian system on $\Gamma$ is \emph{integrable in the sense of Liouville} if it admits $n$ independent integrals of motion in involution among them. That is, if there are $n$ functions $F_1,\cdots,F_n$ so that
	\begin{itemize}
		\item[(i)] $\{F_j,H\}=0$ for all $j=1,\cdots,n$ (integrals of motion);
		\item[(ii)] $\nabla_z F_j$ is not parallel to $\nabla_z F_\ell$ for all $j \neq \ell=1,\cdots, n$;
		\item[(iii)] $\{F_j,F_\ell\}=0$ for all $j,\ell=1,\cdots,n$.
	\end{itemize}
\end{definition}

The fact that, under these hypotheses, the equations of motion are exactly solvable follows from the Liouville-Arnol'd-Jost theorem (see Section 49 in \cite{Arnold1989}).

\begin{theorem}[Liouville-Arnol'd-Jost]
	Let $H$ be an integrable system on $\Gamma$ and suppose that each level set of $H$ is compact, connected and non-degenerate. Let $\tilde{a}_1,\cdots,\tilde{a}_n \in \mathbb{R}$ be such that for all $j$, $f^{-1}_j(\tilde{a}_j)$ is non-degenerate, connected and compact. Then,
	\begin{itemize}
		\item[(a)] $\Sigma_{\tilde{a}} \; := \bigcap_{j=1}^n \{f_j=\tilde{a}_j\}$ is a $n$ dimensional manifold.
		\item[(b)] $\Sigma_{\tilde{a}}$ is connected and compact.
		\item[(c)] There exists a neighborhood $U \subset \mathbb{R}^n$ of $\tilde{a}$, such that
		\begin{equation}
			\Sigma_U \,:=\, \bigcup_{a \in U} \Sigma_a
		\end{equation}
		is diffeomorphic to $U \times \mathbb{T}^n$.
		\item[(d)] In $U\times \mathbb{T}^n$ there exist canonical action-angle coordinates $(\varphi,I) \in \tilde{U} \times \mathbb{T}^n$ such that, if $\sigma_U=\Psi(\tilde{U} \times \mathbb{T}^n)$, then $(f_j \circ \Psi)(\varphi,I)=\tilde{f}_j(I)$.
		\item[(e)] In action-angle coordinates, the motion is either periodic or quasi-periodic:
		\begin{equation}
			I(t)\,=\, I_0 \, , \qquad \varphi(t)\,=\, \varphi_0 + \frac{\partial \tilde{H}}{\partial I}(I_0) t \, .
		\end{equation}
	\end{itemize}
\end{theorem}

Geometrically speaking, the phase space of an integrable system is foliated by invariant tori. The motion then takes place on a torus and depending on the frequency vector
\begin{equation}
	\omega(I_0) \, :=\, \frac{\partial \tilde{H}}{\partial I}(I_0) \, ,
\end{equation}
the motion can be either periodic or quasi-periodic. In particular, if $I_0$ is such that $\omega(I_0)$ has all components with irrational ratio, then the motion is ergodic on the torus.

Nevertheless, the fact that we have $n$ integrals of motion on a phase space of dimension $2n$ completely destroys the possibility of treating the motion knowing only few quantities: to analyze the motion, one needs to know the value of all $n$ integrals of motion.

In some sense integrable system are ``exceptional''.

Indeed, being integrable in the sense of Liouville, that is the situation in which the physical system under analysis has $n$ integrals of motion is quite unusual. Instead, many physical systems can be viewed as \emph{perturbations} of certain integrable ones. In these cases, the Hamiltonian can be written as
\begin{equation}\label{eq:GenericTavolfiore}
	H(\varphi,I) \, =\, h(I) + \varepsilon P(\varphi,I)
\end{equation}
where $\varepsilon>0$ is to be thought as a parameter that measures, in some sense, the size of the perturbation. 

For this type of systems, a necessary condition to approach thermal equilibrium, in the ergodic sense, is to require a \emph{substantial} evolution of the actions $I$'s. Indeed, if for all times $|I(t)-I(0)| \leq \varepsilon$, then the dynamics of the system would explore at most a $\varepsilon$-neighborhood of $I(0)$. This last observation is in contradiction with ergodicity, in particular with Theorem~\ref{thm:ErgodicEquivalence}-(b) with the choice $B=\{(\varphi,I) \in \Gamma \, | \, |I-I(0)| \gtrsim \varepsilon\}$. Writing Hamilton equations, from \eqref{eq:GenericTavolfiore}, we have
\begin{equation}
	\begin{split}
		\dot{\varphi}(t)\,&=\, \frac{\partial h}{\partial I}(I(t))+\varepsilon \frac{\partial P}{\partial \varphi}(\varphi(t),I(t)) \, ,\\
		\dot{I}(t) \,&=\, - \varepsilon \frac{\partial P}{\partial \varphi}(\varphi(t),I(t)) \, .
	\end{split}
\end{equation}
Integrating the first equation, we get for any $T>0$ fixed,
\begin{equation}
	|I(t)-I(0)| \leq \varepsilon t \sup_{t \in [0,T]} \Big\Vert \frac{\partial P}{\partial \varphi}(\varphi(t),I(t)) \Big\Vert \, ,
\end{equation}
and, assuming that $\Vert \frac{\partial P}{\partial \varphi}(\varphi(t),I(t)) \Vert \leq C$, we have that $|I(t)-I(0)| \leq C \varepsilon t$, that is, equivalently, we have to wait at least times of order $t \sim \frac{1}{\varepsilon}$ to observe a variation in the value of the actions of order $O(1)$.

In the following simple examples, we show that this bound is optimal.
\begin{example}
	Let $\Gamma=\mathbb{T}^2 \times \mathbb{R}^2 \ni (\varphi,I)$ and let $\omega \in \mathbb{R}^2$. We consider the Hamiltonian system with Hamiltonian $H(\varphi,I)=h(I)+\varepsilon P(\varphi,I)$ with
	\begin{equation}
		h(I)\,=\, \omega \cdot I \, , \qquad P(\varphi,I)\,=\, -\cos(\omega_1 \varphi_1-\omega_2 \varphi_2)\, ,
	\end{equation}
	Hamilton equations are
	\begin{equation}\label{eq:ExamplesEquationsIPhi}
		\begin{split}
			\dot{\varphi}_1(t) \,&=\, \omega_1 \, ,\\
			\dot{\varphi}_2(t) \,&=\, \omega_2 \, , \\
			\dot{I}_1(t) \, &=\, \varepsilon \, \omega_1 \, \sin\big(\omega_1 \varphi_1(t)-\omega_2 \varphi_2(t)\big) \, ,\\
			\dot{I}_2(t) \, &=\, - \varepsilon \, \omega_2 \, \sin\big(\omega_1 \varphi_1(t)-\omega_2 \varphi_2(t) \big) \, .
		\end{split}
	\end{equation}
	Solving the first two equations yield
	\begin{equation}\label{eq:ExamplesSolutionPhi}
		\varphi_1(t)\, =\,\varphi_1(0)+\omega_1 t \, , \qquad \varphi_2(t)\, =\, \varphi_2(0)+\omega_2 t \, .
	\end{equation}
	Inserting the explicit expressions \eqref{eq:ExamplesSolutionPhi} in \eqref{eq:ExamplesEquationsIPhi} we get
	\begin{equation}
		\begin{split}
			\dot{I}_1(t)\,&=\, \varepsilon \omega_1 \, \sin\big( \omega_1 \varphi_1(0)+ \omega_1^2 t-\omega_2 \varphi_2(0)-\omega_2^2 t \big) \\
			\dot{I}_2(t)\,&=\,- \varepsilon \, \omega_2 \, \sin\big( \omega_1 \varphi_1(0)+ \omega_1^2 t-\omega_2 \varphi_2(0)-\omega_2^2 t \big) \, ,
		\end{split}
	\end{equation}
	therefore, it is possible to perform an explicit integration. We have to distinguish two cases.
	
	In case $\omega_1\neq \pm \omega_2$, we have 
	\begin{equation}
		I_1(t) \,=\, I_1(0) -\frac{\varepsilon}{\omega_1^2-\omega_2^2} \cos\big(\omega_1 \varphi_1(0)-\omega_2 \varphi_2(0)+ (\omega_1^2-\omega_2^2)t \big)
	\end{equation}
	and, therefore, for all times, the actions satisfy
	\begin{equation}
		|I_1(t)-I_1(0)| \leq \frac{\varepsilon}{|\omega_1^2-\omega_2^2|} \, , \qquad \forall t \in \mathbb{R} \, .
	\end{equation}
	Therefore, the system for all times explores only a $\varepsilon$-neighborhood of its initial datum and there is no chance to be ergodic and even if $I_1$ is not exactly conserved. We will say that $I_1$ is \emph{quasi-}conserved along the dynamics.
	
	In case $\omega_1=\pm \omega_2$, we have
	\begin{equation}
		I_1(t)=I_1(0)- \varepsilon \omega_1 \sin\big(\omega_1(\varphi_1(0)-\varphi_2(0)) \big) t \, ,
	\end{equation}
	and therefore the action $I_1$ evolves with speed $\varepsilon \omega_1 \sin\big(\omega_1(\varphi_1(0)-\varphi_2(0)) \big)$ and, apart from a zero-measure set of initial data, for times $t \sim \frac{1}{\varepsilon}$ $I_1(t)$ evolves of a quantity that is of order $1$.
\end{example}

In the previous example, we showed that the evolution of the system is denoted by relations between the harmonics in the perturbation and the frequencies of the unperturbed part $\omega(I)$.

Hamiltonian perturbation theory is a framework to detect and analyze systematically these properties. The main idea is to change coordinates in a clever way: in the transformed Hamiltonian, all contributions that \emph{almost-preserve} the actions should be absent.

A convenient way to construct these transformations is to look for flows of functions $G$, and therefore we define $(\tilde{\varphi},\tilde{I})$ as $(\varphi,I)=\Phi_{G_1}^\varepsilon(\tilde{\varphi},\tilde{I})$. For every $G_1$, these type of transformations are $\varepsilon$-\emph{close to the identity} because, for $\varepsilon=0$, by the properties of the flow, one has $(\varphi,I)=(\tilde{\varphi},\tilde{I})$ and, for $\varepsilon$ small, one has
\begin{equation}\label{eq:NearToIdentity}
	|(\varphi,I)-(\tilde{\varphi}, \tilde{I}) | \, \lesssim \varepsilon \, .
\end{equation}
If, for the angles, this vicinity gives no information on the long-time dynamics, for actions this information can be very important. Indeed, we search for a $G_1$ so that $H \circ \Phi_{G_1}^\varepsilon$ is ``more integrable'' than $H$ itself, e.g.
\begin{equation}
	(H \circ \Phi_{G_1}^\varepsilon)(\tilde{I},\tilde{\varphi}) \,=\, h(\tilde{I})+\varepsilon h_1(\tilde{I}) + \varepsilon^2 \tilde{P}_2(\tilde{I},\tilde{\varphi}) \, .
\end{equation}
If we succeed, our a-priori estimate improves
\begin{equation}\label{eq:IIIII}
	|\tilde{I}(t)-\tilde{I}(0)| \, \lesssim \, \varepsilon^2 t \, ,
\end{equation}
that means that actions \emph{essentially} do not evolve for times $t \lesssim \varepsilon^{-2}$. Thus, thermalization time (if any) must be larger than $\varepsilon^{-2}$. Nevertheless, note that such a sharper estimate improves the information on the vicinity in the tilded variables and not in the original ones. Indeed, combining \eqref{eq:NearToIdentity} with \eqref{eq:IIIII}, we have
\begin{equation}
	|I(t)-I(0)| \, \leq \, \underbrace{|I(t) -\tilde{I}(t)|}_{\lesssim \varepsilon} + \underbrace{|\tilde{I}(t)-\tilde{I}(0)|}_{\lesssim \varepsilon^2 t} + \underbrace{|\tilde{I}(0)-I(0)|}_{\lesssim \varepsilon} \, \lesssim 2 \varepsilon + \varepsilon^2 t \, .
\end{equation}
In other words, actions can evolve of a quantity of order $1$ at times $t \sim \frac{1}{\varepsilon^2}$ while exploring a region of size $\varepsilon$ for times of order $\lesssim \frac{1}{\varepsilon}$. In other words, improving integrability of the Hamiltonian does not improve the control in the size of the actions but only on the times on which such control is possible.

We now discuss how to choose $G_1$. Using \eqref{eq:FlowLieDer}, we have
\begin{equation}
	\begin{split}
		(H \circ \Phi_{G_1}^{\varepsilon}) \,&=\, e^{\varepsilon \mathcal{L}_{G_1}} H \,=\,\Big(1+\varepsilon \mathcal{L}_{G_1} + O(\varepsilon^2) \Big)\Big(h+\varepsilon P+O(\varepsilon^2) \Big) \\
		&=\, h+\varepsilon \Big(\mathcal{L}_{G_1}h+P \Big) + O(\varepsilon^2) \, .
	\end{split}
\end{equation}
The strongest (and, often too optimistic) requirement is that  $\mathcal{L}_{G_1}h+P=0$ and therefore it is possible \emph{to push} the perturbation to higher order. A weaker, but still optimistic and sufficient result, is to look whether it is possible to choose $G_1$ depending of both actions and angles and $h_1$ depending on the actions only such that
\begin{equation}\label{eq:AbstractHomologicalEq}
	\mathcal{L}_{G_1} h + P \, = \, h_1 \, .
\end{equation}
This partial differential equations, in the unknowns $G_1$ and $h_1$ is called \emph{homological equation}. Writing explicitly the Poisson brackets (see Definition~\ref{def:PoissonBracket}) and recalling that $h$ depends on the actions only, \eqref{eq:AbstractHomologicalEq} is recast as
\begin{equation}
	\omega(I) \cdot \frac{\partial G_1}{\partial \varphi}(I,\varphi) +P(I,\varphi) \, = \, h_1(I)
\end{equation}
in the unknowns $G_1$ and $h_1$. 

To solve the homological equation it can be convenient to use Fourier series\footnote{This is not always the case, as we will see in Subsection \ref{subsec:InfiniteDimHamSystem}.}. Then, if
\begin{equation}
	G_1(\varphi,I) \,=\, \sum_{k \in \mathbb{Z}^n} \widehat{G_1(I)}_k e^{i k \cdot \varphi} \, , \qquad P(\varphi,I) \,=\, \sum_{k \in \mathbb{Z}^n} \widehat{P(I)}_k e^{i k \cdot \varphi} \, ,
\end{equation}
the homological equation reads
\begin{equation}
	\begin{split}
		- i\, \omega(I) \cdot k \, \widehat{G_1(I)}_k+\widehat{P(I)}_k\,&=\,0 \, , \qquad k \neq 0 \, , \\
		 \widehat{P(I)}_0 \, &= \, h_1(I) \, .
	\end{split}
\end{equation}
Note that we have a solution once $\omega(I) \cdot k \neq 0$ for all $\widehat{P(I)}_k \neq 0$, which is
\begin{equation}
	\widehat{G_1(I)}_k \, = \, \frac{\widehat{P_1(I)}_k}{i \omega(I) \cdot k} \, .
\end{equation}

\begin{definition}
	A vector $\omega \in \mathbb{R}^n$ is said to be \emph{resonant} with $k \in \mathbb{Z}^n \setminus \{0\}$ if $\omega \cdot k = 0$.
\end{definition}

Note that, if $\omega(I)$ is non-resonant for all $k \in \mathbb{Z}^n \setminus \{0\}$, there are still problems! Indeed, generically, $\omega(I) \cdot k$ accumulates at zero and, therefore, one may not be allowed to perform the sum over $k$ in the Fourier series and, therefore, $G_1(\varphi,I)$ is not guaranteed to exist.

\begin{xca}
	Take $\omega=(\sqrt{2},1)$. Show that, for any $\varepsilon>0$ there exists $k_1,k_2 \in \mathbb{Z}$ so that $|\omega \cdot k| < \varepsilon$.
\end{xca}

\begin{definition}\label{def:NonDeg}
	$h(I)$ is said to be \emph{non-degenerate} if $\det \frac{\partial \omega}{\partial I}(I) \neq 0$.
\end{definition}

\begin{definition}
	A perturbation $P(\varphi,I)$ is said to be \emph{generic in the sense of Poincar\'e} in $B \subset \mathbb{R}^n$ if $\forall I \in B$ and $\forall k \in \mathbb{Z}^n$, there exists $q$ parallel to $k$ such that $\widehat{P(I)}_q \neq 0$.
\end{definition}

The following theorem states in which sense integrable systems are exceptional.

\begin{theorem}[Poincar\'e non existence of first integrals] Let $H(I,\varphi;\varepsilon)$,
\begin{equation}
	H(I,\varphi; \varepsilon) \, = \, h(I)+\sum_{j=1}^{+\infty} \varepsilon^j P_j(\varphi,I)
\end{equation}
be analytic in all its arguments, with $h(I)$ nondegenerate and $P_1$ generic in the sense of Poincar\'e in $B\subset \mathbb{R}^n$. Then, no analytic first integrals $F(I,\varphi;\varepsilon)$ independent of $H$ exist.
\end{theorem}

At a first sight, one may be tempted to conclude that due to the absence of analytic integrals of motion, a generic perturbation of an integrable system is ergodic. Despite the fact that this is tempting, the situation is much more subtle as many quasi-integrable systems exhibits very wild (Cantor-like) integrals of motion, as we will discuss in the next section. This last fact is non-trivial at all!

\subsection{KAM and Nekhoroshev theorems} This subsections presents the two greatest theorems of Hamiltonian perturbation theory of the last century. Roughly speaking, they fully describe the dynamics of weakly perturbed integrable systems. Indeed, according to KAM theorem, many invariant tori survives the perturbation and, according to Nekhoroshev theorem, if the initial datum is outside one of these tori, the evolution of the actions is \emph{exponentially slow}.

\begin{theorem}[Kolmogorov-Arnol'd-Moser]\label{thm:KAM} Let $H(\varphi,I)=h(I)+\varepsilon P(\varphi,I)$. Let us suppose
\begin{itemize}
	\item[(i)] $H$ is analytic on a complex neighborhood of $D=B\times \mathbb{T}^n$;
	\item[(ii)] $h$ is non-degenerate (in the sense of Definition \ref{def:NonDeg}).
\end{itemize}

Then, there exists $\varepsilon_*>0$ such that for any $0<\varepsilon < \varepsilon_*$, there exist
\begin{itemize}
	\item[(a)] a canonical transformation close to the identity $(\varphi,I)=\Phi_G^\varepsilon(\tilde{I},\tilde{\varphi})$;
	\item[(b)] an integrable Hamiltonian $\tilde{h}_\varepsilon : B \to \mathbb{R}$;
	\item[(c)] a set $B_\varepsilon \subset B$ with $\mathrm{meas}(B\setminus B_\varepsilon) \sim \sqrt{\varepsilon}$, such that
	\begin{equation}
		(H \circ \Phi_G^\varepsilon)(\tilde{I},\tilde{\varphi}) \, = \, \tilde{h}_\varepsilon(\tilde{I}) \, , \qquad \forall (\tilde{I},\tilde{\varphi}) \in B_\varepsilon \times \mathbb{T}^n \, .
	\end{equation}
\end{itemize}
\end{theorem}

The set $B_\varepsilon$ which has a large measure is such that $\mathrm{int}(B_\varepsilon)=\varnothing$, so it has a ``\emph{wild}'' structure. Moreover, the set $B_\varepsilon \times \mathbb{T}^n$ is invariant for the dynamics and it has a strictly positive measure. Also $(B\setminus B_\varepsilon) \times \mathbb{T}^n$ has strictly positive measure and is invariant for the dynamics, therefore when KAM theorem is applicable, the Hamiltonian system is \emph{not} ergodic on the surface at constant energy (this is a consequence of Theorem \ref{thm:ErgodicEquivalence}-(c)).

The sketch of the proof of the Theorem is contained in Kolmogorov's paper \cite{Kolmogorov1954OnCO} and then complete proofs are contained in Arnol'd's paper \cite{Arnold1963} and Moser's paper \cite{Moser:430015}. The proof of KAM theorem using the tools of Subsection \ref{sec:IntegrablePerturbations} can be found in \cite{Benettin1984}.

What happens outside of the set $B_\varepsilon$ is also interesting: the actions evolve at most exponentially slow. This is a consequence of Nekhoroshev theorem

\begin{theorem}[Nekhoroshev]
	Let $H(\varphi,I)=h(I)+\varepsilon P(\varphi,I)$. Let us suppose
	\begin{itemize}
		\item[(i)] $h$ and $P$ are analytic on a complex neighborhood of $D=B \times \mathbb{T}^n$, with $B$ an open and bounded subset of $\mathbb{R}^n$
		\item[(ii)] $h$ is quasi-convex on $B$,
		 namely, for all $I \in B$, the unique solution in $v \in \mathbb{R}^n$ of the system 
		\begin{equation}
			\begin{cases}
				\nabla_I h(I) \cdot v = 0  \\
				 \partial^2_I h(I) v \cdot v = 0
			\end{cases}
		\end{equation}
		is $v=0$.
		\item[(iii)] $\inf_I |\nabla_I h(I)| >0$ on $B$.
	\end{itemize}
	Let further $a=\frac{1}{2n}$, then  there exists $C>0$, independent of $\varepsilon$, such that
	\begin{equation}
		|I\circ \Phi_H^t - I| \leq C \varepsilon^{a} \, , \qquad \forall |t| \leq \frac{C}{\varepsilon} e^{\frac{1}{\varepsilon^a}} \, .
	\end{equation}
\end{theorem}

The original version of Nekhoroshev theorem relies on a more general hypothesis. Instead of quasi-convexity, item (ii) can be replaced by the more general condition of \emph{steepness}. Also, the exponentially long time is due to the regularity assumption on the Hamiltonian. Less regular Hamiltonians would yield conservation times that are power-law, see \cite{Marco2003,Bounemoura2011,Barbieri2022,Bambusi2021}. Concerning the proof of the result, there are usually two different strategies: one is the original one by Nekhoroshev (see \cite{Nekhoroshev1977,Nek79}) and, in the version with the optimal exponents, \cite{Guzzo2016}. For clear proofs, not in the steep case, we also mention  \cite{Pschel1993,Ben_Gal_Gio,Giorgilli2022}. On the other side, a different proof was provided by Lochack in \cite{Lochak1992,Lochak1992-B}.

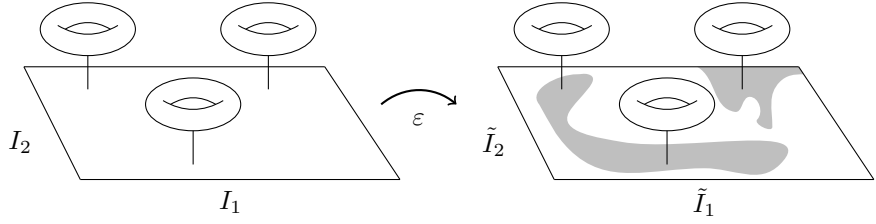
\begin{figure}[h!]

\begin{center}
	\begin{tikzpicture}
			
		\draw[-] (0,0) -- (4.25,0);
		\draw[-] (0,0) -- (-0.75,1.5);
		\draw[-] (-0.75, 1.5) -- (3.25, 1.5);
		\draw[-] (3.25,1.5) -- (4.25,0);

		\draw (0.1,2) ellipse (18pt and 10pt);
		\draw[-] (-0.3,2.05) to[out=330,in=210] (0.5,2.05);
		\draw[-] (-0.2,2) to[out=30, in=150] (0.4,2);	
		\draw[-] (0.1,1.2) -- (0.1,1.65);
		
		\draw (1.5,1) ellipse (18pt and 10pt);
		\draw[-] (1.1,1.05) to[out=330,in=210] (1.9,1.05);
		\draw[-] (1.2,1) to[out=30, in=150] (1.8,1);	
		\draw[-] (1.5,0.2) -- (1.5,0.65);		

%(0,0.4)

		\draw (2.5,2) ellipse (18pt and 10pt);
		\draw[-] (2.1,2.05) to[out=330,in=210] (2.9,2.05);
		\draw[-] (2.2,2.0) to[out=30, in=150] (2.8,2.0);	
		\draw[-] (2.5,1.2) -- (2.5,1.65);
		
		\node at (2,-0.3) {$I_1$};
		
		\node at (-0.8,0.5) {$I_2$};

\draw[->,thick] (4,1) to[out=40,in=140] (5,1);		
		
		\node at (4.5,0.8) {$\varepsilon$};

		\begin{scope}[xshift=6.3cm]
	
%%%
%%%% Macchia 1 (curva arrotondata) che contiene il punto (1.5, 0.2)
%%%\begin{scope}[xshift=1cm,yshift=0.01cm]
%%%\fill[lightgray] (1.5, 0.2) .. controls (1.8, 0.7) and (1.9, 1.1) .. (1.6, 1.3) .. controls (1.3, 1.5) and (1.1, 1.2) .. (1.2, 0.9) .. controls (1.3, 0.7) and (1.5, 0.6) .. (1.5, 0.2);
%%%\end{scope}
%%%
%%%% Macchia 2 (curva arrotondata) che contiene il punto (0.1, 1.2)
%%%\begin{scope}[xshift=-0.3cm]
%%%\fill[lightgray] (0.1, 1.2) .. controls (0.3, 1.3) and (0.6, 1.5) .. (0.9, 1.4) .. controls (1.0, 1.2) and (0.7, 1.0) .. (0.4, 1.0) .. controls (0.2, 1.1) and (0.1, 1.1) .. (0.1, 1.2);
%%%\end{scope}
%%%
%%%% Macchia 3 (curva arrotondata) che contiene il punto (2.5, 1.2)
%%%\begin{scope}[yshift=-1.02cm,xshift=-0.9cm]
%%%\fill[lightgray] (2.5, 1.2) .. controls (2.7, 1.4) and (2.9, 1.5) .. (2.6, 1.6) .. controls (2.3, 1.7) and (2.2, 1.4) .. (2.4, 1.1) .. controls (2.6, 1.0) and (2.5, 1.1) .. (2.5, 1.2);
%%%\end{scope}
%%%
%%%% Macchia 4 (curva arrotondata)
%%%\begin{scope}[yshift=-0.8cm]
%%%\fill[lightgray] (3.4, 0.8) .. controls (3.6, 1.0) and (3.8, 1.3) .. (3.6, 1.5) .. controls (3.4, 1.6) and (3.2, 1.4) .. (3.3, 1.2) .. controls (3.4, 1.0) and (3.5, 0.9) .. (3.4, 0.8);
%%%\end{scope}

	\coordinate (A) at (-0.2,1);
	\coordinate (B) at (0.3,1.4);
	\coordinate (C) at (0.45,1.4);
	\coordinate (D) at (0.2, 0.9);
	\coordinate (E) at (0.9,0.5);
	\coordinate (F) at (3,0.2);
	\coordinate (G) at (2,0.15);
	\coordinate (H)	at (0.2,0.4);
	
	\coordinate (A2) at (3.25,1.5);
	\coordinate (B2) at (1.9,1.5);
	\coordinate (C2) at (2.4, 0.9);
	\coordinate (D2) at (2.6, 1.1);
	\coordinate (E2) at (2.7, 0.8);
	\coordinate (F2) at (2.9, 0.7);
	\coordinate (G2) at (3.0, 1.3);
	\coordinate (H2) at (3.316,1.4); 
	
	\fill[lightgray] (A2) to[out=180,in=0] (B2)
	to[out=-30, in=170] (C2)
	to[out=-10, in=180] (D2)
	to[out=0, in=45] (E2)
	to[out=-135, in=-90] (F2)
	to[out=90, in=-135] (G2)
	to[out=45, in=180] (H2)
	to (A2);

	\fill[lightgray] (A) to[out=130,in=180] (B)
	to[out=0,in=165] (C)
	to[out=-15,in=60] (D)
	to[out=-120,in=170] (E)
	to[out=-10, in=45] (F)
	to[out=-135, in=0] (G)
	to[out=180, in=-45] (H)
	to[out=135,in=-60] (A);

		\draw[-] (0,0) -- (4.25,0);
		\draw[-] (0,0) -- (-0.75,1.5);
		\draw[-] (-0.75, 1.5) -- (3.25, 1.5);
		\draw[-] (3.25,1.5) -- (4.25,0);

		\draw (0.1,2) ellipse (18pt and 10pt);
		\draw[-] (-0.3,2.05) to[out=330,in=210] (0.5,2.05);
		\draw[-] (-0.2,2) to[out=30, in=150] (0.4,2);	
		\draw[-] (0.1,1.2) -- (0.1,1.65);
		
		\draw (1.5,1) ellipse (18pt and 10pt);
		\draw[-] (1.1,1.05) to[out=330,in=210] (1.9,1.05);
		\draw[-] (1.2,1) to[out=30, in=150] (1.8,1);	
		\draw[-] (1.5,0.2) -- (1.5,0.65);

		\draw (2.5,2) ellipse (18pt and 10pt);
		\draw[-] (2.1,2.05) to[out=330,in=210] (2.9,2.05);
		\draw[-] (2.2,2.0) to[out=30, in=150] (2.8,2.0);	
		\draw[-] (2.5,1.2) -- (2.5,1.65);

		\node at (2,-0.3) {$\tilde{I}_1$};
		
		\node at (-0.8,0.5) {$\tilde{I}_2$};

	\end{scope}	
	\end{tikzpicture}
\end{center}

	\caption{Cartoon picture of the phase space of an integrable system and the foliation on invariant tori (left). On the right, the persistence of tori for small perturbations as a consequence of KAM theorem. The set $B_\varepsilon$ is unfaithfully represented by gray shaded regions in the space of actions. The white region in the $\tilde{I}_1$,$\tilde{I}_2$ plane is a cartoon picture of the Arnol'd web where actions evolve exponentially slow as a consequence of Nekhoroshev theorem.}
\end{figure}

From our perspective, KAM theorem says that, if the size of the perturbation is small enough and the other hypotheses of the theorem are satisfied, then no thermalization can take place because the Hamiltonian system is not ergodic and, therefore, not mixing. It is debated whether the smallness threshold vanishes or not in the thermodynamic limit.

\subsection{Infinite dimensional systems}\label{subsec:InfiniteDimHamSystem}
In this section, we briefly examine how the constructions above extend to the infinite-dimensional case. (The material here is taken from \cite{Gallone2022} and from the classical textbooks \cite{Dubrovin1984,Marsden1999,Gelfand2000-fx}.) Here, we will restrict ourselves to the scenario where the spatial variable lies on the one-dimensional torus $\mathbb{T}^1$ and time on $\mathbb{R}$. Furthermore, we will consider only analytic functions. Needless to say, the general theory \emph{does not} require all these properties.

In this framework, the phase space $\Gamma$ is a space of functions and an element can be given by one or more functions, e.g.~$(Q,P) \in \Gamma$. Hamiltonians are functionals of the phase space of the form
\begin{equation}
	\mathscr{H}(Q,P)\,=\, \int_{\mathbb{T}^1} H(x,Q,P,Q_x,P_x,\dots) \, dx
\end{equation}
where $H$ is a function, called \emph{density of the functional $\mathscr{H}$}, that depends on $Q, P$ and their derivatives (up to a finite order).

In this setting, the gradient is replaced by the \emph{functional derivatives}:
\begin{equation}
	\frac{\delta \mathscr{H}}{\delta Q}(Q,P) \, := \, \sum_{j \geq 1} (-1)^j \frac{d^j}{d x^j} \frac{\partial \mathscr{H}}{\partial Q_{jx}}(x,Q,P) \, ,
\end{equation}
where we denoted $Q_{jx}=\partial_x^j Q$. Then, analogously to the finite dimensional case, Hamilton equations require the definition of the \emph{Poisson tensor} $J$. We will consider only the following Poisson tensors:
\begin{equation}\label{eq:PoissonTensors}
	J_E\,=\, \begin{pmatrix}
		\mathbb{O} & \mathbbm{1} \\
		-\mathbbm{1} & \mathbb{O}
	\end{pmatrix} \, , \qquad J_{\mathrm{G}} \, = \, \begin{pmatrix}
		\partial_x & \mathbb{O} \\
		\mathbb{O} & - \partial_x
	\end{pmatrix} \, , \qquad J_{\mathrm{K}} \, = \, \partial_x \, .
\end{equation}
Note that, in general, one does not require Poisson tensors to be non-degenerate. Indeed, in general, there are functions whose $L^2$-gradient is in the kernel of $J$. Those special functions are called \emph{Casimir invariants}.

Once a Poisson tensor is given, e.g.~$J_E$, equations of motion are defined as
\begin{equation}
	\begin{split}
		\begin{pmatrix}
			Q_t(x,t) \\
			P_t(x,t)
		\end{pmatrix} \, = \, J_{E} \begin{pmatrix}
			\frac{\delta \mathscr{H}}{\delta Q}(x,Q,P) \\
			\frac{\delta \mathscr{H}}{\delta P}(x,Q,P)
		\end{pmatrix} \, .
	\end{split}
\end{equation}
To have a more compact notation, we can collect $Z=(Q,P)$ and define the $L^2$-gradient as
\begin{equation}
	\nabla_Z \mathscr{H}(Z) \, = \, \begin{pmatrix}
		\frac{\delta \mathscr{H}}{\delta Q}(x,Q,P) \\
		\frac{\delta \mathscr{H}}{\delta P}(x,Q,P)
	\end{pmatrix} \, .
\end{equation}
We can define the Poisson bracket associated to the Poisson tensor $J$ as
\begin{equation}
	\{\mathscr{F}(Z),\mathscr{G}(Z)\}_J \, = \, \int_{\mathbb{T}^1} \big[(\nabla_Z \mathscr{F})(x,Z)\big]^T \, J  \, \nabla_Z \mathscr{G}(x,Z) \, dx \, .
\end{equation}
Then, all the results of the previous sections hold, apart from the Liouville-Arnol'd-Jost theorem. Indeed, integrability in infinite dimension is a delicate issue and it is not known whether each integrable Hamiltonian admits a transformation in action-angle variables.

Canonical transformations are those transformations of the phase space which preserve the Poisson bracket structure and, analogously to Lemma \ref{lem:FlowCan}, the flow of any Hamiltonian is a canonical transformation.

The formal construction of perturbation theory is the same as in the finite dimensional case, without the requirement of having an initial Hamiltonian in action-angle coordinates. Starting from
\begin{equation}
	\mathscr{H}(Q,P) \, = \, \mathscr{H}_0(Q,P) + \varepsilon \mathscr{P}(Q,P)
\end{equation} 
we need to drop the assumption that $\mathscr{H}_0$ is written in action-angle coordinates. We will require, instead, that the equations of motion of $\mathscr{H}_0$ are exactly solvable. Thus, repeating the steps of \ref{sec:IntegrablePerturbations}, the core of our perturbative analysis is the \emph{homological equation}
\begin{equation}\label{eq:HomologicalInfiniteDim}
	\{\mathscr{H}_0 , \mathscr{G}_1\} + \mathscr{P}_1 \, = \, \mathscr{H}_1
\end{equation}
in the unknowns $\mathscr{G}_1$ and $\mathscr{H}_1$. The additional requirement on $\mathscr{H}_1$ we would like to fulfill is the fact that $\mathscr{H}_1$ is in normal form with respect to $\mathscr{H}_0$, that is $\{\mathscr{H}_0,\mathscr{H}_1\}=0$. In the most general case, equation \eqref{eq:HomologicalInfiniteDim} is not easy to solve, but in case $\mathscr{H}_0$ has a \emph{periodic flow}, homological equation can be solved explicitly.

\begin{definition}\label{eq:TimeAverageH0}
	If $\mathscr{H}_0$ has a periodic flow of period $T>0$, the time average of a functional $\mathscr{F}$ along the flow of $\mathscr{H}_0$ is denoted by
	\begin{equation}\label{eq:TimeAverageDef}
		[\mathscr{F}]_{\mathscr{H}_0}(Z) \, := \, \frac{1}{T} \int_0^T (\mathscr{F} \circ \Phi_{\mathscr{H}_0}^s)(Z) \, ds \, .
	\end{equation}
\end{definition}

\begin{lemma}\label{lem:ConstantMotionAverage}
	Let $\mathscr{F}$ be a functional on $\Gamma$ and $\mathscr{H}_0$ be a Hamiltonian with periodic flow. Then,
	\begin{equation}
		\mathcal{L}_{\mathscr{H}_0} \big([\mathscr{F}]_{\mathscr{H}_0} \big)=0 \, .
	\end{equation}
\end{lemma}
\begin{proof}
	Composing on the left- and right-hand side of \eqref{eq:TimeAverageDef} with the flow of $\Phi_{\mathscr{H}_0}^t$ one gets
	\begin{equation}
		\text{(LHS)} \, = \, [\mathscr{F}]_{\mathscr{H}_0} (\Phi_{\mathscr{H}_0}^t(Z)) \, , \qquad \text{(RHS)} \, = \, \frac{1}{T} \int_0^T (\mathscr{F} \circ \Phi_{\mathscr{H}_0}^{s+t})(Z) \, ds \, .
	\end{equation}
	 Then, using periodicity of the flow and performing the change of coordinates $s+t=t'$, we get 
	 \begin{equation}
	 	[\mathscr{F}]_{\mathscr{H}_0}(\Phi^t_{\mathscr{H}_0}(Z)) \, = \, [\mathscr{F}]_{\mathscr{H}_0}(Z) \, ,
	 \end{equation}
	 which implies $\frac{d [\mathscr{F}]_{\mathscr{H}_0}}{dt}(Z)=0$ and, therefore \eqref{eq:TimeAverageDef}.
\end{proof}

\begin{proposition}\label{prop:Averaging}
	Let $\mathscr{H}_0$ be a Hamiltonian with periodic flow. Then, a solution to the homological equation \eqref{eq:HomologicalInfiniteDim} is given by
	\begin{eqnarray}
		\mathscr{H}_1 \!\!\!\!&=&\!\!\!\! [\mathscr{P}]_{\mathscr{H}_0} \\
		\mathscr{G}_1 \!\!\!\!&=& \!\!\!\! \frac{1}{T} \int_0^T (s-T) \big(\mathscr{P}-[\mathscr{P}]_{\mathscr{H}_0} \big) \circ \Phi^s_{\mathscr{H}_0} \, ds \, . \label{eq:G1Calcolo}
	\end{eqnarray}
\end{proposition}
\begin{proof} 
	Starting from the homological equation \eqref{eq:HomologicalInfiniteDim}, applying $\Phi_{\mathscr{H}_0}^t$ on both sides and integrating over $[0,T]$, one has
	\begin{equation}
		\int_0^T \{\mathscr{H}_0, \mathscr{G}_1\}(\Phi^t_{\mathscr{H}_0})(Z)) \, dt + T [\mathscr{P}_1]_{\mathscr{H}_0} \, = \, T \mathscr{H}_1(Z) \, ,
	\end{equation}
	where we used the fact that $\mathscr{H}_1(\Phi_{\mathscr{H}_0}^t(Z)) = \mathscr{H}_1(Z)$ (requirement of invariance of $\mathscr{H}_1$ along the flow of $\mathscr{H}_0$, Lemma \ref{lem:ConstantMotionAverage}) and the definition of time-average \eqref{eq:TimeAverageDef}. Then, by noting that $\frac{d}{ds} \mathscr{G}_1(\Phi^s_{\mathscr{H}_0}(Z))=-\{\mathscr{H}_0,\mathscr{G}_1\}$, we have
	\begin{equation}
		\int_0^T \{\mathscr{H}_0,\mathscr{G}_1\}(\Phi^t_{\mathscr{H}_0}(Z)) \, dt = \int_0^T \frac{d}{dt} \mathscr{G}_1(\Phi^t_{\mathscr{H}_0}(Z)) \, dt = \mathscr{G}_1(Z)-\mathscr{G}_1(Z)=0
	\end{equation}
	where in the last step we used the periodicity of the flow of $\mathscr{H}_0$. Then, the fact that $[\mathscr{P}_1]_{\mathscr{H}_0}$ commutes with $\mathscr{H}_0$ follows from Lemma \ref{lem:ConstantMotionAverage}.
	
	To obtain \eqref{eq:G1Calcolo} we multiply left- and right- hand side of \eqref{eq:HomologicalInfiniteDim} by $(s-T)$ and compute both sides at $\Phi_{\mathscr{H}_0}^{s}(Z)$. This yields
	\begin{equation}\label{eq:VienDiNotte}
		(s-T) \, \{\mathscr{H}_0,\mathscr{G}_1\}(\Phi_{\mathscr{H}_0}^s(Z))\, + (s-T) \mathscr{P}_1(\Phi_{\mathscr{H}_0}^s(Z)) = (s-T) \mathscr{H}_1(Z) \, , 
	\end{equation}
	and, on the left-hand side, we have
	\begin{equation}\label{eq:LaBefana}
	 (s-T) \, \{\mathscr{H}_0,\mathscr{G}_1\}(\Phi_{\mathscr{H}_0}^t(Z)) \, = \, -(s-T) \, \frac{d}{ds} \mathscr{G}_1(\Phi^s_{\mathscr{H}_0}(Z)) \, .
	\end{equation}
	Inserting \eqref{eq:LaBefana} in \eqref{eq:VienDiNotte} and integrating from $0$ to $T$, we have
	\begin{equation}
		\begin{split}
			-\int_0^T (s-T) \frac{d}{ds} \mathscr{G}_1(\Phi_{\mathscr{H}_0}^s(Z)) \, ds \, &= \, \int_0^T \mathscr{G}_1(\Phi_{\mathscr{H}_0}^s(Z)) \, ds \, \\
			&= \, - T \mathscr{G}_1(Z) + \int_0^T \mathscr{G}_1(\Phi_{\mathscr{H}_0}^s(Z)) \, ds
		\end{split}
	\end{equation}
	and, requiring that $[\mathscr{G}_1]_{\mathscr{H}_0}(Z)=0$, we obtain \eqref{eq:G1Calcolo}, which completes the proof.
\end{proof}

We will use these results later in Subsection \ref{subsect:KdVResonantNF} to derive the Korteweg-de Vries equation from a continuous interpolation of the FPUT lattice.

\section{The Fermi-Pasta-Ulam-Tsingou chain}
The Fermi-Pasta-Ulam-Tsingou model or, briefly FPUT, is a Hamiltonian system on $\mathbb{R}^{2N}$ with
\begin{equation}
	H_{\mathrm{F}}(q,p) \, = \, \sum_{j \in \mathbb{Z}_N} \Big[ \frac{p_j^2}{2} + \Phi_{\mathrm{F}}(q_{j+1}-q_j) \Big]
\end{equation}
where $\mathbb{Z}_N=\mathbb{Z}/N$  (i.e.~ $\mathbb{Z}$ with the equivalence relation $j\sim j+N$) and
\begin{equation}\label{eq:FPUT-Potential}
	\Phi_{\mathrm{F}}(z) \, = \, \frac{z^2}{2}+ \alpha \frac{z^3}{3} + \beta \frac{z^4}{4} \, .
\end{equation}
In this section, we will first derive the model from a realistic Hamiltonian system. Then, we will discuss the content of the FPUT report, the formulation of the FPUT paradox and the subsequent numerical experiments and attempts of explanation. We conclude by discussing the relation of the FPUT paradox and KAM theorem.

\subsection{Derivation of the model} In this section we give a derivation of the model used by Fermi and co-workers in their numerical simulation. This is pedagogical for two reasons: one the one side, it explicitly shows what is the physical meaning of the model; on the other side it permits to obtain explicitly the value of the parameters of the FPUT model starting from realistic potentials. This derivation clarifies the relation with other realistic models considered in the literature, see e.g.~\cite{Galgani1972,Benettin2023}.

A good description of the model can be found in the original FPUT report \cite{FPU55}:
\begin{quote}
	We imagine a one-dimensional continuum with the ends kept fixed and with forces acting on the elements of the string. In addition to the usual linear term expressing the dependence of the force on the displacement of the element, this force contains higher order terms. For the purpose of the numerical work this continuum, is replaced by a finite number of points (at most 64 in our actual computation) so that the partial differential equation defining the motion of this string is replaced by a finite number of total differential equations. We have, therefore, a dynamical system of 64 particles with forces acting between neighbors with fixed endpoints.
\end{quote}

The model of ``discretized string'' of length $L$ we have in mind consists of $N$ particles interacting with their nearest neighbors. Differently from the original FPUT model, we put periodic boundary conditions, but this will not change dramatically the derivation. The case of fixed ends can always be reconstructed by following the dynamics of antisymmetric initial data. The state of the particle labeled by $j$ is characterized by a position $x_j \in \mathbb{R}$ and its momentum $v_j \in \mathbb{R}$ (here, to avoid confusion, we denote the physical momentum by $v$). If we suppose that each particle has mass $m>0$, the Hamiltonian of the model is
\begin{equation}\label{eq:HamPhysVar}
	H_{\text{string}}(x,v) \,=\, \sum_{j \in \mathbb{Z}_N} \Big[\frac{v_j^2}{2m}+\Phi(x_{j+1}-x_j) \Big] \, ,
\end{equation}
where $\mathbb{Z}_N=\mathbb{Z}/N$ and $\Phi$ is the interaction potential that depends on the inter-particle interaction. Possible choices could be the harmonic one, Morse potential, Lennard-Jones etc.

To be a good model of a string, the Hamiltonian \eqref{eq:HamPhysVar} must have an equilibrium when all particles are such that $v_j=0$ and $x_j=j a$. Indeed, this corresponds to the state where all particles are equispaced at distance $a$, that has the meaning of \emph{lattice spacing}, with
\begin{equation}
	a\,=\, \frac{L}{N} \, .
\end{equation}
We can now introduce the variables $q_j$ that are the displacements from the equilibrium as
\begin{equation}
	x_j \, =\, a q_j + j \frac{L}{N} \, ,
\end{equation}
which states also that, in physical dimensions, if $x_j$ has the dimension of a length, $q_j$ is dimensionless. To get a dimensionless momentum $p_j$, we define
\begin{equation}
	v_j \, = \, \frac{ma}{\tau} p_j \, ,
\end{equation}
where $\tau \in \mathbb{R}$ is a \emph{natural time-scale}, which has the dimension of a time and whose value is going to be fixed later.

Last, to get a dimensionless time variable $T$, it is natural to set
\begin{equation}
	t\,=\, T \, \tau \, .
\end{equation}

Using that the transformation of coordinates
\begin{equation}
	(q,p,H,T) \mapsto (x,v,H_{\mathrm{string}},t)=(\alpha q, \beta p, \gamma H, \delta T)
\end{equation}
is a canonical rescaling (Definition \ref{def:CanonicalRescaling}) if and only if $\alpha \beta = \gamma \delta$, we obtain, as a consequence of Lemma \ref{lem:CanResc}
\begin{equation}
	\gamma=\frac{ma^2}{\tau^2} \, ,\qquad H_{\mathrm{string}}(x(q),v(p)) \, = \, \frac{m a^2}{\tau^2} H(q,p) \, .
\end{equation}
The explicit computation yields
\begin{equation}\label{eq:ContoHamiltonianaFPUT}
	H(q,p) \, = \, \sum_{j \in \mathbb{Z}_N} \Big[ \frac{p_j^2}{2}+\frac{\tau^2}{ma^2}\Phi\big(a+a(q_{j+1}-q_j)\big)\Big] \, .
\end{equation}
Expanding for $q_j$ small, we get
\begin{equation}
	\begin{split}
	\Phi\big(a+a(q_{j+1}-q_j)\big) \, &= \, \Phi(a)+a\Phi'(a)(q_{j+1}-q_j)+\frac{a^2}{2} V''(a) (q_{j+1}-q_j)^2 \\
	&+\frac{a^3}{3!}V'''(a)(q_{j+1}-q_j)^3+\frac{a^4}{4!} V''''(a)(q_{j+1}-q_j)+\cdots
	\end{split}
\end{equation}
The term $\Phi(a)$ is a constant and therefore does not contribute to the dynamics, while due to periodic boundary conditions, we have $\sum_{j \in \mathbb{Z}_N}( q_{j+1}-q_j) = 0$, whence the second term is not present.\footnote{In case of fixed boundary conditions, $\sum_{j=2}^N(q_{j+1}-q_j)=q_N-q_1=0$.} Choosing now 
\begin{equation}
	\tau^2 \, = \, \frac{m}{V''(a)} \, ,
\end{equation}
Hamiltonian \eqref{eq:ContoHamiltonianaFPUT} reads
\begin{equation}
	\begin{split}
	H(q,p) \, &= \, \sum_{j \in \mathbb{Z}_N} \Big[ \frac{p_j^2}{2} + \frac{1}{2}(q_{j+1}-q_j)^2 + \frac{1}{3!} \frac{a V'''(a) }{V''(a)}(q_{j+1}-q_j)^3 \\
	& \qquad +\frac{1}{4!} \frac{a^2 V''''(a)}{V''(a)}(q_{j+1}-q_j)^4+\cdots\Big]
	\end{split}
\end{equation}
which is the FPUT Hamiltonian if we set
\begin{equation}
	\alpha\,=\, \frac{a V'''(a)}{2 V''(a)} \, , \qquad \beta\,=\,\frac{a^2 V''''(a)}{6 V''(a)} \, .
\end{equation}

\subsection{Statistical mechanics of the harmonic chain}\label{subsec:StatMechHarmonic}
To control the approach to thermal equilibrium of the FPUT chain, one needs to find a meaningful indicator. A good starting point is to choose a set of observables and to compute their expected values at thermal equilibrium. Fermi and his collaborators decided to use the energies of the normal modes of oscillation. Indeed, if $\alpha=\beta=0$, the harmonic chain is integrable, and it is possible to prove that having defined
\begin{equation}\label{eq:DirectFourierLulu}
	\begin{split}
		\hat{q}_k \, &= \, \frac{1}{\sqrt{N}} \sum_{j \in \mathbb{Z}_N} q_j e^{2 \pi i k \frac{j}{N}} \, ,\\
		\hat{p}_k \, &= \, \frac{1}{\sqrt{N}} \sum_{k \in \mathbb{Z}_N} p_j e^{2 \pi i k \frac{j}{N}} \, ,
	\end{split}
\end{equation}
and
\begin{equation}\label{eq:OmegaK}
	\omega_k \, = \, \Big| 2 \sin \big( {\textstyle \frac{\pi k}{2 N}} \big) \Big| \, = \, |2 \sin \big( {\textstyle \frac{\pi}{2} h k } \big)| ,
\end{equation}
the quantities  
\begin{equation}
    E_k\,=\, \frac{1}{2} \Big( |\hat{p}_k|^2 + \omega_k^2 |\hat{q}_k|^2 \Big) \, 
\end{equation}
are conserved for all times.  

One of the consequences of statistical mechanics, particularly interesting because it allows for highly successful (and occasionally unsuccessful\footnote{If applied universally, the principle of equipartition of energy leads to results that contradict experiments, such as the Rayleigh-Jeans law or incorrect values for the specific heat of solids at low temperatures. The failure of these predictions, however, is due to the quantum nature of matter. Indeed, the Rayleigh-Jeans result led Max Planck to introduce the hypothesis of light quanta. The interested reader is referred to the chapter ``Equipartition and Critique'' in \cite{Gallavotti1999}.}) quantitative predictions, is the \emph{principle of equipartition of energy}. Without delving into a general and detailed treatment, we will illustrate it here in the case of the harmonic chain as follows.

\begin{lemma}[Equipartition of energy]\label{lem:EquipartitionOfEnergy}
	For the harmonic chain, one has	
	\begin{equation}
		 \lim_{\substack{E \to +\infty \\ N \to +\infty \\ E/N=\epsilon}}\langle E_k \rangle_{\mu_{\mathrm{mc}}(E,N)}\, = \, \epsilon \, .
	\end{equation}
\end{lemma}
\begin{proof}
Inverting \eqref{eq:DirectFourierLulu}, we get
	\begin{equation}
		\begin{split}
			q_j \,&=\, \frac{1}{\sqrt{N}} \sum_{k=1}^N \hat{q}_k e^{- 2\pi i k \frac{j}{N}} \, , \\
			p_j\,&=\, \frac{1}{\sqrt{N}} \sum_{k=1}^N \hat{p}_k e^{2 \pi i k \frac{j}{N}} \, ,
		\end{split}
	\end{equation}
	and therefore, since this is a linear map, we have $\prod_{j=1}^N dq_j dp_j= C_N \prod_{j=1}^N d \hat{q}_k \hat{p}_k$, for some explicit constant $C_N$. We can rely on ensemble equivalence and compute $\langle E_k \rangle_{\mu_{\mathrm{can}}(\beta,N)}$, with inverse temperature $\beta=\frac{1}{\epsilon}$. The canonical partition function is
	\begin{equation}
		\begin{split}
			Z_{\mathrm{can}}(\beta,N) \, &= \, \int_{\mathbb{R}^{2N}} dq dp  \, e^{-\beta H(N)} \, = \, C_N \int_{\mathbb{R}^{2N}} d \hat{q} d \hat{p} \, e^{-\beta \sum_{k=1}^N E_k(\hat{q}_k,\hat{p}_k)} \\
			&=C_N \prod_{k=1}^N \int d\hat{q}_k d\hat{p}_k e^{-\beta E_k(\hat{q}_k,\hat{p}_k)} \, . 
		\end{split} 
	\end{equation}
	Therefore,
	\begin{equation}
		\begin{split}
			\langle E_k \rangle_{\mu_{\mathrm{can}}(\beta,N)} \,&=\, \frac{C_N \int d \hat{q} d \hat{p} \, E_k(\hat{q}_k,\hat{p}_k) \, e^{-\beta \sum_{k=1}^N E_k(\hat{q}_k,\hat{p}_k)}}{C_N \int d\hat{q} d \hat{p} \, e^{- \beta \sum_{k=1}^N E_k(\hat{q},\hat{p})}} \\
			&=\, \frac{\int d\hat{q} d \hat{p} \, E_k(\hat{q}_k,\hat{p}_k) \, e^{-\beta E_k(\hat{q}_k,\hat{p}_k)}}{\int d\hat{q} d \hat{p} \, e^{-\beta E_k(\hat{q}_k,\hat{p}_k)}} \\
			&=\,-\frac{\partial}{\partial \beta} \ln \Big(\int d \hat{q}_k d \hat{p}_k \, e^{-\beta E_k(\hat{q}_k,\hat{p}_k)} \Big) \, .
		\end{split}
	\end{equation}
	The last integral can be solved explicitly
	\begin{equation}
		\int d \hat{q}_k d \hat{p}_k \, e^{-\beta E_k(\hat{q}_k,\hat{p}_k)} \, = \, \frac{2 \pi}{\beta \omega_k} \, ,
	\end{equation}
	and using that $\beta=\frac{1}{\epsilon}$ completes the proof.
\end{proof}

Later on, it was originally suggested by Livi, Pettini, Ruffo, Sparpaglione and Vulpiani (see Eq.~(12) in \cite{Livi1985}) that a good indicator to detect equipartition of energy is the so-called \emph{spectral entropy} defined as
\begin{equation}
	\eta(t) \,=\, - \sum_{k=1}^{N/2} \nu_k(t)  \ln \nu_k(t)
\end{equation}
with
\begin{equation}\label{eq:SpectralEntropy}
 \qquad \nu_k(t)\,=\, \frac{\overline{E_k}(t)}{E}
\end{equation}
where the bar denotes the time-average over the interval $[0,t]$. We remark here that a part of the literature uses a different spectral entropy based on $\nu_k(t)= \frac{E_k(t)}{E}$.

The number of effectively excited modes is then defined as
\begin{equation}\label{eq:Nex}
	N_{\mathrm{excited}}(t) \,=\, e^{\eta(t)} \, .
\end{equation}
Since at thermal equilibrium all normal modes have the same energy, one sees from \eqref{eq:SpectralEntropy} that when all $E_k$'s are constant and equal among them, then at equipartition $N_{\mathrm{excited}} = N$. The definition of effectively excited modes is reasonable because if an energy $E$ is equally distributed among $n$ modes (that is, each of the modes has either energy $E/n$ or $0$), one has $N_{\mathrm{excited}}=n$.

\subsection{FPUT's original numerical experiment and the paradox}

To describe the FPUT experiment, the motivations and the results, there are no better words than Ulam's introduction to the Report \cite{FPU55}:

\begin{quote}
	After the war, during one of his frequent summer visit to Los Alamos, Fermi became interested in the development and potentials of the electronic computing machines. He held many discussions with me on the kind of future problems which could be studied through the use of such machines. We decided to try a selection of problems for heuristic work where in absence of closed analytic solutions experimental work on a computing machine would perhaps contribute to the understanding of properties of solutions. This could be particularly fruitful for problems involving the asymptotic-long time or ``in the large'' behavior of non-linear physical systems. In addition, such experiments on computing machines would have at least their postulates clearly stated. This is not always the case in actual physical object or model where all the assumptions are not perhaps explicitly recognized.
	
	[\dots] The plan was then to start with the possibly simplest such physical models and to study the results of the calculation of its long-time behavior. Then one would gradually increase the generality and the complexity of the problem calculated on the machine. The Los Alamos report \textsc{la}-1940 presents the results of the very first such attempt. [\dots] These were to be studied preliminary to setting up ultimate models for motions of system where ``mixing'' and ``turbulence'' would be observed. The motivation was to observe the \emph{rates} of mixing and ``thermalization'' with the hope that the calculational results would provide hints for a future theory. One could venture a guess that one of the motive in the selection of problems could be traced to Fermi's earliest interest in the ergodic theory. [\dots]
	
	The results of the calculations (performed on the old \textsc{Maniac} machine) were interesting and quite surprising to Fermi. He expressed the opinion that they really constituted a little discovery in providing intimations that the prevalent beliefs in the universality of ``mixing and thermalization'' in non-linear systems may not be always justified.
\end{quote}

Fermi and collaborators considered the model
\begin{equation}
	H_{\mathrm{F}}(q,p) \, = \, \sum_{j=2}^{N=64} \Big[\frac{p_j^2}{2}+\Phi_{\mathrm{F}}(q_j-q_{j-1}) \Big] \, ,
\end{equation}
with fixed boundary conditions, i.e.~$q_1=q_N=0$, $p_1=p_N=0$. Concerning the nonlinearities, in the original report, several forms are considered. Here we limit ourselves to the analysis of the $\Phi_{\mathrm{F}}$, with $\Phi_{\mathrm{F}}$ given in \eqref{eq:FPUT-Potential} and to the case of periodic boundary conditions. The conclusion and the phenomenology is the same, but the analysis of the periodic case is somehow simpler.

It is nowadays common to refer to the FPUT model as one of the latter and, in particular, to classify them as
\begin{itemize}
	\item $\alpha$--model, when $\alpha \neq 0$ and $\beta=0$;
	\item $\beta$--model, when $\alpha=0$ and $\beta \neq 0$;
	\item $\alpha+\beta$--model, when $\alpha \neq 0 $ and $\beta \neq 0$.
\end{itemize}
The first result we shall comment concerns the time evolution of the energy of Fourier modes $E_k(t)$ along time. This is reported in Figure \ref{fig:FPUOriginal}. 

\begin{figure}[h!]
	\begin{center}
		\includegraphics[width=0.8\textwidth]{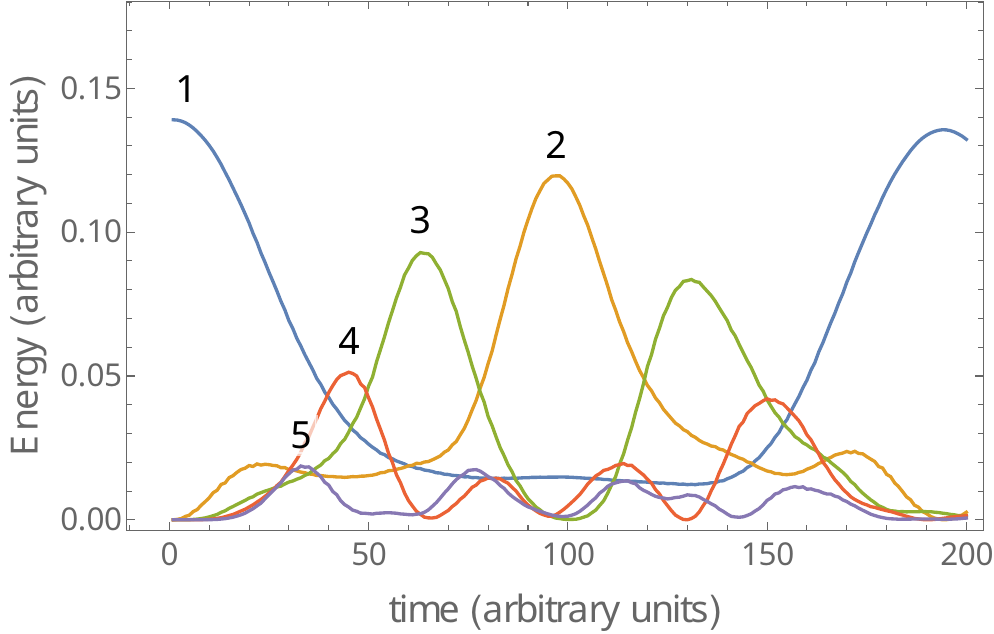}
	\end{center}
	\caption{Energy of the Fourier modes during the time evolution of the FPUT system with $N=64$, periodic boundary conditions and $\alpha=0.25$, $\beta=0.001$. Initial datum is a sine-wave with amplitude $A=0.95$ ($\epsilon=0.0022$). Only the dynamics first 5 Fourier modes are plotted.}\label{fig:FPUOriginal}
\end{figure}

In Figure \ref{fig:FPUOriginal} one sees the time evolution of the Fourier energy modes $E_k$ in time. If $\alpha=\beta=0$, that is, in absence of nonlinearity, $E_k(t)$ would be constant in time. Thanks to the presence of the nonlinearity, the energy, which is originally all in the mode with the lowest Fourier index, is shared among the other normal modes. Indeed, looking at the short-time behavior of $E_k(t)$, one sees that mode 1 ceases energy to the others and, initially mode 2 grows faster than the others. Then, mode 4 overcomes all the others reaching a peak and so on. At the right side of the figure, the systems is almost back to its original state where almost all the energy is stored in mode 1.

\begin{figure}
	\begin{center}
		 \includegraphics[width=0.8 \textwidth]{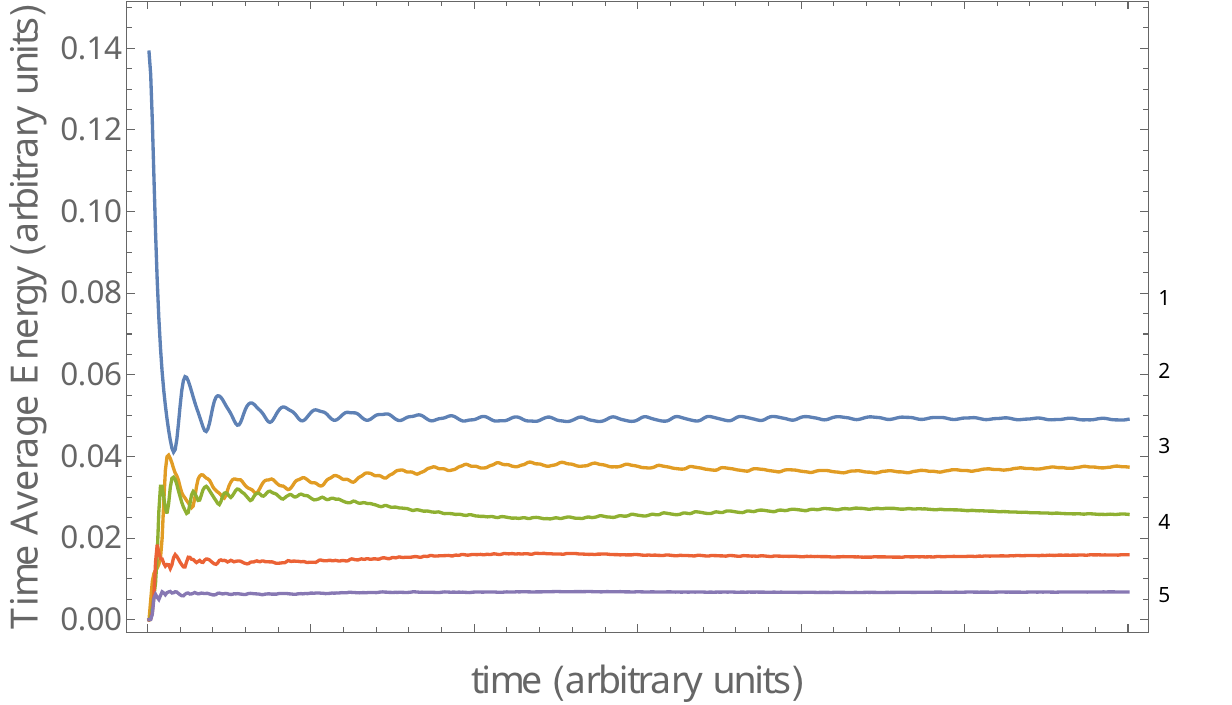}
	\end{center}
	
	\caption{Time averages of the energy of Fourier modes. Observe the convergence to a value which is not equipartition of energy. (Same data of Figure \ref{fig:FPUOriginal}).}\label{fig:EnergyAveragesFPUO}
\end{figure}

Second, in Figure \ref{fig:EnergyAveragesFPUO} we plotted the time-averages
\begin{equation}
	\overline{E}_k(t) \, :=\, \frac{1}{t} \int_0^t E_k(s) \, ds \, .
\end{equation}
These are good indicators of the approach to thermal equilibrium as discussed in subsection~\ref{subsec:StatMechHarmonic}. We observe that time-averages converge to certain values that are different from mode to mode. This is a signature of the fact that the system is \emph{not} at thermal equilibrium, otherwise we would observe equipartition of energy (Lemma \ref{lem:EquipartitionOfEnergy}). 

\begin{quote}
	the results of our computations show features which were, from the beginning, surprising to us. Instead of a gradual, continuous flow of energy from the first mode to the higher modes, all of the problems show an entirely different behavior. 
	[\dots] All our problems have at least this feature in common. Instead of a gradual increase of all the higher modes, the energy is exchanged, essentially, among only a certain few. It is, therefore, very hard to observe the rate of ``thermalization'' or mixing in our problem, and this was the initial purpose of the calculation.
	
	If one should look at the problem from the point of view of statistical mechanics, the situation could be described as follows: the phase space of a point representing our entire system has a great number of dimensions. Only a very small part of its volume is represented by the regions where only one or few out of all possible Fourier modes have divided among themselves almost all the available energy. If our system with nonlinear forces acting between the neighboring points should serve as a good example of a transformation of the phase space which is ergodic or metrically transitive, then the trajectory of almost every point should be everywhere dense in the whole phase space. 

[\dots] Certainly, there seems to be very little, if any, tendency towards equipartition of energy among all degrees of freedom at a given time. In other words, the system certainly do not show mixing.

[\dots]	It is not easy to summarize the results of the various special cases. One feature which they have in common is familiar from certain problems in mechanics of systems with few degrees of freedom. In the compound pendulum problem one has a transformation of energy from one degree of freedom to another and back again, and not a continually increasing sharing of energy between the two. What is perhaps surprising in our problem is that this kind of behavior still appears in systems with, say, 16 or more degrees of freedom.
\end{quote}

\subsection{Further numerical experiments}\label{subsec:FurtherNumerics}

The dynamics of the model, stimulated new numerical experiments. Just to mention a few of significant ones, in the '60s\footnote{Despite the paper being published in 1972, the result was already spread among mathematicians in the '60s as shown by Ulam's preface that is dated 1962 or by Zabusky and Kruskal's paper in 1965 \cite{Zabusky1965} (see also Section \ref{subsec:TravellingFPU}).} Tuck and Tsingou-Menzel \cite{Tuck1972}, observed the phenomena of superrecurrences. 
\begin{quote}
	At first the problem behaved as expected and energy appeared and grew in harmonics of the fundamental, but soon the process became strangely selective and not diffusion-like. At 25 oscillations, the string configuration began to retrace its steps, passing through previous complications in reverse order. By $\sim 50$ oscillations, the complications were unscrambled and the string was back-all but a discrepancy of a few percent-to its half sinusoid starting configuration. [\dots] This behavior was not at all according to statistical-mechanics expectations. Nor it was it Poincar\'e-like; for example, the time for the recurrence of $N$ particles to $1-\varepsilon$ of their starting configuration [\dots] for $N=16$, $\varepsilon=1/100$  amounts to $10^{35}$ oscillations. [\dots]
	
\begin{figure}[h!]
	\begin{center}
		\includegraphics[width=0.8 \textwidth]{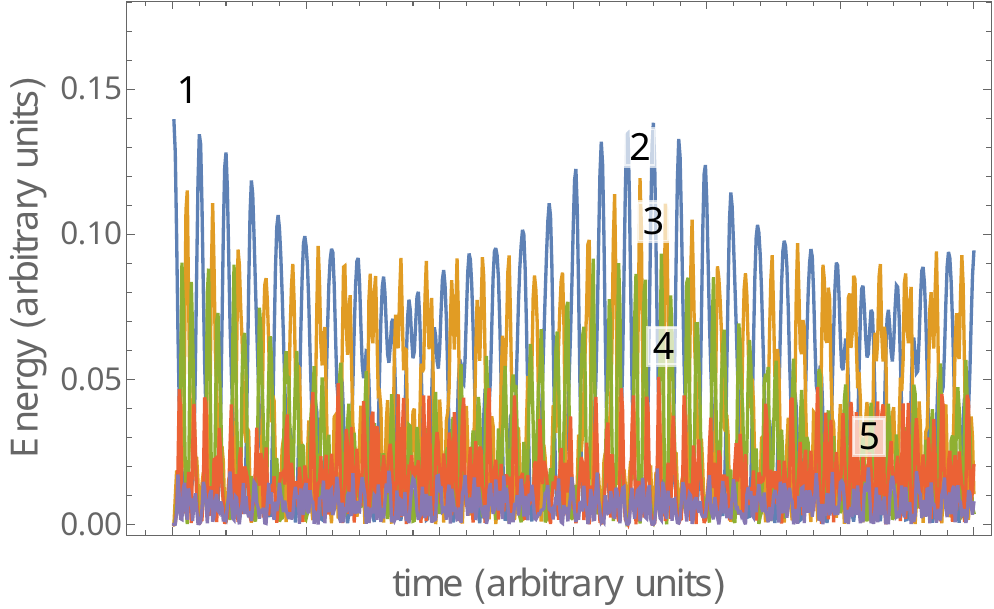}
	\end{center}
	\caption{Same of Figure \ref{fig:FPUOriginal}, but longer time evolution. One sees the phenomenon of \emph{superrecurrences} first observed by J. Tuck and M.~Tsingou-Menzel in \cite{Tuck1972}. (In 1958 Mary Tsingou married Joseph Menzel and took Mary Tsingou-Menzel as married name.)}\label{fig:Superrec}
\end{figure}	
	
	If the conjecture about equipartition was correct, the closing error should mount at each recurrence until it became random. At the next recurrence, the closing error was doubled, and it continued to increase, leveling off at about 390 fundamental oscillations to an amplitude 20\% below the starting amplitude. At the next recurrence, our complacency was jolted, the closing error was less. This trend continued so that by $\sim 780$ cycles, the string was closer than ever before to its starting configuration.
Clearly the recurrence had a superperiod.
\end{quote}

When observing the evolution of $\overline{E_k}(t)$, one sees that they reach a plateau  (Figure \ref{fig:EnergyAveragesFPUO}). This is often named \emph{the FPUT state} and the key question was whether this is a macroscopic state of equilibrium or not.

A breakthrough idea arrives in 1982, with a new point of view on the problem in \cite{Fucito1982}. The novelty of the paper is the idea of \emph{metastability} coming from the theory of spin glasses, that was being developed in Rome during the same years by the group of Parisi. According to this point of view, the \emph{FPUT state} is just an apparently stationary state: at first, the FPUT state is formed and then, on a much longer time-scale, equipartition is reached.

\begin{figure}[h!]
	\begin{center}
		\includegraphics[width=0.6\textwidth]{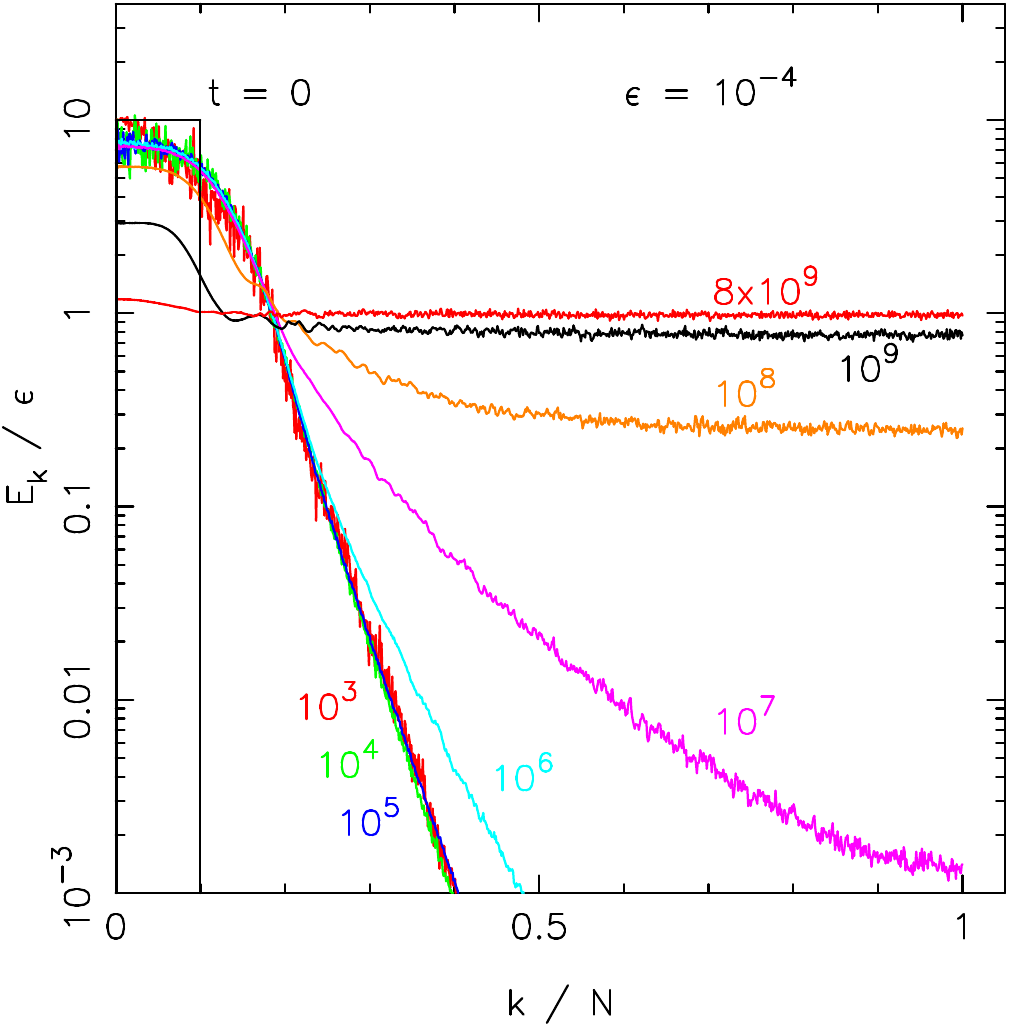}
	\end{center}
\caption{Time averages of energies for FPUT model with $N=1023$ and $\epsilon=10^{-4}$. First phase: times $10^3$-$10^5$ metastable scenario. Second phase: times $10^6$-$10^9$ approach to equipartition of energy. Initially only 10\% of normal modes are excited. Each point is the average over 24 random extraction of the initial phases (from \cite{Benettin2022}).} \label{fig:EvolutionSpectrum}
\end{figure}

The main phenomenon they understood is depicted in Figure \ref{fig:EvolutionSpectrum}. At first, only 10\% of normal modes are excited with the same energy. Within a time of order $\sim 10^3$, the energy is shared among few modes with an exponential cutoff on higher modes. This is the \emph{packet} or \emph{FPUT state}. This packet remains stable up to times of order $10^5$. Only at later times, the tail of high modes started to increase up to the reach of equipartition. 

Therefore, summarizing, the two-scales scenario consists in the co-existence of two mechanisms: first, the quick formation of a packet of low-frequency modes with an exponential tail having an apparently stabilized slope (formation of the metastable state),
and, secondly, the ﬁnal approach to global equipartition through the rising of
the equipartition level of the complementary packet.

Then, a subsequent key observation is the order of magnitude of the width of the packet. Indeed, in \cite{Berchialla2004}, the authors observe that with great precision, the width of the packet is of order $\epsilon^{\frac{1}{4}}$. This latter fact was proved for a certain set of initial data (and in a suitable scaling regime) in Theorem \ref{thm:Bambusi-Ponno}, in particular eq.~\eqref{eq:ExponentialDecay}.

\begin{figure}[h!]
	\begin{center}
		\includegraphics[width=0.8\textwidth]{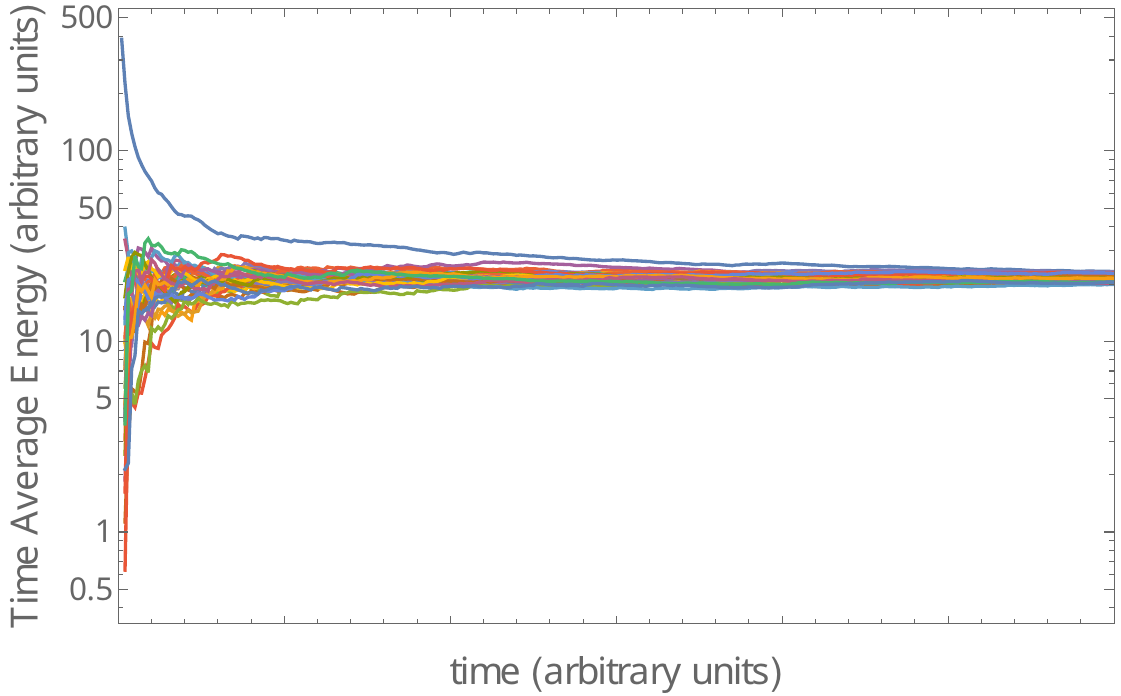}
	\end{center}
	\caption{Time averages of Fourier energies for the FPUT model with initial datum of high energy in logarithmic scale. Observe the quick tendency to equipartition. $N=64$, $\alpha=1$, $\beta=0.1$. Initial datum is a sine-wave with amplitude $A=50.0$ ($\varepsilon=22$).}
\end{figure}

%\begin{figure}[h!]
%	\begin{center}
%		\begin{tikzpicture}
%			\draw[-] (0,0) -- (6.2,0);
%			\draw[-] (0,0) -- (0,5.5);
%			\draw[-] (0.2,0) -- (0.2,0.1); 
%			\draw[-] (0.7,0) -- (0.7,0.1);
%			\draw[-] (1.2,0) -- (1.2,0.1);
%			\draw[-] (1.7,0) -- (1.7,0.1);
%			\draw[-] (2.2,0) -- (2.2,0.1);
%			\draw[-] (2.7,0) -- (2.7,0.1);
%			\draw[-] (3.2,0) -- (3.2,0.1);
%			\draw[-] (3.7,0) -- (3.7,0.1);
%			\draw[-] (4.2,0) -- (4.2,0.1);
%			\draw[-] (4.7,0) -- (4.7,0.1);
%			\draw[-] (5.2,0) -- (5.2,0.1);
%			\draw[-] (5.7,0) -- (5.7,0.1);
%			
%			\draw[-] (0,0.2) -- (0.1,0.2);
%			\draw[-] (0,1.2) -- (0.1,1.2);
%			\draw[-] (0,2.2) -- (0.1,2.2);
%			\draw[-] (0,3.2) -- (0.1,3.2);
%			\draw[-] (0,4.2) -- (0.1,4.2);
%			\draw[-] (0,5.2) -- (0.1,5.2);
%			
%			\node at (3.1,-0.4) {$E$};
%			\node[rotate=90] at (-0.4,2.75) {$\log(t_{\text{equip}})$};
%			
%			\draw[-,thick,blue] (0.3,5) to[out=-80, in=110] (1,1.3);
%			\draw[-,thick,blue] (1,1.3) to[out=-70,in=170] (3.5,0.5);
%			\draw[-,thick,blue] (3.5,0.5) to[out=-10,in=175] (5.7,0.3);
%		\end{tikzpicture}
%	\end{center}
%	\caption{Qualitative behaviour of the equipartition time for FPUT lattice with $N=64$ particles. For low energies, the equipartition time diverges as one expects from Nekhoroshev-type estimates and, in particular, if KAM theorem applies it should not exists. Then there is a higher region where thermalization time has a power-law behaviour with respect to the energy. The original figure is Fig.~8 in \cite{Ponno2011-CHAOS}.}
%\end{figure}

\begin{figure}[h!]
	\begin{center}
		\includegraphics[width=0.8\textwidth]{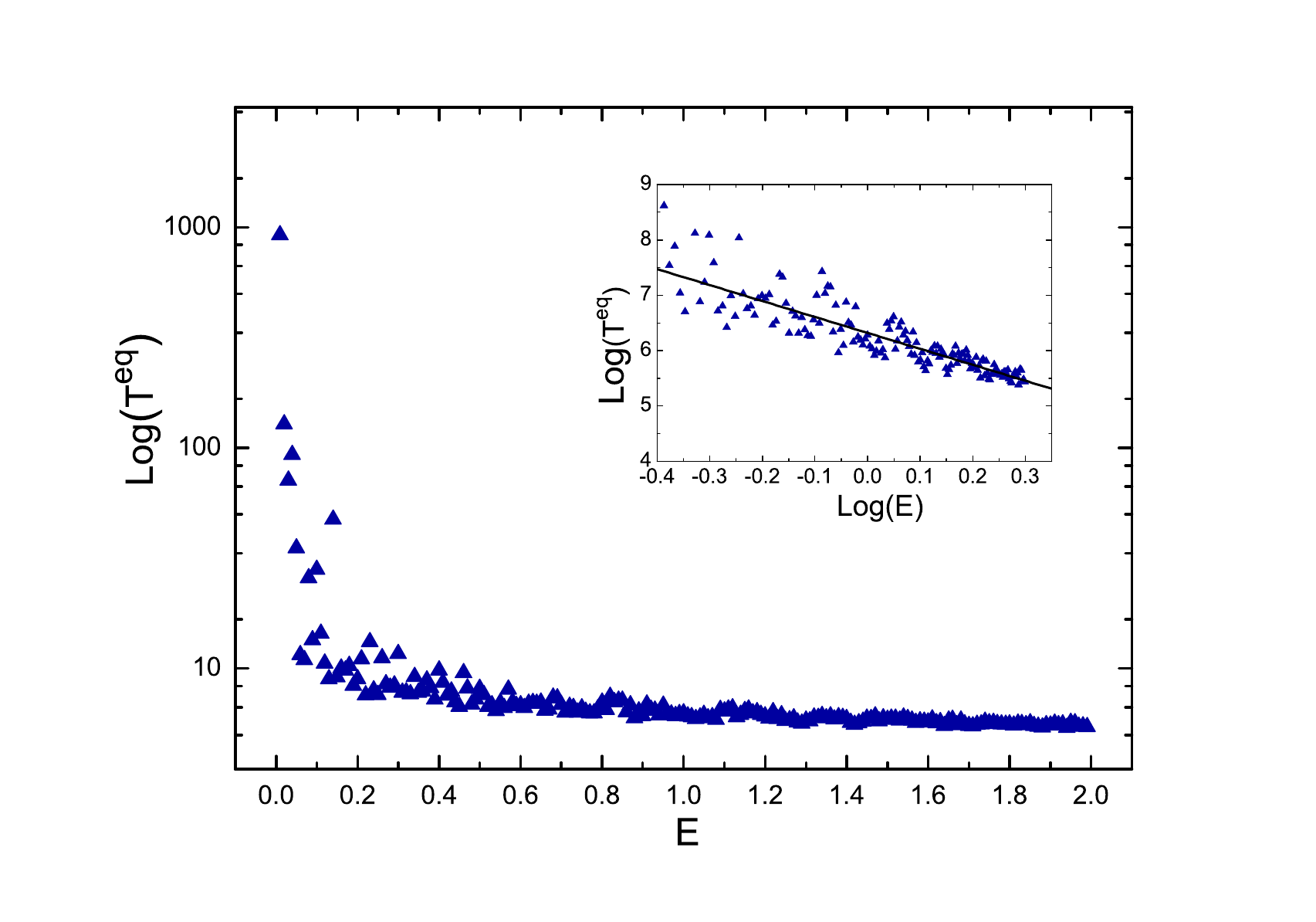}
	\end{center}
	\caption{Estimate of equipartition time for FPUT lattice with $N=64$ particles. For low energies, the equipartition time diverges as one expects from Nekhoroshev-type estimates and, in particular, if KAM theorem applies it should not exists. Then there is a higher region where thermalization time has a power-law behavior with respect to the energy. Figure taken from \cite{Ponno2011-CHAOS}.}
\end{figure}

The analysis, mainly through numerical tools, of the scaling properties of the FPUT state has been the subject of several works, among the others \cite{Benettin2008-1D,Benettin2011}. A convenient quantity to analyze is the \emph{width} $w$ of the state, that is defined in terms of the spectral entropy and of $N_{\mathrm{excited}}$ \eqref{eq:Nex}:
\begin{equation}
	w(t) \, := \, \frac{N_{\mathrm{excited}}(t)}{N} \, .
\end{equation}

Considering initial data like the ones used by Fermi and collaborators, the width can be considered as a function of time $t$, the total number of particles $N$, the specific energy of the system $\epsilon$ and the total number of initially excited modes $w_0$. At the formation of the metastable state, $w$ is weakly depending on $t$ and therefore one has
\begin{equation}
	w(t,N,\epsilon,w_0) \to w_{\mathrm{prethermal}}(N,\epsilon,w_0) \, .
\end{equation}
In fact, $w$ has been observed to satisfy some scaling laws that depend on the choices of the initial phases of the modes in $w_0$. 
\begin{itemize}
	\item[(a)] When only one normal mode is excited or for random initial phases, one has
	\begin{equation}
		w(t,N,\epsilon,w_0) \, = \, \epsilon^{\frac{1}{4}} \mathfrak{w}(\epsilon^{\frac{3}{8}} w_0^{\frac{3}{2}} t)
	\end{equation}
	and $\mathfrak{w}$ is a suitable function of a single variable, with a sigmoid-like profile. 
	\item[(b)] When coherent pattern of phases are excited, e.g.~all phases set to $\frac{\pi}{2}$ or equally spaced, the scaling is
	\begin{equation}
		w(t,N,\epsilon,w_0) \, = \, (E w_0)^{\frac{1}{4}} \tilde{\mathfrak{w}}((E w_0)^{\frac{3}{8}} w_0^{\frac{3}{2}} t) \, ,
	\end{equation} 
	where $\tilde{\mathfrak{w}}$ is, again, a sigmoid-like function different from ${\mathfrak{w}}$.
\end{itemize}
In particular, note that, formally the scalings are the same apart from replacing $\epsilon$ with $E w_0$.

\subsection{Relation with KAM and Nekhoroshev}
To the author's knowledge, the first attempt to explain the FPUT paradox using the persistence of quasi-periodic motions guaranteed by KAM theorem is due to Izrailev and Chirikov \cite{Izrailev1966StatisticalPO}. 

The authors have in mind a two-step situation: either the energy is low enough, and then KAM theorem applies or the energy is high enough and a sort of stochasticity allows the system to reach equipartition. 

Nevertheless, as it was pointed out in \cite{Ford1992} and \cite{Weissert1997}, considering the FPUT Hamiltonian as a perturbation of the harmonic chain in harmonic action-angles variable violates hypothesis (ii) in Theorem \ref{thm:KAM}. 

The first attempt to overcome this difficulty is due to Nishida in 1971 \cite{Nishida71} and Sanders in 1977 \cite{sanders1978theory} who tried to perform a preliminary step of normal form in the FPUT Hamiltonian and then to use the Birkhoff normal form as unperturbed Hamiltonian to apply KAM theorem. Let us discuss, for simplicity, the idea for the $\beta$-model. In harmonic action-angle coordinates, the Hamiltonian of the $\beta$-FPUT chain is
\begin{equation}
	\begin{split}
		H_{\mathrm{F}}(\varphi,I) \, = \, &\sum_{j=1}^N \omega_j I_j \\
		&+ \, \sum_{j_1,j_2,j_3,j_4} \sqrt{I_{j_1} I_{j_2} I_{j_3} I_{j_4}} f_{j_1,j_2,j_3,j_4}(\varphi_{j_1},\varphi_{j_2}, \varphi_{j_3}, \varphi_{j_4}) \, .
	\end{split}
\end{equation}
where $f$ is an explicit function. Then, in the framework of perturbation theory one can hope to find a function $G_1$ such that $\tilde{H}_{\mathrm{F}}(\tilde{\varphi}, \tilde{I}) = H_{\mathrm{F}} \circ \Phi_{G_1}^1(\tilde{\varphi},\tilde{I})$ is of the form
\begin{equation}
	\tilde{H}_F(\tilde{\varphi},\tilde{I}) \, = \sum_{j=1}^N \omega_j \tilde{I}_j + \sum_{j_1,j_2} \Omega_{j_1,j_2} \tilde{I}_{j_1} \tilde{I}_{j_2} + O(\tilde{I}^3) \, .
\end{equation}
for some explicit constants $\Omega_{j_1,j_2}$ and for which, considering 
\begin{equation}
	h_0(\tilde{I}) \,:=\, \sum_{j=1}^N \omega_j \tilde{I}_j + \sum_{j_1,j_2} \Omega_{j_1,j_2} \tilde{I}_{j_1} \tilde{I}_{j_2} \, ,
\end{equation}
one has that $h_0$ satisfy condition (ii) in Theorem \ref{thm:KAM}. If this was the case, then KAM theorem can be applied and justifies the quasi-periodic observations of FPUT. 

Nishida, in his 1971 paper, considers the $\beta$ model with fixed boundary conditions. Assuming that the linear frequencies of the lattice are non-resonant, the first step of normal form can be done and he computes the normal form showing that it satisfies (ii) in Theorem \ref{thm:KAM}. He does not prove the fact, which is in general false, that frequencies are non-resonant. He refers to an unpublished result by Izumi which seems to contain the proof that resonances do not occur if either the number of particles in the lattice is prime or if it is a power of $2$. The result by Izumi is not available and, nevertheless, has been proved by Hemmer in 1959 \cite{hemmer1959dynamic}. 

Applicability of KAM, at that point, was left to the control of resonances. This latter point has been solved by Rink in \cite{Rink2001,Rink2005} for the $\beta$ model and by Henrici and Kappeler for the $\alpha+\beta$ model \cite{Henrici2007,Henrici2009}. 
To summarize the results, we denote by $U_R(I)$ the ball of radius $R>0$ centered at $I$, that is
\begin{equation}
	U_R(I)\,=\, \Big\{I' \in \mathbb{R}^N \, \big| \, \sup_{j=1,\cdots,N} |I_j-I'_j| \leq R\Big\} \, .
\end{equation}

\begin{theorem}[Nishida '71, Rink '06, Henrici-Kappeler '07, '09]
	For any $N>0$, $\alpha \in \mathbb{R}$ and $\beta \neq 0$, there exists $\varepsilon_N>0$ such that, for any $0<\varepsilon<\varepsilon_N$ there exists a set $B_\varepsilon \subset U_{\varepsilon}$ with positive measure such that $U_\varepsilon(0) \setminus B_\varepsilon$ has positive measure and the following holds
	\begin{itemize}
		\item[(a)] $B_\varepsilon\times \mathbb{T}^N$ is invariant under the flow of FPUT;
		\item[(b)] If $I_0 \in B_\varepsilon$ the motion is quasi-periodic and fills densely a manifold which is diffeomorphic to a $N$ dimensional torus.
	\end{itemize}
\end{theorem}

This theorem establishes non-ergodicity of the FPUT chain for initial data with energy low enough. Moreover, it states that when the energy is low enough, for a large set of initial data, the dynamics consist in the FPUT recurrence for all times.

 Nevertheless, its meaningfulness in the thermodynamic limit remains an open problem. In other words, how does the threshold $\varepsilon_N$ behave with $N$ for large $N$? A very quick decay of $\varepsilon_N$ with $N$ as $N\to \infty$ would imply that KAM theorem is meaningless in the statistical mechanics setting, thus leaving the possibility of thermalization for the FPUT chain in the limit. On the contrary, if $\varepsilon_N>0$ in the limit, then this would imply that - even in the thermodynamic limit - KAM theorem is applicable and for low energy initial data, the system would not be ergodic.

\section{FPUT, Toda and KdV hierarchy}\label{sec:FPUT-Integrable}\label{sec:FPUTIntegrable}
In this section we develop the point of view on the FPUT problem as a perturbation of certain integrable models. Indeed, integrable approximations, with the methods of Hamiltonian perturbation theory, give a quantitative and sometimes mathematically rigorous description of the phenomenon in the thermodynamic limit. We will first discuss the relation with the integrable Toda chain, referring to the recent review \cite{Benettin2022} for many details. 

A larger part is devoted to the relation with the Korteweg-de Vries equation, for which a complete account of recent results have not been collected yet. 

\subsection{Connection with Toda} In the late '60s and early '70s it was discovered the existence of a one-dimensional integrable non-linear lattice by Toda \cite{Toda1967,Toda1967-B,Toda1970} and by H\'enon \cite{Hnon1974}. The system, known as \emph{Toda lattice}, has Hamiltonian
\begin{equation}
	H_{\mathrm{T}}(q,p) \, = \, \sum_{j \in \mathbb{Z}_N} \Big[\frac{p_j^2}{2}+\Phi_{\mathrm{T}}(q_{j+1}-q_j) \Big]
\end{equation}
where
\begin{equation}\label{eq:TodaPotential}
	\Phi_{\mathrm{T}}(z) \, = \, A (e^{-B z}+Bz-1 )
\end{equation}
for $A$ and $B$ real free parameters. It was first noted by Manakov \cite{manakov1974complete} that it is more convenient to look at the FPUT system as a perturbation of the Toda system then of a Harmonic chain. Indeed, choosing suitably $A=\frac{1}{4 \alpha^2}$ and $B=2 \alpha$ in \eqref{eq:TodaPotential}, one has
\begin{equation}
	\Phi_{\mathrm{F}}(z)-\Phi_{\mathrm{T}}(z) \, = \, \Big(\beta-\frac{2}{3}\alpha^2 \Big) \frac{z^4}{4} + O(z^5).
\end{equation}
One thus defines the parameter $\beta_T := \frac{2}{3}\alpha^2$, where the subscript ``T'' stands for ``Toda''.

In a remarkable numerical work \cite{Ferguson1982} the authors provide a method to compute Toda actions and numerically compute the evolution of Toda actions along the FPUT flow. They observe that 
\begin{quote}
	Toda action variables [\dots] are nearly constant; this should be compared with the large variation of the energies $E_k$, numbers which are proportional to the action variables of the Harmonic chain. [\dots]
	
\begin{figure}[h!]
	\begin{center}
		\includegraphics[width=0.48\textwidth]{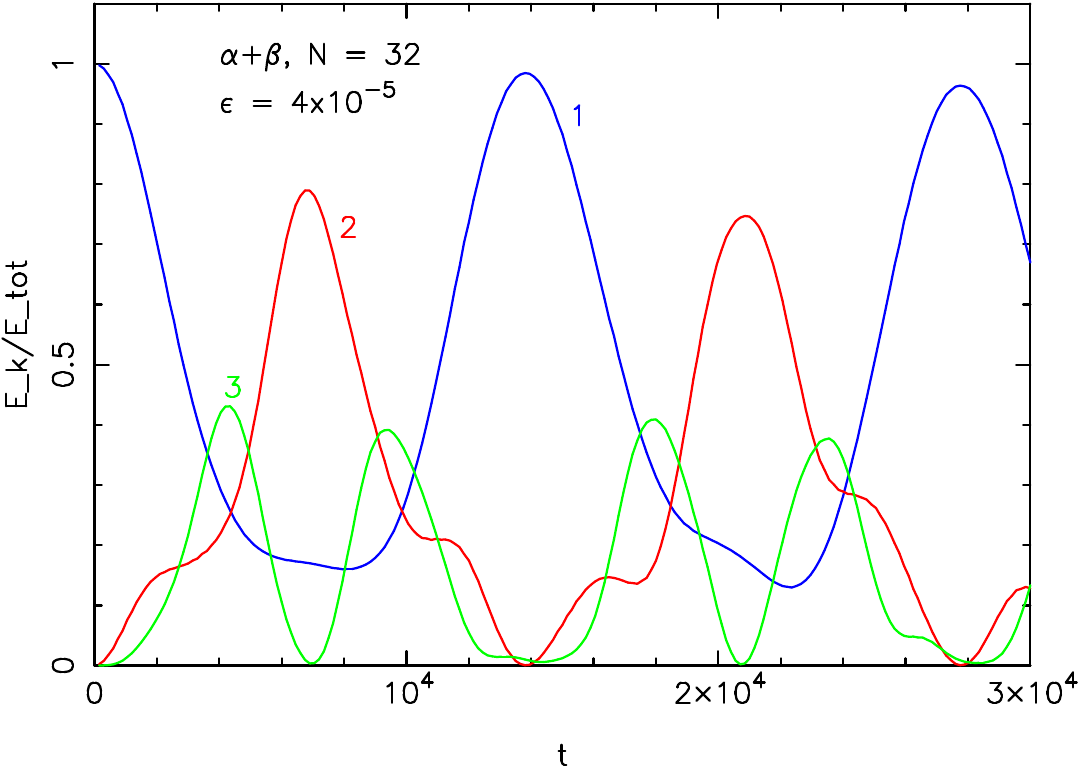} \includegraphics[width=0.48\textwidth]{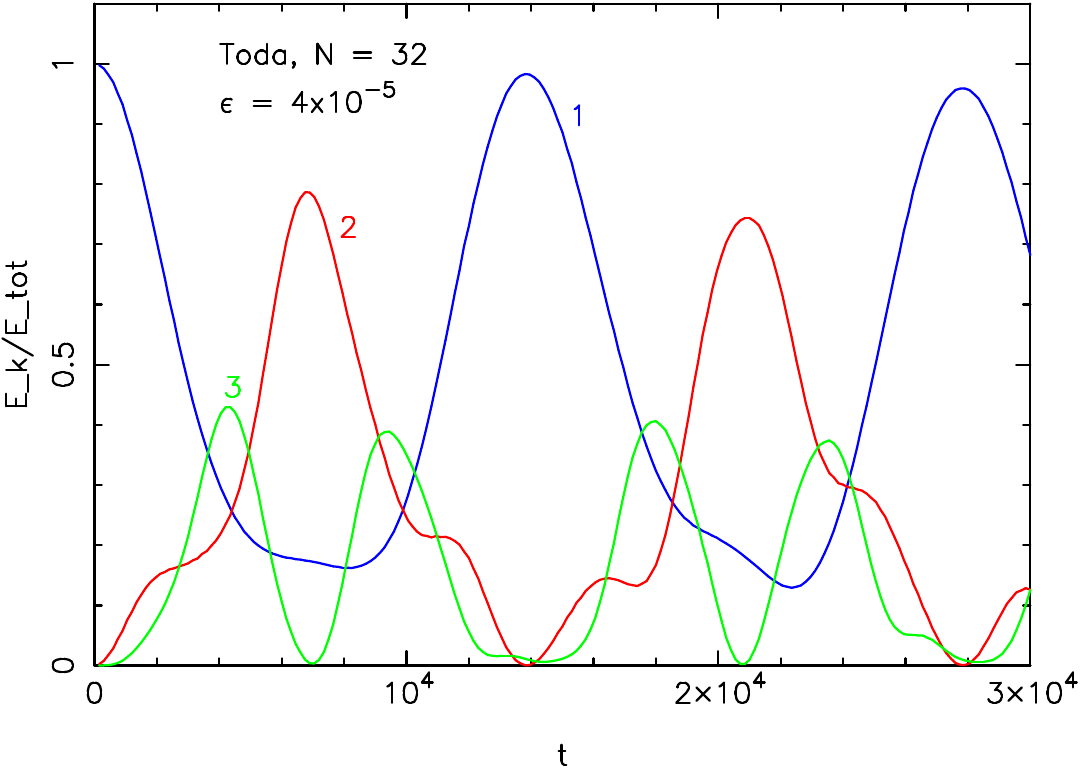}
	\end{center}
	\caption{Comparison between FPUT (left) and Toda (right) dynamics with the same initial datum. From \cite{Benettin2022}.}
\end{figure}	
	
	The fact that the Toda action variables remain nearly constant during the integration of the quadratic chain has important ramifications. We must first realize that the Toda action variables will in general change with time, since we are not integrating the equations of motion for the Toda chain. Recall that each phase point $(q, p)$ describing the state of the quadratic chain lies on a (Toda) torus which is characterized by the corresponding graph of the discriminant. Therefore we can visualize the trajectory of the phase point of the quadratic lattice as a curve winding around a torus which is changing with time.
\end{quote}

This point of view is the basis of a series of works by Benettin and Ponno that are summarized in the beautiful review \cite{Benettin2022} to which I devote the reader for a detailed discussion. We limit ourselves in mentioning and quickly commenting the following two analytic results.

\begin{theorem}[Bambusi-Maspero \cite{Bambusi-Maspero-2016}]\label{thm:BamMas16}
	Fix $s\geq 1$, $\sigma \geq 0$; then there exist positive constants $R_0'$, $C_2$, $T$, such that the following holds true. Consider a real initial datum with
	\begin{equation}
		\frac{E_1(0)}{N} \, = \, \frac{E_{N-1}(0)}{N} \, = \, R^2 e^{-2 \sigma} \mu^4 \, , \qquad E_k(0)=0 \, \qquad k \neq 1,N-1 \, .
	\end{equation}
	with $R < R_0'$. Then, along the corresponding solution of the FPUT system, one has
	\begin{equation}
		\frac{E_k(t)}{N} \leq \frac{16 R^2 \mu^4 e^{-2\sigma k}}{k^{2s}} \, , \qquad \forall 1 \leq k \leq \lfloor {\textstyle\frac{N}{2}} \rfloor 
	\end{equation}
	and for times 
	\begin{equation}\label{eq:66Toda}
		 |t| \leq \frac{T}{R^2 \mu^4} \frac{1}{|\beta-\beta_T|+C_2 R \mu^2} \, .
	\end{equation}
\end{theorem}

One can also prove a theorem for initial data extracted with Gibbs measure at inverse temperature $b$. Here, we do not use $\beta$ to avoid confusion with the $\beta$ parameter of the FPUT chain. Thus, the Gibbs measure is
\begin{equation}
	d \mu_b \, := \, \frac{1}{Z_{\mathrm{can}}(b)} e^{-bH(q,p)} \delta \big({\textstyle \sum_{j=0}^{N-1}(q_{j+1}-q_j)}\big) \, \delta\big( {\textstyle \sum_{j=0}^{N-1} p_j}\big) dq dp
\end{equation}
where $Z_{\mathrm{can}}(b)$ is the usual canonical partition function (see Definition \ref{def:CanonicalPF})
\begin{equation}
	Z_{\mathrm{can}}(b) \, = \, \int_{\mathbb{R}^N \times \mathbb{R}^N} e^{-bH(q,p)} \delta \big({\textstyle \sum_{j=0}^{N-1}(q_{j+1}-q_j)}\big) \, \delta\big( {\textstyle \sum_{j=0}^{N-1} p_j}\big) dq dp \, .
\end{equation}

We will consider H\'enon's integrals of motion that are obtained as traces of powers of a certain (Lax-)matrix \cite{Hnon1974}. We will denote those quantities as $J^{(m)}$. For concreteness, the first few H\'enon integrals are
\begin{eqnarray}
	J^{(2)}(q,p) \!\!\!\!&=&\!\!\!\!\sum_{j=1}^{N} \left(\frac{p_j^2}{2}+e^{-(q_{j+1}-q_j)} \right) \\
	J^{(3)}(q,p)\!\!\!\! &=&\!\!\!\! -\sum_{j=1}^{N} \left(\frac{1}{3}p_j^3+(p_j+p_{j+1})e^{-(q_{j+1}-q_j)} \right)
\end{eqnarray}
and from their expressions, we see that they are extensive quantities (that is, they are expected to grow linearly with $N$). Despite this fact, one can estimate the difference $|J^{(m)} \circ \Phi_{H_F}^t - J^{(m)}|$ with the variance using Chebyshev inequality and prove (Proposition 5.1 in \cite{Grava2020}) that the latter quantity is of order $O(\sqrt{N} \beta^{-1})$.

\begin{theorem}[Grava-Maspero-Mazzucca-Ponno \cite{Grava2020}]\label{thm:GMMP}
	Fix $m \in \mathbb{N}$. There exists constants $N_0,b_0,C_0,C_1,C_2>0$ (all depending on $m$), such that for any $N>N_0$, $b>b_0$ and any $\delta_1,\delta_2>0$ one has
	\begin{equation}
		\mathrm{Prob}_{\mu_b} \left(\frac{1}{\sqrt{N}}|J^{(m)} \circ \Phi_{H_{\mathrm{F}}}^t-J^{(m)}| > C_2 \frac{\delta_1}{\beta} \right) \leq \delta_2 C_0
	\end{equation}
	for every time $t$ fulfilling
	\begin{equation}\label{eq:610Toda}
		|t| \, \leq \, \frac{\delta_1 \sqrt{\delta_2}}{((\beta-\beta_T)^2+C_1 b^{-1})^{\frac{1}{2}}}b \, .
	\end{equation}
\end{theorem}
%the authors prove also that (Proposition 5.6 in \cite{Grava2020})
%\begin{equation}
%	\sigma_{J^{(m)}} \geq C \frac{\sqrt{N}}{b^{\frac{3}{2}}} \, .
%\end{equation}
%\textcolor{red}{Non capisco come fa ad essere intensiva quella roba là.}
The first reason for which the result is interesting is that it holds in the thermodynamic limit and shows that Toda integrals of motion are \emph{actually} almost-conserved quantities. Second, and -- from a certain point of view -- even more important, this result shows that it is quantitatively convenient to look at FPUT as a perturbation of the Toda model even in the thermodynamic limit. Indeed, when the tangency with the Toda chain is increased, also the times for which the Toda integrals of motion are adiabatic invariants increase (this last fact is evident when, in \eqref{eq:610Toda} or \eqref{eq:66Toda} one puts $\beta=\beta_T$).

Then, the only drawback is that Theorem \ref{thm:BamMas16} does not hold in the thermodynamic limit and Theorem \ref{thm:GMMP} holds for \emph{a finite number of Toda integrals of motion}. While for Theorem \ref{thm:BamMas16} there are serious reasons to believe it is optimal, The finite amount of Toda integrals of motion in Theorem \ref{thm:GMMP} is expected to be a technical limitation.

\subsection{KdV as almost-one-directional traveling waves}\label{subsec:TravellingFPU} In this section, we revisit the derivation of Zabusky and Kruskal of the KdV from the FPUT lattice \cite{Zabusky1965}. It is worth noting that, the analysis of exact solutions for nonlinear string is actually motivated by the FPUT problem (see \cite{Zabusky1962}). Then, we will extend their derivation to higher order KdV equations, a recent result obtained in \cite{Gallone2021}.

We start by introducing two \emph{physically dimensionless} analytic functions $P,Q: \mathbb{T} \times \mathbb{R} \to \mathbb{R}$ such that
\begin{equation}\label{eq:InterpolatingFields}
	\begin{split}
		p_j(t) \,&=\, \sqrt{\epsilon} P(x,\tau) \,\big|_{x=hj, \tau=ht}  \, ,\\
		q_j(t) \,&=\, \frac{\sqrt{\epsilon}}{h} Q(x,\tau) \big|_{x=hj,\tau=ht} \, .
	\end{split}
\end{equation}
Note that, from the equations of motion of the FPUT system, $\sum_{j=1}^N p_j$ is a constant of motion of the system. It is convenient to work with the condition
	$\sum_{j=1}^N p_j=0$, which translates into the condition of zero-average for $P$:
\begin{equation}\label{eq:ZeroAverageP}
	\int_{\mathbb{T}} P(x) \, dx \,=\, 0 \, .
\end{equation}

Inserting \eqref{eq:InterpolatingFields} in the FPUT equations of motion
\begin{equation}\label{eq:HamEqFPUT}
	\begin{split}
		\dot{p}_j(t) \, &= \, -\frac{\partial H}{\partial q_j}(q(t),p(t)) \, , \\
		\dot{q}_j(t) \, &= \, p_j(t) \, ,
	\end{split}
\end{equation}
we have
\begin{eqnarray}
	\dot{p}_j(t)\!\!\!\!&=&\!\!\!\!  h\sqrt{\epsilon} P_\tau(x,\tau) \big|_{x=hj, \tau=ht} \, , \label{eq:PdotC}\\
	\dot{q}_j(t)\!\!\!\!&=&\!\!\!\! \sqrt{\epsilon} Q_\tau(x,\tau) \big|_{x=hj, \tau=ht} \, , \label{eq:QdotC}\\
	-\frac{\partial H}{\partial q_j}(q(t),p(t)) \!\!\!\!&=&\!\!\!\! \Phi'\big(q_{j+1}(t)-q_j(t)\big) - \Phi'\big(q_j(t)-q_{j-1}(t)\big) \, . \label{eq:PhiPrimoPrimo}
\end{eqnarray}
The right-hand side of \eqref{eq:PhiPrimoPrimo} has a linear term, that is
\begin{equation}
	\begin{split}
		q_{j+1}(t)+&q_{j-1}(t)-2 q_j(t)  \\ &= \, \frac{\sqrt{\epsilon}}{h} \big[ Q(x+h,\tau)+Q(x-h,\tau)-2 Q(x,\tau) \big] \big|_{x=hj, \tau=ht} \\
		&=\frac{2 \sqrt{\epsilon}}{h} \sum_{r=1}^{+\infty} \frac{h^{2r}}{(2r)!} \partial_x^{2r} Q(x,\tau) \big|_{x=hj, \tau=ht}
	\end{split}
\end{equation}
where in the last step we used Taylor series. Keeping only terms formally of order at most $h^3$, we get
\begin{equation}\label{eq:LinearTermApproximated}
	q_{j+1}(t)+q_{j-1}(t)-2 q_j(t)  \, \cong \, \sqrt{\epsilon} h Q_{xx} + \frac{\sqrt{\epsilon} h^3}{12} Q_{xxxx} + \cdots \, .
\end{equation}
Concerning the quadratic term, we can write
\begin{equation}
	\begin{split}
		\alpha \big[&(q_{j+1}(t)-q_j(t))^2-(q_{j}(t)-q_{j-1}(t))^2 \big] \\
		&=\frac{\alpha \epsilon}{h^2} \left[ \left( \sum_{r=1}^{+\infty} \frac{h^r}{r!} \partial_x^r Q(x,\tau) \right)^2-\left(\sum_{r=1}^{+\infty} (-1)^r \frac{h^r}{r!} \partial_x^r Q(x,\tau) \right)^2 \right] \Big|_{x=hj,\tau=ht} \, .
	\end{split}
\end{equation}
Keeping only terms of order at most $\epsilon h$, we get
\begin{equation}\label{eq:CubicTermApproximated}
	\begin{split}
		\alpha \big[(q_{j+1}(t)-q_j(t))^2-&(q_{j}(t)-q_{j-1}(t))^2 \big]\, \\
		& \cong \, 2 \alpha \epsilon h Q_x(x,\tau) Q_{xx}(x,\tau) \big|_{x=hj,\tau=ht} \, .
	\end{split}
\end{equation}
Inserting first \eqref{eq:CubicTermApproximated}, \eqref{eq:LinearTermApproximated} in \eqref{eq:PhiPrimoPrimo}, then \eqref{eq:PhiPrimoPrimo}, \eqref{eq:PdotC}, \eqref{eq:QdotC} and \eqref{eq:InterpolatingFields} in \eqref{eq:HamEqFPUT} we get
\begin{equation}\label{eq:ApproximateSystemWaves}
	\begin{split}
		Q_\tau \,&=\, P \, ,\\
		P_\tau \, &= \, Q_{xx}+\frac{h^2}{12} Q_{xxxx} + 2 \alpha \sqrt{\epsilon} Q_x Q_{xx} \, .
	\end{split}
\end{equation}
Inserting the first equation in the second one, we see that $Q$ satisfies a \emph{Boussinesq equation}
\begin{equation}
	Q_{\tau \tau} \,=\, Q_{xx} + \frac{h^2}{12} Q_{xxxx} + 2 \alpha \sqrt{\epsilon} Q_x Q_{xx} \, ,
\end{equation}
which is a nonlinear perturbation of the wave equation. (Let us remark, at this point, that Zakharov gave an heuristic argument to interpret the long thermalization of FPUT in terms of the Boussinesq equation \cite{Zakharov1974}.) 

It is standard to split the flow of the wave equation into left- and right-traveling waves. This procedure, dates back at least to Riemann (see Chapter VIII, paragraph 2 in \cite{RiemannBook}) and, in fact, these are sometimes referred to as \emph{Riemann invariants}. We thus define
\begin{equation}\label{eq:LeftAndRightTV}
	\begin{split}
		L \,&:= \, \frac{Q_x +P}{\sqrt{2}} \, , \qquad \text{(Left-travelling wave)} \\
		R \, &:= \, \frac{Q_x-P}{\sqrt{2}} \, , \qquad \text{(Right-travelling wave)} \, .
	\end{split}
\end{equation}
At a linear level, the names for $R$ and $L$ are justified by the fact that this transformation maps the wave equation into a pair of transport equations
\begin{equation}
	\left\{ \begin{array}{l}
		Q_\tau=P \\
		P_\tau = Q_{xx}
	\end{array} \right. \, \qquad \longleftrightarrow
	\qquad \left\{ \begin{array}{l} 
	L_\tau = L_x \\
	R_\tau = - R_x 
	\end{array} \right. \, ,
\end{equation}
and their solution, that can be computed with the method of characteristics, consists in left- (or, respectively, right-) translations of the initial datum:
\begin{equation}
	L(x,\tau) \, = \, L_0(x+\tau) \, , \qquad R(x,\tau) \, = \, R_0(x-\tau) \, ,
\end{equation}
where $L_0$ and $R_0$ are the initial data.  Note that, due to \eqref{eq:ZeroAverageP} and the fact that $Q_x$ has zero-average on the torus, also $L$ and $R$ have zero-average.

Inserting \eqref{eq:LeftAndRightTV} into the system \eqref{eq:ApproximateSystemWaves}, one obtains
\begin{equation}\label{eq:EquationsLandR}
	\begin{split}
		L_\tau \,&=\, L_x + \frac{h^2}{24} (R_{xxx}+L_{xxx})+\frac{\alpha \sqrt{\epsilon}}{\sqrt{2}}(R+L)(R_x+L_x) \,  \\
		R_\tau\,&=\, -R_x - \frac{h^2}{24}(R_{xxx}+L_{xxx})-\frac{\alpha \sqrt{\epsilon}}{\sqrt{2}} (R+L)(R_x+L_x) \, .
	\end{split}
\end{equation}
Note that equations for $R$ are obtained from the equations for $L$ by exchanging $L \leftrightarrow R$ and by changing sign. This is the manifestation of the \emph{time-reversal symmetry} of the model. 

In the $L$ and $R$ variables, equations for $L$ and $R$ are coupled. If we assume that one of the two fields is small with respect to the other, e.g.~$R \ll 1$, we obtain 
\begin{equation}\label{eq:KdVForL}
	L_\tau \,=\, L_x + \frac{h^2}{24} L_{xxx} + \frac{\alpha \sqrt{\epsilon}}{\sqrt{2}} L L_x
\end{equation}
which is the Korteweg-de Vries equation.

Despite of being purely formal and there should be a lot to discuss about the reminder and the reduction from the lattice to the continuum, this derivation has a lot of interesting features. 

First, in the original FPUT report \cite{FPU55} , the authors state that the discretization is a model of a \emph{one dimensional continuum}. Then there is a series of attempts to explain the FPUT paradox by analyzing the nonlinear equation arising from a string whose discretization would yield a FPUT-like dynamical system, among them we mention the works by Zabusky and Kruskal \cite{Zabusky1962,Kruskal1964} and the one by Lax \cite{Lax1964} that culminated in the numerical simulation of the dynamics of the Korteweg-de Vries equation in \cite{Zabusky1965}. In this latter paper, the authors observe numerically the presence of ``a nonlinear physical process in which interacting localized pulses do not scatter irreversibly.''

Then, the authors say
\begin{quote}
	we should emphasize that at $T_R$ [that is, the \emph{recurrence time}] all the solitons arrive almost in the same phase and almost reconstruct the initial state through the nonlinear interaction. This process
proceeds onwards, and at $2T_R$ one again has
a ``near recurrence'' which is not as good as
the first recurrence. Tuck,
at the Los Alamos Scientific Laboratories, observed this phenomenon as well as eventual ``superrecurrences'' in calculations for a similar problem. We can understand these phenomena in terms of soliton interactions. For $t >T_R$ the successive focusings get poorer due to solitons arriving more and more out of phase with each other and then eventually gets better again when their phase relationship changes. Furthermore, because the solitons are remarkably stable entities, preserving their identity through numerous interactions, one would expect this system to exhibit thermalization (complete energy sharing among the corresponding linear normal
modes) only after extremely long times, if ever.
\end{quote}

Thus, summarizing, the presence of solitons, i.e.~extremely stable solitary waves, is given as a possible justification of the FPUT recurrence and the ``superrecurrence'' observed by Tuck and Tsingou-Menzel in \cite{Tuck1972} (see Figure \ref{fig:Superrec}). 

Continuing in the investigation of nonlinear waves, in 1967 Gardner, Greene, Kruskal and Miura found a method to solve the Korteweg-de Vries equation \cite{Gardner1967}, opening the avenue to the analysis of \emph{infinite-dimensional} integrable systems. 

At present, it is well-known that KdV is a Hamiltonian PDE, whose Hamiltonian functional is
\begin{equation}
	\mathscr{H}_{\mathrm{KdV}_3}(U) \, = \, \int \left[-\frac{a}{2} (U_{x})^2+\frac{b}{6}U^3 \right] \, dx
\end{equation}
with $a,b \in \mathbb{R}$ and with Poisson tensor $J_K=\partial_x$. Being integrable means that KdV has \emph{infinitely many} constants of motion that are in involution among them. Also, using Lax-Magri recurrence (that has its deepest roots in the bi-Hamiltonian structure of the Korteweg-de Vries equation) it is possible to construct recursively all of them. The first non-trivial one, after $\mathscr{H}_{\mathrm{KdV}_3}$, is
\begin{equation}
	\mathscr{H}_{\mathrm{KdV}_5}(U) \, = \, \int  \Big[\frac{5}{36} b^2 U^4 + \frac{5}{6} b a U^2 U_{xx} + a^2 (U_{xx})^2 \Big]\, dx
\end{equation}
whose vector field is the KdV\textsubscript{5} equation, given by $U_t = J_K \nabla_{L^2} \mathscr{H}_{\mathrm{KdV}_5}$
\begin{equation}
	U_t \, = \, \frac{5}{3} b^2 U^2 U_x + \frac{20}{3} a b U_x U_{xx} + \frac{10}{3} a b U U_{xxx} + 2 a^2 U_{xxxxx} \, .
\end{equation}
The collection of all the equations KdV\textsubscript{n}, for $n$ odd is often referred to as the \emph{Korteweg-de Vries hierarchy}. The first terms of the hierarchy are
\begin{equation}\label{eq:KdVHierarchy}
	\begin{split}
		U_t \, &=\, U_x \\
		U_t \, &=\, b U U_x +a U_{xxx} \\
		U_t \, &=\,  \frac{5}{3} b^2 U^2 U_x + \frac{20}{3} a b U_x U_{xx} + \frac{10}{3} a b U U_{xxx} + 2 a^2 U_{xxxxx}
	\end{split}
\end{equation}

Second, we discuss the parameters appearing in the equation \eqref{eq:KdVForL}. Indeed, we see that the nonlinear term is proportional to $\alpha \sqrt{\epsilon}$. This is a signature of the fact that $\alpha$, once different from zero, is actually an irrelevant parameter in the analysis of the FPUT problem and its sole role is to set the scale of energies. Indeed, keeping $\epsilon$ as free parameter, one can introduce a new $\epsilon'=\alpha^2 \epsilon$ and then the FPUT dynamics would depend only on $\epsilon'$ and new parameters $\beta'=\frac{\beta}{\alpha^2}$, $\gamma'=\frac{\gamma}{\alpha^3}$, $\cdots$. 

We can think of this ``KdV approximation'' of the FPUT system to hold for any choice of the parameters $\epsilon$ or $h$. Nevertheless we can see that different relative scalings of the limit can be dramatic. Indeed, if $\sqrt{\epsilon} \lesssim h^2$, that is 
\begin{equation}
	\epsilon \, \lesssim \frac{1}{N^4} \, ,
\end{equation}
then in the KdV equation the nonlinear and the dispersive term weight ``the same'' or, at most, the dispersive term wins and therefore there is no problem of global well-posedness for the approximating equation. If, on the contrary, we let $\sqrt{\epsilon} \gg h^2$, we can formally neglect the dispersive term and we would obtain a inviscid Burgers equation which develops shock at finite time. We will come back later to this point, with the analysis of the works \cite{Poggi1995,Gallone2022-PRL,Gallone2024} and with the analysis of the blow up of the underlying integrable approximation in \cite{Bambusi-Libro}.

Third, if the left-traveling wave is an exact solution for the wave equation, due to the form of 
\eqref{eq:EquationsLandR}, it is not an exact solution for the \emph{coupled} system. From a geometric point of view, this amounts to say that the set $\{R=0\}$ is an invariant manifold for the wave equation, but it is not for the coupled system. One can therefore ask whether it is possible or not to construct the set of \emph{almost}-left traveling waves. And in case of positive answer, are the equations of motion on the manifold approximately integrable?

This is the question that we posed in \cite{Gallone2021} and we are reviewing here the construction. 
%{\color{blue}[RIMUOVERE]Repeating the argument from eq.~\eqref{eq:InterpolatingFields} to \eqref{eq:EquationsLandR} keeping also terms proportional to $\epsilon \sim h^4$, one gets first,
%\begin{equation}
%	\begin{split}
%		q_{j+1}+&q_{j-1}-2q_j = \sqrt{\epsilon} h \left(Q_{xx}+\frac{h^2}{12} Q_{xxxx} + \frac{h^4}{360} Q_{xxxxxx} \right) + O(\sqrt{\epsilon}h^7) \\
%		\alpha \big[(q_{j+1}&-q_j)^2-(q_j-q_{j-1})^2 \big] \\
%		&=\sqrt{\epsilon} h \, \alpha \sqrt{\epsilon} \left(2Q_x Q_{xx}+ \frac{h^2}{3} Q_{xx} Q_{xxx}+\frac{h^2}{6} Q_x Q_{xxxx} \right) + O(\epsilon h^5) \\
%		\beta \big[(q_{j+1}&-q_j)^3-(q_j-q_{j-1})^3 \big]= \sqrt{\epsilon} h \, 3\beta \epsilon Q^2_xQ_{xx} + O(\epsilon^{\frac{3}{2}} h^3)
%	\end{split}
%\end{equation}
%}
To proceed in the construction, it is convenient to introduce the discrete derivative of step $h$ as
\begin{equation}\label{eq:Dh}
	D_h U(x) \, := \, \frac{U(x+h/2)-U(x-h/2)}{h}
\end{equation}
whose asymptotic expansion for $h$ small is 
\begin{equation}\label{eq:DhAsymptExp}
	D_h \, = \, \partial_x + \frac{h^2}{24} \partial_{x}^3 +\frac{h^4}{1920} \partial_x^5+O(h^6) \, .
\end{equation}
Then, the equations of motion of the fields $Q$ and $P$ are, from \eqref{eq:PdotC}, \eqref{eq:QdotC}, \eqref{eq:PhiPrimoPrimo}, and \eqref{eq:HamEqFPUT}
\begin{equation}
	\begin{split}
		Q_\tau(x) \, &= \,P(x) \, , \\
		P_\tau(x) \, &= \, \frac{1}{h \sqrt{\epsilon}} \Big[ \Phi_{\mathrm{F}}'\big({\textstyle\frac{\sqrt{\epsilon}}{h} (Q(x+h)-Q(x))} \big) \\
		&\qquad \qquad -\Phi_{\mathrm{F}}'\big({\textstyle\frac{\sqrt{\epsilon}}{h} (Q(x)-Q(x-h))} \big) \Big] \, .
	\end{split}
\end{equation}
At this point we can use the definition \eqref{eq:Dh}, to get
\begin{equation}
	\begin{split}
		\frac{1}{h \sqrt{\epsilon}} \Big[& \Phi_{\mathrm{F}}'\big({\textstyle\frac{\sqrt{\epsilon}}{h} (Q(x+h)-Q(x))} \big)  -\Phi_{\mathrm{F}}'\big({\textstyle\frac{\sqrt{\epsilon}}{h} (Q(x)-Q(x-h))} \big) \Big] \\
		&=\frac{1}{\sqrt{\epsilon}} D_h \Phi'_{\mathrm{F}}\big({\textstyle \frac{\sqrt{\epsilon}}{h}(Q(x+h/2)-Q(x-h/2))}\big) \\
		&=\frac{1}{\sqrt{\epsilon}} D_h \Phi'_F\big(\sqrt{\epsilon} D_h Q(x)\big) \, .
	\end{split}
\end{equation}
Therefore, the equations of motion are compactly written as
\begin{equation}
	\begin{split}
		Q_\tau \, &=\, P \\
		P_\tau \,&=\, \frac{1}{\sqrt{\epsilon}} D_h \Phi_{\mathrm{F}}'(\sqrt{\epsilon} D_h Q) \, .
	\end{split}
\end{equation}
This rewriting is useful because we now introduce  different \emph{Riemann invariants} as
\begin{equation}
	\begin{split}
		\lambda \, &:= \, \frac{D_h Q+P}{\sqrt{2}} \, , \qquad \, \text{(Left-travelling wave)} \, , \\
		\rho \, &:=\, \frac{D_h Q - P}{\sqrt{2}} \, , \qquad \, \text{(Right-travelling wave)} \, .
	\end{split}
\end{equation}
Note that, due to the asymptotic expansion \eqref{eq:DhAsymptExp}, the new Riemann invariants are $h^2$ close to the ones defined in \eqref{eq:LeftAndRightTV}. 

The equations of motion in these new variables are
\begin{equation}\label{eq:LRwithF}
	\begin{split}
		\lambda_\tau \, &=\, \mathcal{F}(\lambda,\rho) \, , \\
		\rho_\tau \, &= \,-\mathcal{F}(\rho,\lambda) \, ,
	\end{split}
\end{equation}
where $\mathcal{F}(U,V)$ is given by
\begin{equation}\label{eq:FunctionalF}
	\begin{split}
		\mathcal{F}(U,V) \, &= \,U_x+ \Big[\frac{h^2}{24} U_{xxx}+ \frac{\alpha \sqrt{\epsilon}}{\sqrt{2}} (U+V)(U+V)_x \Big] \\
		&\quad + \Big[ \frac{h^4}{1920} U_{xxxxx}+\frac{\alpha \sqrt{\epsilon} h^2}{8 \sqrt{2}} (U+V)_x(U+V)_{xx} \\
		&\quad + \frac{\alpha \sqrt{\epsilon} h^2}{24 \sqrt{2}} (U+V) (U+V)_{xxx} + \frac{3 \beta \epsilon}{4} (U+V)^2 (U+V)_x \Big] + \cdots \, .
	\end{split}
\end{equation}
We now want to construct recursively an invariant manifold for the system \eqref{eq:LRwithF} close to the invariant manifold of the one-directional traveling waves for the wave equations. Then, we assume there exists an equation slaving $\rho$ to $\lambda$ of the form
\begin{equation}
	\rho \, = \, \mathcal{G}(\lambda) \, , \qquad \mathcal{G}=O(\sqrt{\epsilon}) \, .
\end{equation}
The vector field $\mathcal{G}$ can be obtained recursively from the second equation of \eqref{eq:LRwithF}, indeed the system \eqref{eq:LRwithF} becomes
\begin{equation} \label{eq:SystemWIthInvariant}
	\begin{split}
		\lambda_\tau \, &= \, \mathcal{F}(\lambda,\mathcal{G}(\lambda)) \, ,\\
		\mathcal{G}_\tau \,&= \, -\mathcal{F}(\mathcal{G}(\lambda),\lambda) \, .
	\end{split}
\end{equation}
Using the chain rule, we obtain for the Left-hand side of the second equation
\begin{equation}\label{eq:LHSConsistency}
	\sqrt{\epsilon} \mathcal{G}_2(\lambda)_\tau=\sqrt{\epsilon} \nabla_{L^2} \mathcal{G}(\lambda) \lambda_\tau = \sqrt{\epsilon} \nabla_{L^2} \mathcal{G}(\lambda) \lambda_x + \cdots \, ,
\end{equation}
and for the Right-hand side we get
\begin{equation}\label{eq:RHSConsistency}
	- \mathcal{F}(\sqrt{\epsilon} \mathcal{G}_2(\lambda), \lambda) \, = \, - \sqrt{\epsilon} \nabla_{L^2} \mathcal{G}_2(\lambda) \lambda_x - \frac{\alpha \sqrt{\epsilon}}{\sqrt{2}} \lambda \lambda_x + \cdots \, .
\end{equation}
Putting together \eqref{eq:LHSConsistency} and \eqref{eq:RHSConsistency} we get, up to higher order terms,
\begin{equation}
	\nabla_{L^2}\mathcal{G}_2(\lambda) \lambda_x \, = \, - \frac{\alpha}{2 \sqrt{2}} \lambda \lambda_x \, .
\end{equation}
If we now require $\mathcal{G}_2(\lambda)$ to have zero-average on the torus, we get as unique solution 
\begin{equation}
	\mathcal{G}_2(\lambda) \, = \, \frac{\alpha}{4 \sqrt{2}} \big( \langle \lambda^2 \rangle - \lambda^2 \big) \, .
\end{equation}
Here and henceforth we denote by $\langle \lambda \rangle =\int_{\mathbb{T}} \lambda \, dx$ the average of $\lambda$ on the torus.

Summarizing, so far we obtained that the equation $\rho=\frac{\alpha \sqrt{\epsilon}}{4 \sqrt{2}}(\langle \lambda^2 \rangle - \lambda^2)$ defines an almost-invariant manifold for the FPUT system. On this manifold, the equations of motion for the variable $\lambda$ are given by \eqref{eq:FunctionalF} and \eqref{eq:LRwithF}
\begin{equation}\label{eq:OnInvariantManifold}
	\begin{split}
		\lambda_\tau \,&=\, \mathcal{F}(\lambda,\sqrt{\epsilon} \mathcal{G}_2(\lambda)) \\
		&=\, \lambda_x + \Big[\frac{h^2}{24} \lambda_{xxx} +\frac{\alpha \sqrt{\epsilon}}{\sqrt{2}} \lambda \lambda_x \Big]+\Big[ \frac{h^4}{1920} \lambda_{xxxxx} \\
		&\quad +\frac{\alpha \sqrt{\epsilon} h^2}{24 \sqrt{2}} (3 \lambda_x \lambda_{xx}+\lambda \lambda_{xxx}) +\frac{3 \epsilon}{4}\Big( \beta - \frac{\alpha^2}{2} \Big) \lambda^2 \lambda_x + \frac{\alpha^2 \epsilon}{8} \langle \lambda^2 \rangle \lambda_x \Big] \, .
	\end{split}
\end{equation}

One would now like to study equation \eqref{eq:OnInvariantManifold}. In the spirit of perturbation theory, one could ask, for example, whether there exists a close-to-identity transformation of vector fields that bring \eqref{eq:OnInvariantManifold} in an equation that is a linear combination of the first terms of the KdV hierarchy \eqref{eq:KdVHierarchy}.

This latter problem has been analyzed by Hiraoka and Kodama in a series of works for equations on the line \cite{Kodama1985,Kodama1987,Hiraoka} and it has been adapted on the torus by Ponno, Rink and myself in \cite[Theorem 4]{Gallone2021}. The net result is that, when $\alpha \neq 0$, for all values of $\alpha$ and $\beta$, there exists a near to identity transformation of vector fields that brings \eqref{eq:OnInvariantManifold} into an integrable equation (up to terms of order $h^6$). This is in agreement with the most recent numerical experiments that according to which, the metastable state persists for times of order $\epsilon^{-5/4}$. 

Indeed, if we link the two parameters $\epsilon=h^4$, we showed formally that the dynamics is integrable up to order $h^4$. That is, KdV actions are almost preserved for times $\tau \lesssim h^{-4}$. Recalling that $\tau=ht$, we see that KdV actions are preserved up to times $t \lesssim h^{-5}$, that is $\epsilon^{-\frac{5}{4}}$, that is the same time-scale observed in Figure~\ref{fig:EvolutionSpectrum}.

Then, one may be curious to see what happens at the next order. Applying the same idea (computations are longer), one has that the Kodama transformation brings the equations of motion on the invariant manifold on an integrable form \emph{if and only if}
\begin{equation}
	14 \alpha^3-27 \alpha \beta + 12 \gamma \, = \, 0 \, .
\end{equation}
Note that, this latter relation is satisfied for parameters of the Toda model, which are
\begin{equation}
	\beta\,=\, \beta_T \, = \, \frac{2}{3} \alpha^2 \, , \qquad \gamma \, = \, \gamma_T \, = \, \frac{1}{3} \alpha^3 \, . 
\end{equation}

Even though these results are far from conclusive, as many justifications and theorems are still missing, it is remarkable that the perturbative order in $h$ at which the normal form becomes non-integrable coincides with numerical observations, also highlighting the special role of the Toda chain.

\subsection{The continuous model and KdV as resonant normal form}\label{subsect:KdVResonantNF}
The idea one wants to follow is to \emph{embed} the FPUT chain in a continuous model. At first, one follows the construction of analytic functions $Q$ and $P$ as \eqref{eq:InterpolatingFields}. Then, one inserts \eqref{eq:InterpolatingFields} in the Hamiltonian
\begin{equation}
	\begin{split}
		H(Q,P) \, &= \, \sum_{j \in \mathbb{Z}_N} \left[ \frac{\epsilon P^2(x)}{2} +  \Phi_F\big(\sqrt{\epsilon} D_hQ(x+h/2) \big) \right] \Bigg|_{x=jh,\tau=ht} \\
		&=\,\epsilon N \frac{1}{N} \sum_{j \in \mathbb{Z}_N} \left[ \frac{P^2(x)}{2} +  \frac{1}{\epsilon} \Phi_F\big(\sqrt{\epsilon} D_hQ(x+h/2) \big) \right] \Bigg|_{x=jh,\tau=ht} \, .
	\end{split}
\end{equation}
With a suitable limiting procedure, we now decouple the limit $h\to 0$ and the limit $N\to 0$. The reader might be confused by this procedure, and it is natural to be. We imagine the limiting procedure $N \to +\infty$, to be done on functions $Q$ and $P$ whose support, in Fourier space, is constant in $k/N$. A prototypical example is a set of data for which only the first ten percent of Fourier modes is initially excited (as it is done in the numerical simulations of \cite{Benettin2008-1D,Benettin2011,Benettin2022}). In this \emph{adapted} limiting procedure, the first sum then becomes a Riemann sum and 
\begin{equation}
	\mathscr{H}(Q,P) \, := \,  \widetilde{\lim_{N \to +\infty}} \frac{1}{\epsilon N} H(Q,P)
\end{equation}
which gives
\begin{equation}\label{eq:HCont}
	\mathscr{H}(Q,P) \, = \, \int_{\mathbb{T}}\left[ \frac{P^2(x)}{2} + \frac{1}{\epsilon} \Phi_F\big(\sqrt{\epsilon} D_h Q(x+h/2) \big) \right] \, dx \, .
\end{equation}

Then, using the explicit expression of $\Phi_F$ \eqref{eq:FPUT-Potential} and the asymptotic expansion of $D_h$ in \eqref{eq:DhAsymptExp} we get, up to an error of order $O(h^4)$,
\begin{equation}
	\begin{split}
			\frac{1}{\epsilon} \Phi_F \big(&\sqrt{\epsilon} D_h Q(x+h/2) \big) \\ &=\frac{1}{2}(D_h Q(x+h/2))^2+\frac{\alpha \sqrt{\epsilon}}{3}(D_h Q(x+h/2))^3 + O(\epsilon,h^4) \\
			&=\frac{1}{2} (Q_{x}(x))^2+ \frac{h}{2}Q_x(x) Q_{xx}(x) + \frac{h^2}{8}(Q_{xx}(x))^2 \\
			&\quad+ \frac{h^2}{6}Q_x(x) Q_{xxx}(x)+ \frac{\alpha \sqrt{\epsilon}}{3}(Q_x(x))^3
	\end{split}
\end{equation}
Using now periodicity of $ Q(x)$, we note that
\begin{equation}
	\begin{split}
		\int_{\mathbb{T}} Q_x(x) Q_{xx}(x) \, dx &=\frac{1}{2} \Big[ (Q_x(1))^2-(Q_x(0))^2\Big]=0 \, , \\
		\int_{\mathbb{T}} Q_x(x) Q_{xxx}(x) \, dx &=-\int_{\mathbb{T}} (Q_{xx}(x))^2 \, dx \, ,
	\end{split}
\end{equation}
and therefore, we can rewrite \eqref{eq:HCont} as
\begin{equation}\label{eq:Tu6517}
	\begin{split}
		\mathscr{H}(Q,P) \, &=\, \int_{\mathbb{T}} \left[\frac{P^2(x)}{2}+\frac{1}{2}(Q_x(x))^2 \right] \, dx  \\
		&\quad+ \int_{\mathbb{T}} \left[-\frac{h^2}{48} (Q_{xx}(x))^2+ \frac{\alpha \sqrt{\epsilon}}{3}(Q_x(x))^3 \right] \, d x + O(\epsilon,\sqrt{\epsilon}h^2,h^4) \, .
	\end{split}
\end{equation}
We can now introduce the left and right traveling waves as in \eqref{eq:LeftAndRightTV}. 
The transformation $(Q,P) \mapsto (L,R)$ is non canonical and transforms the Poisson tensor from $J_E$ to $J_G$ (see \eqref{eq:PoissonTensors}). Then, calling $\mathscr{K}$ the transformed Hamiltonian, we have
\begin{equation}\label{eq:LRHamiltonianFPUT}
	\begin{split}
		\mathscr{K}(L,R) \,&=\, \int_{\mathbb{T}} \left[ \frac{L^2}{2}+\frac{R^2}{2}\right] \, dx \\
		&\quad +\int_{\mathbb{T}} \left[ \frac{h^2}{16}(L_x+R_x)^2+\frac{\alpha \sqrt{\epsilon}}{6 \sqrt{2}}(L+R)^3 \right] \, dx + O(\epsilon,\sqrt{\epsilon}h^2,h^4) \, .
	\end{split}
\end{equation}

We are now in the position to apply perturbation theory for infinite dimensional systems. Indeed, defining
\begin{equation}
	\begin{split}
		\mathscr{K}_0(L,R) \, &= \, \int_{\mathbb{T}} \left[ \frac{L^2}{2}+ \frac{R^2}{2} \right] \, dx \, , \\
		\mathscr{P}_1(L,R) \, &= \, \int_{\mathbb{T}} \left[ \frac{h^2}{16}(L_x+R_x)^2+\frac{\alpha \sqrt{\epsilon}}{6 \sqrt{2}}(L+R)^3 \right] \, dx \, ,
	\end{split}
\end{equation} 
we have that the flow of $\mathscr{K}_0$ is periodic of period $T=1$ and thus Proposition \ref{prop:Averaging} can be applied. In other words, there exists a close-to-identity canonical transformation that maps the Hamiltonian $\mathscr{K}$ into $\tilde{\mathscr{K}}$, with
\begin{equation}\label{eq:HamKTilde}
	\tilde{\mathscr{K}}(\lambda,\rho) \, = \, \mathscr{K}_0(\lambda,\rho) + \mathscr{Z}_1(\lambda,\rho) + \cdots 
\end{equation}
with $\mathscr{Z}_1(\lambda,\rho)=[\mathscr{P}_1]_{\mathscr{K}_0}(\lambda,\rho)$. Then, one has to compute the latter. With this aim, we start by introducing the anti-derivative operator, that is defined as
\begin{equation}\label{eq:DxMuno}
	\partial_x^{-1} U(x) \,:=\, \sum_{k \in \mathbb{Z} \setminus \{0\}} \frac{1}{2 \pi i k} \hat{U}_k e^{2 \pi i k x} \, .
\end{equation}
Note that, in particular
\begin{equation}
	\partial_x \partial_{x}^{-1} U \, = \, U-\langle U \rangle_{\mathbb{T}} \, , \qquad \partial_x^{-1} \partial_{x} U \, = \, U-\langle U \rangle_{\mathbb{T}} \, .
\end{equation}
Then, the following Lemma summarizes the computations needed to obtain the normal form and the generator in Proposition \ref{prop:Averaging} for the Hamiltonian \eqref{eq:Tu6517}.

\begin{lemma}\label{lem:AveragingLemmaComput}
	Let $U$ and $V$ be functions on the torus $\mathbb{T}$. Then
	\begin{itemize}
		\item[(a)] for any $U,V \in L^1(\mathbb{T})$,
		\begin{equation}
			\int_{\mathbb{T}} \int_0^1	U(x-s) V(x+s) \, ds \, dx \,= \, \left( \int_{\mathbb{T}} U(x) \, dx \right) \left( \int_{\mathbb{T}} V(x) \, dx \right) \, ,	
		\end{equation}
		\item[(b)] for any $U,V \in L^2(\mathbb{T})$
		\begin{equation}
			\begin{split}
				\int_{\mathbb{T}} \int_0^1 &s \,  U(x-s) V(x+s) \, ds \, dx \\
				&=\,\frac{1}{2} \left(\int_{\mathbb{T}} U(x) \, dx \right) \left( \int_{\mathbb{T}} V(x) \, dx \right)+\frac{1}{2} \int_{\mathbb{T}} U(x) \partial_x^{-1} V(x) \, dx
			\end{split}
		\end{equation}
	\end{itemize}
\end{lemma}
\begin{proof}
	Both (a) and (b) can be proved by passing to Fourier coefficients. Indeed, for (a), we have
	\begin{equation}
		\begin{split}
			\int_{\mathbb{T}} \int_0^1 &U(x-s) V(x+s) \, ds \, dx \,\\
			&= \, \int_0^1 \int_0^1 \sum_{k,k' \in \mathbb{Z}} \hat{U}_k \hat{V}_{k'} e^{2 \pi i k (x-s)} e^{2 \pi i k' (x+s)} \, dx \, ds \\
			&= \sum_{k,k' \in \mathbb{Z}} \hat{U}_k \hat{V}_{k'} \delta_{k+k',0} \delta_{k-k',0} \, = \, \hat{U}_0 \hat{V}_0 \, .
		\end{split}
	\end{equation}
	
	Concerning (b), we have 
	\begin{equation}
		\begin{split}
			\int_0^1 \int_{\mathbb{T}} s U(x+s) V(x-s) \, ds \, dx \, &= \, \sum_{k \in \mathbb{Z}} \hat{U}_k \hat{U}_{-k} \int_0^1 s \, e^{4 \pi i k s} \, ds \\
			&=\sum_{k \in \mathbb{Z}} \hat{U}_k \hat{V}_{-k} \left(\frac{1}{2} \delta_k+\frac{1}{4 \pi i k}(1-\delta_{k,0}) \right) \, .
		\end{split}
	\end{equation}
	One gets the thesis by recalling the definition of $\partial_x^{-1}$ in \eqref{eq:DxMuno}.
\end{proof}
Applying Lemma \ref{lem:AveragingLemmaComput}, we get
\begin{equation}
	\mathscr{Z}_1(\lambda,\rho) \, = \, \int_{\mathbb{T}} \left[\frac{h^2}{48} \big((\lambda_x)^2+(\rho_x)^2 \big)+\frac{\alpha \sqrt{\epsilon}}{6 \sqrt{2}} \big(\lambda^3+\rho^3\big) \right] \, dx
\end{equation}
and for the generating function
\begin{equation}
	\mathscr{G}_1(\lambda,\rho) \, = \, \int_0^1 \Big( \frac{\alpha \sqrt{\epsilon}}{4 \sqrt{2}} \big( \rho^2 \partial_x^{-1} \lambda - \lambda^2 \partial_x^{-1} \rho \big)+\frac{1}{2} \frac{h^2}{4!} \rho \lambda_x \Big) \, dx \, .
\end{equation}

Note that, the equations of motion associated to $\widetilde{\mathscr{K}}$ are decoupled at the leading order. This is to say that, at leading orders, there is no interaction between left- and right-traveling waves $\lambda$ and $\rho$. Thus, the normal form of the continuum version of the FPUT model consists in a \emph{pair of counter-propagating} Korteweg-de Vries equations:
\begin{equation}\label{eq:EquationsOfMotion12}
	\begin{split}
		\lambda_\tau \, &=\, \lambda_x + \frac{\alpha \sqrt{\epsilon}}{\sqrt{2}} \lambda \lambda_x-\frac{h^2}{24} \lambda_{xxx} + \cdots \\
		\rho_\tau \, &=\, -\rho_x - \frac{\alpha \sqrt{\epsilon}}{\sqrt{2}} \rho \rho_x + \frac{h^2}{24} \rho_{xxx} + \cdots \, . 
	\end{split}
\end{equation}
Note that this construction is coherent with Zabusky and Kruskal's idea and, if $\rho \sim 0$ at $t=0$, then $\rho_\tau=O(\epsilon,\sqrt{\epsilon}h^2,h^4)$. That is, the renormalized right traveling wave remains small for large times.

The next question we would like to address, is how to connect the behavior of the continuum model to the lattice or, in other words, how to get information on the lattice dynamics from the pair of counter-propagating Korteweg-de Vries? The connection has been clarified by \cite{Bambusi2006}. In particular the following theorem holds

\begin{theorem}[Bambusi-Ponno \cite{Bambusi2006}]\label{thm:Bambusi-Ponno}
	Consider an initial condition of the form
	\begin{equation}
		\frac{E_{1/N}(0)}{N} \, = \, C_0 \mu^4 \, , \qquad E_{k/N}(0) \, = \, 0 \, , \qquad \forall k \neq 1 \, ,
	\end{equation}
	where $C_0$ is any fixed constant and $\mu=h=1/N \ll 1$. 
	
	Then, for any fixed time $T_f$ there exist positive constants $\mu^*$, $\sigma$, $C_1$ and $C_2$ (dependent on $C_0$ and $T_f$) such that, for all $k, \mu < \mu^*$ and $|t| \leq T_f/ \mu^3$
	\begin{equation}\label{eq:ExponentialDecay}
		\frac{E_{k/N}(t)}{N} \, \leq \mu^4 C_1 e^{-\sigma \frac{k}{N \mu}} \, .
	\end{equation}
\end{theorem}

Before sketching the proof, let us remark that in the regime $E/N=\epsilon \sim \mu^4$, equation \eqref{eq:ExponentialDecay} reads
\begin{equation}
	\frac{E_{k/N}(t)}{N} \leq E_{\mathrm{tot}} C_1 e^{-\epsilon^{-\frac{1}{4}}\sigma \frac{k}{N}} \, ,
\end{equation}
which is the scaling found in \cite{Berchialla2004} and discussed in Subsection \ref{subsec:FurtherNumerics}.

\begin{proof}[Sketch of the proof]
	One first performs an analytic interpolation of the initial datum defining $Q$ and $P$ as in \eqref{eq:InterpolatingFields}. Then, one passes to Riemann invariants $\lambda$ and $\rho$. 
	
	Then, for the functional setting, for $\sigma,s\geq0$ one first introduces the spaces $A_{\sigma,s}$ as
	\begin{equation}
		A_{\sigma,s} \, := \, \{ F \in L^2(\mathbb{T}) \, | \, \Vert F \Vert_{A_{s,\sigma}} < +\infty\} \, ,
	\end{equation}
	where 
	\begin{equation}
		\Vert F \Vert_{A_{s,\sigma}}^2 \, := \, \sum_{k \in \mathbb{Z}} |\hat{F}_k|^2 |k|^{2s} e^{2 \sigma |k|} \, .
	\end{equation}
	Then, one considers as phase space of the continuous model 
	\begin{equation}
		\Gamma_{s,\sigma} \, := \, A_{s,\sigma} \oplus A_{s,\sigma}
	\end{equation}
	and denotes by $B_{s,\sigma}(R)$ the ball of radius $R$ in $\Gamma$, i.e.
	\begin{equation}
		B_{s,\sigma}(R) \, := \, \{Z \in \Gamma_{s,\sigma} \, | \, \Vert z \Vert_{s,\sigma}<R\} \, .
	\end{equation}
	
	One thus performs the normal form transformation using $\mu$ as perturbative parameter and proves a normal form Lemma. In particular, one proves thar for all $r \geq 5$, $\exists \mu_*=\mu_*(r)$ such that if $\mu < \mu_*$, there exists a canonical transformation $\Phi_{\mathscr{G}_1}^{\mu^2}: B_{r,\sigma}(1) \to B_{r,\sigma}(2)$ such that
	\begin{itemize}
		\item[(a)] $\mathscr{H} \circ \Phi_{\mathscr{G}_1}^{\mu^2}=\mathscr{H}_0+[\mathscr{P}]_{\mathscr{H}_0} + \mathscr{R}$
		\item[(b)] $\sup_{\Vert Z \Vert_{r,\sigma}} \Vert J \nabla_{L^2}\mathscr{R} \Vert_{0,\sigma} \leq C_r \mu^{4-\frac{12}{6+r}}$ \, ,
		\item[(c)] $\forall 1 \leq r_1 \leq r$, $\Phi_{\mathscr{G}_1}^{\mu^2} : B_{r_1,\sigma} \to \Gamma_{r_1,\sigma}$ and fulfills
		\begin{equation}
			\sup_{\Vert z \Vert_{r_1,\sigma}\leq  1} \Vert Z - \Phi_{\mathscr{G}_1}^{\mu^2}(Z) \Vert_{r_1,\sigma} \,\leq\ C \mu^{2-\frac{6}{6+r}} \, .
		\end{equation}
	\end{itemize}
	At this point, one has to translate the information on the dynamics of the normal form into an information on the evolution of Fourier normal modes. To do so, we recall that
	\begin{equation}
		\omega_k \, = \, \Big| 2 \sin \big( {\textstyle \frac{\pi k}{2 N}} \big) \Big| \, = \, |2 \sin \big( {\textstyle \frac{\pi}{2} h k } \big)| ,
	\end{equation}
	and therefore,
	\begin{equation}
		\begin{split}
			E_k \,&=\, \epsilon \frac{|\hat{P}_k|^2+ h^{-2} \omega_k^2|\hat{Q}_k |^2}{2} \\
			&=\, \epsilon \frac{|\hat{P}_k|^2+\pi^2 k^2 |\hat{Q}_k|^2}{2} + \epsilon \left(\frac{\omega_k^2}{h^2 \pi^2 k^2}-1 \right) \frac{\pi^2 k^2}{2} |\hat{Q}_k|^2 \, .
		\end{split}
	\end{equation}
	Passing now to $L$ and $R$ variables as in \eqref{eq:LeftAndRightTV}, defining $\epsilon=\mu^4$ and $h=\mu^{2}$, one has
	\begin{equation}
		\begin{split}
			E_k \, & \leq \, \mu^4 \frac{|\hat{L}_k|^2+|\hat{R}_k|^2}{2}+\mu^4 \left|\frac{\omega_k^2}{\mu^2 \pi^2 k^2}-1 \right|(|\hat{L}_k|^2+|\hat{R}_k|^2)
		\end{split}
	\end{equation}
	and, by Taylor expanding $\omega_k^2$, one has $|\omega_k^2-h^2 \pi^2 k^2| \leq C h^4 k^4$ for all $k \in \mathbb{Z}$ and one gets
	\begin{equation}
		|E_k| \, \leq \, \mu^4 \frac{|\hat{L}_k|^2+|\hat{R}_k|^2}{2}+C \mu^6 k^2(|\hat{L}_k|^2 + |\hat{R}_k|^2) \, .
	\end{equation}
	
	Neglecting the remainder, any information on the evolution \emph{along} the KdV of $\hat{L}_k$ and $\hat{R}_k$ can be translated into the evolution of $E_k$. Thus, one exploits the following theorem by Kappeler and P\"oschel \cite{Kappeler2003}: let $u \in A_{s,\sigma}$, then $u(t) \in A_{s,\sigma(t)}$ for all $t \in \mathbb{R}$. Thus, supposing $\mathscr{R}=0$, we would have that for the times of validity of the KdV approximation, one has
	\begin{equation}
		|E_k(t)| \, \leq \mu^4 C e^{-\sigma(t) |k|} \, .
	\end{equation}
	
	Then, if the remainder was zero and the restriction to the lattice was for free, one would have proven the theorem. The last, technical part, is to prove that for times $T \leq T_f/\mu^3$, both the remainder and the discretization give small corrections.
\end{proof}

\begin{xca}\label{exrc:BetaModelNF}
	Take the continuum system with Hamiltonian \eqref{eq:HCont} and set $\alpha=0$ in $\Phi_F$. Then show that, instead of \eqref{eq:EquationsOfMotion12} one obtains a pair of counter propagating modified Korteweg-de Vries equations (mKdV) as resonant normal forms. (Hint: pair terms of order $h^2$ and $\epsilon$; the exact mKdV equations are given in \eqref{eq:BetaMKDV}).
\end{xca}

\subsection{Zero-dispersive limit of KdV and Burgers turbulence} 
Any information on the thermodynamic limit of the Fermi-Pasta-Ulam-Tsingou is encoded formally in the zero-dispersion limit of the KdV. Formally, this limit can be encoded by taking $h\to0$ in the equations of motion \eqref{eq:EquationsOfMotion12}. 

This limit is quite singular as solutions of the equation after the limit are not defined for all times. Indeed, the zero-dispersion limit is, in general, a hard problem and only recently some rigorous mathematical results have been obtained in the context of the Bejamin-Ono equation \cite{Grard2024}. To my knowledge, this zero-dispersion limit problem was first addressed by Lax and Levermore \cite{Lax1983-I,Lax1983-II,Lax1983-III} in the case of the Korteweg-de Vries equation, using the integrability of this equation. Lax and Levermore initiated a series of remarkable results on special initial data, see in particular \cite{Claeys2008,Deift1997,Grava2007,Venakides1985} and references therein. Nevertheless, all properties needed for the analysis of the thermodynamic limit of the FPUT model are still lacking.

So without the attempt of being mathematically rigorous, we now derive some dynamical information of the Burgers equation that are actually relevant in FPUT dynamics. The idea is taking formally the limit $h\to0$ in the equations \eqref{eq:EquationsOfMotion12} and to consider long-wavelength initial data. In this section we will focus on the one-parameter set of initial data given by
\begin{equation}\label{eq:InitialDataDisc}
	\begin{split}
		q_j(0) \,&=\, \frac{N \sqrt{\epsilon}}{\pi} \cos \varphi \sin \big({\textstyle \frac{2 \pi j}{N} }\big) \\
		p_j(0) \, &= \, \frac{\omega_1 N \sqrt{\epsilon}}{\pi} \sin \varphi \cos \big({\textstyle \frac{2 \pi j}{N}}\big) \, ,
	\end{split}
\end{equation}
where $\omega_k$ is given by \eqref{eq:OmegaK}. In the $N \to +\infty$ limit, the initial datum is given by
\begin{equation}\label{eq:InitialDataCont}
	\begin{split}
		Q(x,0) \,&=\,\frac{\sqrt{\epsilon}}{\pi} \cos \varphi \sin(2 \pi x) \, , \\
		P(x,0) \, &=\, 2\sqrt{\epsilon} \sin \varphi \cos(2 \pi x) \, .
	\end{split}
\end{equation}
This analysis, started in the work Poggi-Ruffo-Kantz in 1995 \cite{Poggi1995} and then, the relevant non-trivial properties have been derived, discussed and numerically tested in recent works \cite{Gallone2022-PRL,Gallone2024,Gallone-Grande-2026}. In the latter ones, a second order normal form is performed in order to get more refined information. In these notes, for the sake of simplicity, we limit ourselves to the analysis of the first order, that is: we keep terms of order at most $\sqrt{\epsilon}$. Hamiltonian \eqref{eq:HamKTilde} with $h=0$ reads
\begin{equation}
	\widetilde{\mathscr{K}}_{h=0}(\lambda,\rho) \, = \, \int_{\mathbb{T}} \frac{\lambda^2 + \rho^2}{2} \, dx \, + \, \int_{\mathbb{T}} \frac{\alpha \sqrt{\epsilon}}{6 \sqrt{2}}(\lambda^3+\rho^3) \, dx \, .
\end{equation}
The equations of motion associated to this Hamiltonian, up to terms $O(\sqrt{\epsilon})$, are
\begin{equation}
\label{eq:larhoev}
\left\{
\begin{split}
& \lambda_\tau=\left[1+\frac{\alpha\sqrt{\epsilon}}{\sqrt{2}}\lambda \right]\lambda_x \, , \\
& \rho_\tau=-\left[1+\frac{\alpha\sqrt{\epsilon}}{\sqrt{2}}\rho\right]\rho_x \, .
\end{split}
\right.
\end{equation} 

Concerning the initial conditions $(\lambda_0,\rho_0)$ satisfied by the fields $\lambda$ and $\rho$, one has to invert the transformation generating the change of coordinates $(R,L) \mapsto (\rho,\lambda)$ with $h=0$.
This is the flow at time $1$ of the Hamiltonian $\mathscr{G}_1$, which is obtained from \eqref{eq:G1Calcolo} by setting $h=0$, explicitly
\begin{equation}
\label{eq:g1}
\mathscr{G}_1(\lambda,\rho)=\frac{\alpha\sqrt{\epsilon}}{4\sqrt{2}}\ \left\langle\rho\partial_x^{-1}(\lambda^2)+\rho^2\partial_x^{-1}\lambda\right\rangle\ .
\end{equation} 
One can then express the new fields $\lambda$ and $\rho$ in terms of the old ones, $L$ and $R$, 
\begin{equation}
\begin{split}
\lambda&= e^{-\mathcal{L}_1}L=L-\{L,\mathscr{G}_1\}(L,R)+O(\epsilon)\ , \\ 
\rho&= e^{-\mathcal{L}_1}R=R-\{R,\mathscr{G}_1\}(L,R)+O(\epsilon)\ .
\end{split}
\end{equation}
Neglecting a remainder of $O(\epsilon)$, the final result is
\begin{equation}
\label{eq:Caoa}
\begin{split}
\lambda=&  L+\frac{\alpha \sqrt{\epsilon}}{4\sqrt{2}}\left(R^2-\langle R^2\rangle\right)+
\frac{\alpha \sqrt{\epsilon}}{2\sqrt{2}}(LR+L_x\partial_x^{-1}R)\ ,\\
\rho=& R+\frac{\alpha \sqrt{\epsilon}}{4\sqrt{2}}\left(L^2-\langle L^2\rangle\right)+
\frac{\alpha \sqrt{\epsilon}}{2\sqrt{2}}(LR+R_x\partial_x^{-1}L)\ .
\end{split}
\end{equation} 
Substituting in the latter expression the initial condition \eqref{eq:InitialDataCont}, where $\theta=\varphi-\pi/4$, and neglecting the remainder $O(\epsilon)$, one gets
\begin{equation}
\label{eq:indatlarho}
\left\{
\begin{split}
\lambda_0= & 2\cos\theta\cos(2\pi x)+\frac{\alpha\sqrt{\epsilon}(\sin^2\theta-2\sin2\theta)}{2\sqrt{2}}\cos(4\pi x)\\
\rho_0 = &-2\sin\theta\cos(2\pi x)+\frac{\alpha\sqrt{\epsilon}(\cos^2\theta-2\sin2\theta)}{2\sqrt{2}}\cos(4\pi x)
\end{split}
\right.\ .
\end{equation}
We finally observe that the transformation (\ref{eq:Caoa}) preserves the space average of the fields, so that $\langle \lambda\rangle=\langle L\rangle=0$ and $\langle \rho\rangle=\langle R\rangle=0$.

\begin{figure*}[t]
\begin{subfigure}{0.46\textwidth}
		\includegraphics[width=\textwidth]{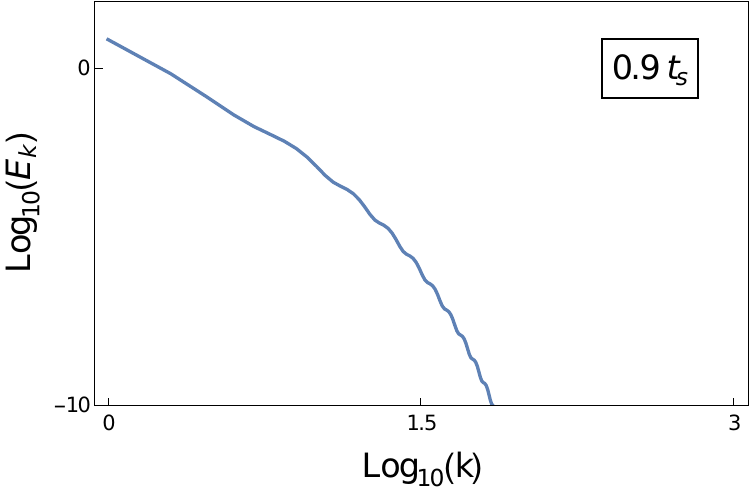} 
		\caption{}
\end{subfigure}
\begin{subfigure}{0.46\textwidth}
	\includegraphics[width=\textwidth]{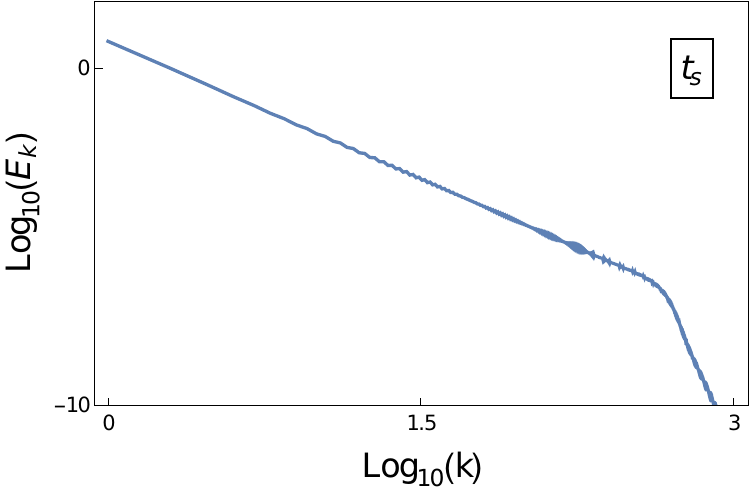}
	\caption{}
\end{subfigure}
	
\vspace{0.2cm}	
	
\begin{subfigure}{0.46\textwidth}
	\includegraphics[width=\textwidth]{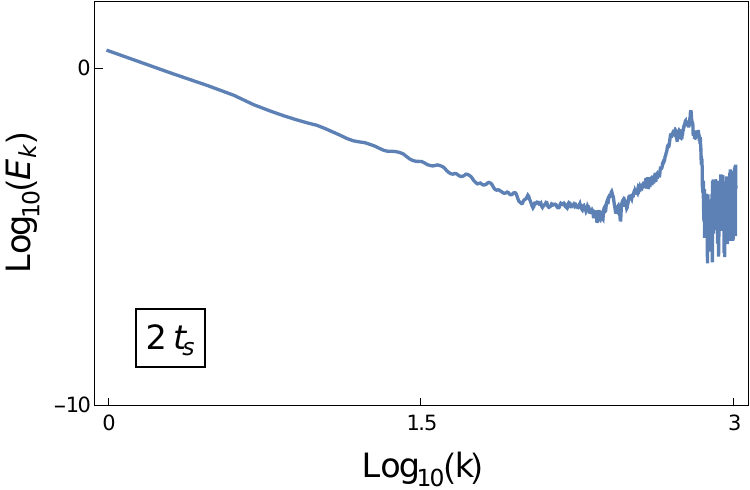}
	\caption{}
\end{subfigure}
\begin{subfigure}{0.46\textwidth}
	\includegraphics[width=\textwidth]{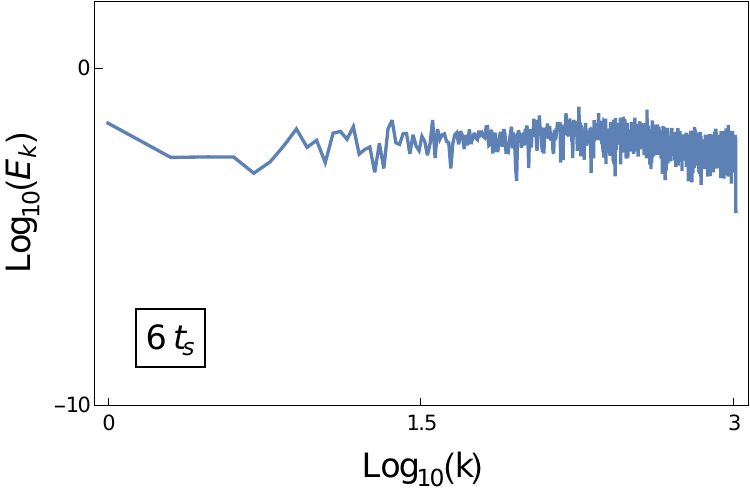}
	\caption{}
\end{subfigure}
				
		\caption{Fourier spectrum of the FPUT chain highlighting the power--law region. The panels show the spectrum at different times, measured in units of the shock time $t_s$ of the Burgers equation. Panel a: spectrum before the shock time. Panel b: spectrum at the shock time where the power--law slope is $\zeta=\frac{8}{3}$ as predicted by the Burgers dynamics. Panel c: spectrum in the turbulent range of the Burgers equation with $\zeta=2$. Panel d: spectrum at equipartition. The initial datum is given by \eqref{eq:InitialDataDisc} with $\theta=0$ with $\epsilon=0.05$, $\alpha=1$ and $\beta=0.5$. From \cite{Gallone2024}.}\label{fig:Spectra}
\end{figure*}

With equations \eqref{eq:larhoev} at hand and the initial conditions \eqref{eq:Caoa} one can now analyze the approximate dynamics and make a comparison with numerics.

The system (\ref{eq:larhoev}), (\ref{eq:indatlarho}), has the form of two decoupled, inviscid Burgers equations. From now on, we will omit the term inviscid for simplicity. The solution of the Burgers equation $U_\tau=f(U)U_x$, with initial datum $U_0(x)$, is implicitly defined by the equation $U-U_0(x+f(U)\tau)=0$ (see, e.g.~Section 3.4 in \cite{evans2010partial}). The latter identity admits an explicit solution if the implicit function theorem applies, namely, taking the derivative with respect to $U$, if
\begin{equation}
1-U_0'(x+f(U)\tau)f'(U)\tau=1-\tau\frac{d}{d\xi}f(U_0(\xi))\neq0\ ,
\end{equation}
where $\xi\equiv x+f(U)\tau$.
The above condition is satisfied for all $\tau$'s in the interval $[0,\tau_s)$, where $\tau_s$, the shock time, is given by 
\begin{equation}
\label{eq:16}
\frac{1}{\tau_s}= \max_x\frac{d}{dx}f(U_0(x))\ .
\end{equation}
The system \eqref{eq:larhoev} consists of two independent equations of the form $\lambda_\tau=f(\lambda) \lambda_x$ and $\rho_\tau=f(\rho)\rho_x$. We define the shock time of the FPUT system as $\tau_s=\min\{\tau_s^l,\tau_s^r\}$, whereas the left and right shock times $\tau^l_s$ and $\tau_s^r$ are given by
\begin{equation}
\label{eq:taulr}
\begin{split}
\frac{1}{\tau_s^l}&=\max_{x\in[0,1]}\left[\frac{d}{dx}\Phi(\lambda_0(x))\right]\ \ ;\\ 
\frac{1}{\tau_r^l}&=\max_{x\in[0,1]}\left[-\frac{d}{dx}\Phi(\rho_0(x))\right]\ .
\end{split}
\end{equation}
Here $\Phi$ is the function defined as
\begin{equation}\label{eq:Phi}
	\Phi(\lambda):=1+ \frac{ \alpha\sqrt{\epsilon}}{\sqrt{2}} \lambda \, ,
\end{equation}
where $\lambda_0$ and $\rho_0$ are given in (\ref{eq:indatlarho}). In $\Phi(\lambda_0(x))$ and $\Phi(\rho_0(x))$ we have consistently neglected terms of order $O(\epsilon)$. The explicit computation of $\tau_s^l$ and $\tau_s^r$ in \eqref{eq:taulr} for the left shock time yields
\begin{equation}
\label{eq:taul}
\tau_s^l=\left(\frac{1}{2\pi\sqrt{2 \epsilon} \alpha}\right)\frac{1}{\cos\theta} 
\end{equation}
whereas the right shock time is given by
\begin{equation}
\label{eq:taur}
\tau_s^r=\left(\frac{1}{2\pi\sqrt{2 \epsilon}\alpha}\right)\frac{1}{|\sin\theta|} \, .
\end{equation}
In the range $-\frac{\pi}{4}\leq\theta\leq\frac{\pi}{4}$, for $a$ small enough, and any $\alpha$, $\beta$, 
the inequality $\tau_s^l\leq\tau_s^r$ holds,
the equality being valid only for $\theta=\pm\frac{\pi}{4}$.   Recalling that $\tau=\frac{t}{N}$, it follows that in the same range of $\theta$ and $\epsilon$ and for any
$\alpha$, $\beta$, $\tau_s=\min\{\tau_s^l,\tau_s^r\}=\tau_s^l$, one gets 
\begin{equation}\label{eq:ts}
	t_s=\left( \frac{N}{2 \pi \sqrt{2 \epsilon} \alpha} \right) \frac{1}{\cos \theta} \, .
\end{equation}

We remark that the name ``shock time for FPUT'' can be misleading. Indeed, once $\beta > 0$ in \eqref{eq:HamEqFPUT}, the solutions to the equations of motion of the lattice model are defined for all times. Shock appears only in the continuum approximation when $h=0$. Indeed, as we shall see soon, at the Burgers shock time the Fourier Energy Spectrum of the FPUT system displays a large region with a precise and universal power-law decay. This latter fact follows from a general asymptotics of the Burgers spectrum at the shock time.

Given the generalized Burgers equation $U_\tau=f(U)U_x$, we first write its solution in Fourier series, namely $U(x,\tau)=\sum_k\hat{U}_k(\tau)e^{i 2\pi kx}$, where the Fourier coefficient $\hat U_k(\tau)$
can be expressed in terms of the initial condition $U_0(x)$ by the explicit formula
\begin{equation}\label{eq:uk}
	\hat{U}_k(\tau)=\frac{1}{2 \pi i k} \oint U'_0(x)e^{- i 2 \pi k[x-\tau f(U_0(x))]} dx \, .
\end{equation}
Taking into account that $U=U_0(x+\tau f(U))$, and introducing the variable $\xi$ such that $\xi=x+\tau f(U_0(\xi))$, one has
\begin{equation}\label{eq:Cinquanta}
\begin{split}
\hat{U}_k(\tau)=& \oint e^{-2 \pi i kx}U_0(x+f(U)\tau)\ dx \\
=&   \oint e^{-2 \pi i k[\xi-\tau f(U_0(\xi))]} U_0(\xi) \frac{d}{d\xi}\left[\xi-\tau f(U_0(\xi))\right]\ d\xi \\
=& -\frac{1}{2 \pi i k}\oint U_0(\xi)\frac{d}{d\xi}\left[e^{-2 \pi i k[\xi-\tau f(U_0(\xi))]}\right]\ d\xi\\
=&
 \frac{1}{2 \pi i k}\oint U'_0(\xi)e^{-2 \pi i k[\xi-\tau f(U_0(\xi))]}\ d\xi\ .
\end{split}
\end{equation}
Using \eqref{eq:Cinquanta} in the first equation of
\eqref{eq:larhoev}, that is with $f=\Phi$, yields \eqref{eq:uk}.

Then, for a more general class of nonlinearities $f$ and a more general class of initial data $U_0$ we prove the following.

\begin{theorem}
Let us now assume that the initial datum $U_0(x)$ of the generalized Burgers equation
$U_\tau=f(U)U_x$ satisfies the following three conditions: 
\begin{enumerate}
\item[(i)] $U_0(x)=\sum_{n=-M}^Mc_ne^{2 \pi i nx}$; $c_0=0$, $M$ finite;
\item[(ii)] $df(U_0(x))/dx$ admits a finite number $m$ of absolute maximum points $x_1,\dots,x_m\in[0,1[$;
\item[(iii)] $\gamma_j\equiv d^3f(U_0(x_j))/dx^3\neq0$.
\end{enumerate}
Then, for $k$ large
\begin{equation}
	|\hat{U}_k(\tau_s)|^2 \, \sim \, k^{-\frac{8}{3}} \, .
\end{equation}
\end{theorem}

This is a generalization of the case where the first Fourier mode is excited. Indeed, in that case, the number of critical points is two, a maximum and a minimum.

\begin{proof}
We start  by the expression (\ref{eq:uk}) for $\hat U_k(\tau)$, and split the unit integration interval into $m$ disjoint subintervals $I_1,\dots,I_m$, such that $I_j$ contains only the maximum point $x_j$ in its interior. Thus
\begin{equation}
\label{eq:uktau}
\hat{U}_k(\tau)=\frac{1}{k}\sum_{n=-M}^Mc_nn\sum_{j=1}^m\int_{I_j}e^{2 \pi i nx}
e^{-2 \pi i k[x-\tau f(U_0(x))]}\ dx\ .
\end{equation}
In the asymptotics $k\to\infty$ each of the integrals on $I_j$ is treated with the method of stationary phase (see e.g.~ Section 2.4 in \cite{erdélyi1956asymptotic}). One has to take into account that, by the definition \eqref{eq:16} of the shock time $\tau_s$, and by the hypotheses (ii) and (iii) above, in the interval $I_j$
\begin{equation}
x-\tau_s f(U_0(x))=x_j-\tau_s f(U_0(x_j))-\frac{\tau_s\gamma_j}{6}(x-x_j)^3+O((x-x_j)^4)\ .
\end{equation}
Thus, if $I_j=[x_j-a_j,x_j+b_j)$, for $k\gg 1$ and $\tau=\tau_s$ one finds
\begin{equation}
\label{eq:Ij}
\begin{split}
\int_{I_j}e^{2 \pi i nx}&e^{-2 \pi i k[x-\tau_s f(U_0(x))]}\ dx \\
&\sim\frac{e^{2 \pi i nx_j}e^{-2 \pi i k[x_j-\tau_s f(U_0(x_j))]}}{(\pi\tau_s |\gamma_j| k)^{\frac{1}{3}}}\ 
\frac{2\pi}{3^{\frac{2}{3}}\Gamma(\frac{2}{3})}\ .
\end{split}
\end{equation}
 Inserting (\ref{eq:Ij}) into (\ref{eq:uktau}) at $\tau=\tau_s$ one gets
\[
\hat u_k(\tau_s)\sim -i\left[\sum_{j=1}^m
\frac{U_0'(x_j)e^{-2 \pi i k[x_j-\tau_s f(U_0(x_j))]}}{(9\pi\tau_s|\gamma_j|)^{\frac{1}{3}}\Gamma(\frac{2}{3})}\right] k^{-\frac{4}{3}}\ ,
\] 
where $\Gamma(\frac{2}{3})$ is the Euler gamma function at $\frac{2}{3}$.
Its square modulus yields the asymptotic formula 
\begin{equation}
\label{eq:fesgen}
|\hat U_k(\tau_s)|^2\sim\left|\sum_{j=1}^m\frac{U_0'(x_j)e^{-2 \pi i k[x_j-\tau_sf(U_0(x_j))]}}{
(9\pi \tau_s |\gamma_j|)^{\frac{1}{3}}\Gamma(\frac{2}{3})}\right|^2k^{-\frac{8}{3}} \, ,
\end{equation}
which holds as $k\to+\infty$.
\end{proof}

Using this latter result, for our choice of initial data we obtain formula
\begin{equation}\label{eq:FESatSHOCK}
	E_k(t_s) \, \sim \, 0.779 \, k^{-\frac{8}{3}} \, ,
\end{equation}
 for the normalized Fourier Energy Spectrum. In fact, the function
$d\Phi(\lambda_0(x))/dx$, which enters the definition (\ref{eq:taulr}) of the shock time (\ref{eq:taul}), displays a single absolute maximum point $\hat x$ if $\epsilon$ is small enough. Then, formula (\ref{eq:fesgen})
simplifies to
\begin{equation}
\label{eq:laktaus}
|\hat \lambda_k(\tau_s)|^2\sim\left|\frac{\lambda_0'(\hat x)}{
(9\pi \tau_s d^3\Phi(\lambda_0(\hat x))/dx^3)^{\frac{1}{3}}\Gamma(\frac{2}{3})}\right|^2k^{-\frac{8}{3}} = C\ k^{-\frac{8}{3}}\ ,
\end{equation}
where the constant $C$ is independent of $k$. Assuming further that the form of the Fourier Energy Spectrum at the shock time is given by $C k^{-\frac{8}{3}}$ for all $k\geq 1$, the normalized Fourier Energy Spectrum turns out to be
\begin{equation}
\frac{E_k(t_s)}{\sum_{k>0}E_k(t_s)}=\frac{|\hat \lambda_k(\tau_s)|^2}{\sum_{k>0}|\hat \lambda_k(\tau_s)|^2}=\frac{k^{-\frac{8}{3}}}{\zeta_R(\frac{8}{3})}\ ,
\end{equation}
where $\zeta_R(s)=\sum_{k>0}k^{-s}$ is the Riemann zeta function. One finds numerically that
$\zeta_R(\frac{8}{3})=1.28419\dots$, whose reciprocal is $0.77870\dots$, which justifies formula \eqref{eq:FESatSHOCK}. 

In comparing Burgers prediction with numerics on the FPUT lattice, one sees a good agreement on long-wavelength (small $k$). Therefore, one can in principle expect that, when the FPUT solution is supported essentially on small $k$, Burgers dynamics could work as an approximate description of FPUT lattice. Indeed, this should happen at small times, once a long wavelength initial datum is given. 

We now analyze the growth of the Fourier Energy Spectrum for short times along Burgers dynamics. To this aim, we consider for simplicity the dynamics of the inviscid Burgers equation with all coefficients equal to $1$, i.e.
\begin{equation}\label{eq:Starrr}
	U_t=UU_x \, .
\end{equation}
Expanding in Fourier series, and considering only $k\geq 0$, we get
\begin{equation}\label{eq:FourierBurg}
	\frac{d \hat{U}_k}{dt} = i  \pi k \sum_{1 \leq p \leq k-1} \hat{U}_p \hat{U}_{k-p} + 2 i  \pi k \sum_{p \geq 1} \hat{U}_{k+p} \hat{U}^*_{p} \, .
\end{equation}
In order to study the initial dynamics where the high modes contain a small amount of energy, we can neglect the second term in \eqref{eq:FourierBurg}. Thus, by imposing that $\hat{U}_0=0$, we can explicitly solve the evolution equations. One immediately observes that
\[
	\begin{split}
	\frac{d \hat{U}_1}{dt}&=0  \, , \quad
	\frac{d \hat{U}_2}{dt}=i \pi (\hat{U}_1)^2 \, , \quad 
	\frac{d \hat{U}_3}{dt}=2 i \pi \hat{U}_1 \hat{U}_2 \, , \qquad\cdots
	\end{split}
\]
and then one guesses a solution of the form $\hat{U}_1(t)=\hat{U}_1(0)$, $\hat{U}_2(t) \sim t$, $\hat{U}_3(t) \sim t^2$ and so on. Thus, for the generic mode $k \in \mathbb{Z}$, one predicts
\begin{equation}\label{eq:AsymptT}
	\hat{U}_k(t) \sim t^{k-1} \, .
\end{equation}
This ansatz is compatible with \eqref{eq:FourierBurg}, since one has $\frac{d}{dt}\hat{U}_k(t)=t^{k-2}$ on the left--hand side and $t^{p-1}t^{k-p-1}=t^{k-2}$ on the right--hand side.

This analysis predicts that the initial growth of the energy of the normal modes follows the scaling law
\begin{equation}\label{eq:InitialGrowthCC}
	E_k(t) = |\hat{U}_k(t)|^2 \sim t^{2k-2} \, .
\end{equation}
This result can be also obtained using the explicit representation of the solution of the inviscid Burgers equation and the integral representation of the Bessel functions (see Appendix A of \cite{Gallone2024}). However, solving the equations using the iterative procedure described above, describes the physical origin of the scaling, which comes from an energy cascade from small to higher values of $k$. Indeed, in neglecting the second term on the right--hand side of \eqref{eq:FourierBurg}, we are neglecting the effects of a backward cascade of energy and we are keeping only the forward cascade. 

In Figure~\ref{fig:AngularN} (left panel) we interpolate the initial growth of the energies of the first modes with straight lines and we see a general agreement with eq.~\eqref{eq:InitialGrowthCC}.

\begin{figure}[t]
\includegraphics[width=0.48\textwidth]{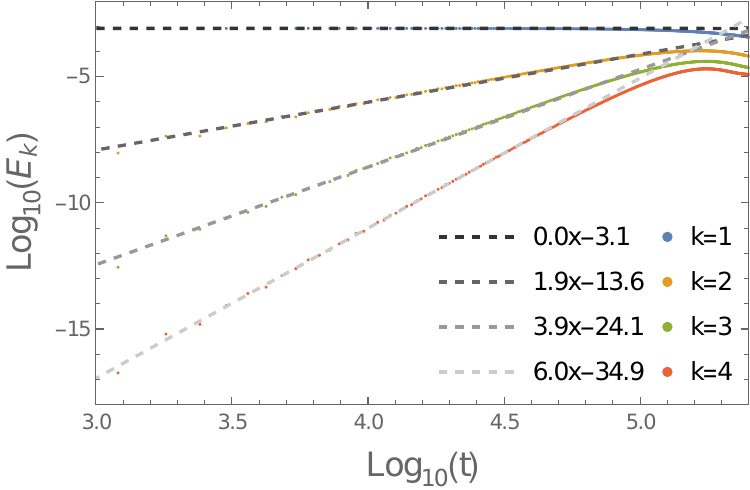}
	\includegraphics[width=0.48\textwidth]{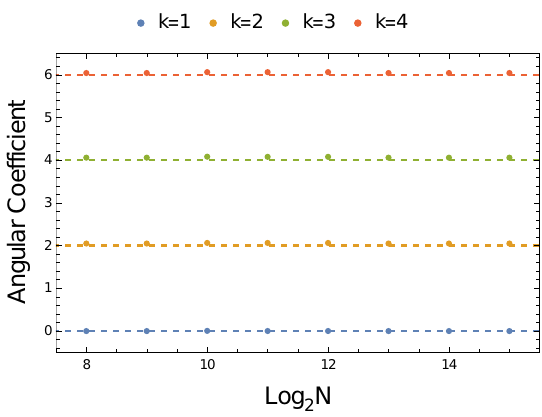}
	
	\caption{Angular coefficients of lines interpolating the growth of Fourier modes with indices $k=1,2,3,4$ for different values of $N$ showing the persistence of the slopes when increasing $N$ with fixed $\varepsilon$. The initial datum is a traveling wave excitation with $\alpha=1$, $\beta=\frac{1}{2}$ and $\epsilon=0.0005$. From \cite{Gallone2024}.}\label{fig:AngularN}
\end{figure}

In Figure~\ref{fig:AngularN} (right panel) we analyze the behavior of the angular coefficients as functions of $N$. This analysis shows that the scenario described by the Burgers equation is stable in the thermodynamic limit.

For initial data that are not exactly of the form of \eqref{eq:InitialDataDisc}, one observes that the behavior of Figure~\ref{fig:AngularN}(-left) remains the same, apart from the additional presence of small amplitude oscillations around the interpolating line. The latter phenomenon can be clearly observed in Figures 2 and 11 of \cite{Ponno2011-CHAOS}.

Differently from the asymptotics \eqref{eq:FESatSHOCK}, formula \eqref{eq:AsymptT} strongly relies on the form of the nonlinearity. Therefore, when the energy of the initial datum is increased, the monomial with higher degree in \eqref{eq:larhoev} cannot be neglected anymore and this effect dramatically changes the exponent of the power--law of the time evolution.

Before completing this section, let us briefly comment on the relation between the spectra of Figure \ref{fig:Spectra} and turbulence. According to Kolmogorov, the Fourier spectrum of a turbulent fluid has the following structure. At small $k$, energy in injected and the spectrum may depend on the mechanism of energy injection. At intermediate $k$ one observes the \emph{inertial range}, where the energy spectrum fits a power-law. At larger $k$, dissipation is responsible for an exponential decay. Note that panel (b) and (c) present precisely this structure even if, in our case, there is no presence of dissipation (it is very delicate to understand, but dispersion -- in the FPUT case -- is responsible for the exponential damping).

Needless to say, proving such a behavior -- even in simplified models -- is one of the major open problems in dispersive PDEs for which very few mechanisms of transport of energy are known. This is the problem of growth of Sobolev norms (see \cite{Colliander2010,Hani2013,Guardia2015,
Haus2015,Guardia2016,Guardia2022,Giuliani2022,Giuliani2025,Pasquali2025,MasperoMurgante,Langella2025,BLM2022}). Existence of initial data that does not thermalize, such as self-similar solutions is another very active field of research both with simplified models or stochastic PDEs \cite{Beck2024,Dolce2024,Kierkels2015}.

\section{Beta model, higher dimension and quantum prethermalization}

The point of view raising from the analysis of the FPUT problem could be summarized as follows. On some extent, when the energy is low enough, the process of thermalization of the system constitutes of two different \emph{separated} time-scales. On the first time-scale, a partial ergodization is observed. Then, thermalization occurs in a much longer timescale. This is a paradigmatic example of what is nowadays called \emph{prethermalization} and it is illustrated in the cartoon picture Figure \ref{fig:Prethermal}. 

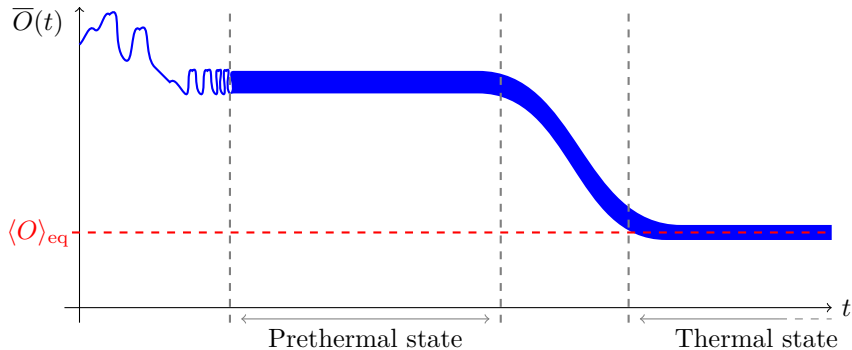
\begin{figure}[h!]
	\begin{center}
\begin{tikzpicture}
	\draw[->] (-0.2,0) -- (10,0);
	\draw[->] (0,-0.2) -- (0,4);
	
   \coordinate (A) at (0,3.5);
    \coordinate (B) at (0.2,3.7);
    \coordinate (C) at (0.4,3.9);
    \coordinate (D) at (0.6,3.3);
    \coordinate (E) at (0.8,3.7);
    \coordinate (F) at (1,3.2);
    \coordinate (G) at (1.2,3.0);
    \coordinate (H) at (1.4,2.85);
    \coordinate (I) at (1.5,3.15);
    \coordinate (J) at (1.6,2.85);
    \coordinate (K) at (1.7,3.15);
    \coordinate (L) at (1.8,2.85);
    \coordinate (M) at (1.85,3.15);
    \coordinate (N) at (1.9,2.85);
    \coordinate (O) at (1.95,3.15);
    \coordinate (P) at (2,2.85);
    
       \draw[thick, smooth,color=blue] 
        (A) to[out=45, in=135] (B)
        to[out=-45, in=135] (C)
        to[out=45, in=135] (D)
        to[out=-45, in=135] (E)
        to[out=45, in=135] (F)
        to[out=-45, in=135] (G)
        to[out=45, in=135] (H)
        to[out=-45, in=135] (I)
        to[out=45, in=135] (J)
        to[out=-45, in=135] (K)
        to[out=45, in=135] (L)
        to[out=-45, in=135] (M)
        to[out=45, in=135] (N)
         to[out=-45, in=135] (O)
        to[out=45, in=135] (P);
        
     \coordinate (UB) at (2.05,3.15);
     \coordinate (BB) at (2.05,2.85);
     \coordinate (UF) at (5.3,3.15);
     \coordinate (BF) at (5.3,2.85);
     \coordinate (TU) at (8.0,1.1);
     \coordinate (TB) at (7.8,0.9);
     \coordinate (EU) at (10,1.1);
     \coordinate (EB) at (10,0.9);
     
     \fill[blue]
     	(P) to[out=45, in=180] (UB)
     	to[out=0, in=180] (UF)
     	to[out=0, in =180] (TU)
     	to[out=0, in=180] (EU)
     	to[out=-90, in=90] (EB) 
     	to[out=180, in=0] (TB)
     	to[out=180, in=0] (BF)
     	to[out=180, in=0] (BB)
     	to[out=180, in=-45] (P);
     	
     \draw[-,dashed,color=red,thick] (-0.1,1) -- (10,1);
     \draw[-, dashed, color=gray,thick] (2,-0.2) -- (2, 4);
     \draw[-, dashed, color=gray,thick] (5.6,-0.2) -- (5.6, 4);
	 \draw[-, dashed, color=gray,thick] (7.3,-0.2) -- (7.3, 4);     	
	 
	 \node at (-0.55,3.8) {$\overline{O}(t)$};
	 \node at (10.2,0) {$t$};
	 \node at (-0.55,1) {\color{red}$\langle O \rangle_{\mathrm{eq}}$};
	 
	 \draw[<->,color=gray] (2.15,-0.15) -- (5.45,-0.15);
	 
	 \node at (3.8,-0.4) {Prethermal state};
	 
	 \draw[<-,color=gray] (7.45,-0.15)--(9.3,-0.15);
	 \draw[-,dashed, color=gray] (9.3,-0.15) -- (10,-0.15);
	 
	 \node at (9,-0.4) {Thermal state};
     	%\draw[thick, smooth, color=green]
%     	(P) to[out=-45, in=180] (BB)
 %    	to[out=0, in=180] (BF);
\end{tikzpicture}
	\end{center}
	\caption{Cartoon of the idea of prethermalization. $\overline{O}(t)$ denotes the time-average of the observable $O$, $\langle O \rangle_{\mathrm{eq}}$ denotes the value of the observable $O$ at thermal equilibrium. The time-average of $O$ has an initial transient and then it becomes almost stationary for a certain (usually, long) time. This is the prethermal state. Eventually, the time-average of $O$, on a much longer time-scale, evolves and reaches its thermal value.}\label{fig:Prethermal}
	
\end{figure}

 Thus, in other words, prethermalization refers to the fact that certain physical systems, when initialized far from equilibrium, quickly relax to long-living non-thermal states, before eventually
reaching the thermal state on much longer timescales \cite{Berges2004,Moeckel2008,Bertini2015,
Gring2012,Mori2016,Lindner2017,
Howell2019,Mallayya2019,
Huveneers2020,RubioAbadal2020,Birnkammer2022,
Collura2022}. 

In this conclusive section we briefly discuss the phenomenon of prethermalization in higher dimensional classical lattices and in quantum systems.

\subsection{Beta and higher order models} Most of these lecture notes are devoted to the analysis of FPUT model when $\alpha \neq 0$, the so-called $\alpha+\beta$ model. Clearly, something dramatic happens in the analysis of the $\beta$-model, that is setting $\alpha=0$. Indeed, in that case, the Toda chain is not tangent to the FPUT model and, in the discrete setting, the reference integrable model is the harmonic chain. 

In the continuum model, one can repeat the analysis of Section \ref{sec:FPUTIntegrable} and, instead of the pair of counter propagating KdV equations, one would deal with a pair of counter propagating mKdV equations which are
\begin{equation}\label{eq:BetaMKDV}
	\begin{split}
		\lambda_\tau \,&=\, \big({\textstyle 1+\frac{3 \beta \epsilon}{4} \langle \rho^2 \rangle }\big)\lambda_x + \frac{3 \beta \epsilon}{4} \lambda^2 \lambda_x + \frac{h^2}{24} \lambda_{xxx} \\
		\rho_\tau \,&=\, -\big({\textstyle 1+\frac{3 \beta \epsilon}{4} \langle \lambda^2 \rangle }\big)\rho_x - \frac{3 \beta \epsilon}{4} \rho^2 \rho_x-\frac{h^2}{24} \rho_{xxx} \, .
	\end{split}
\end{equation} 
Those two equations are integrable and therefore, to look for breaking of integrability, higher order computations have to be done.

Moreover, to conclude an argument as in Theorem \ref{thm:Bambusi-Ponno}, there is the lack of theorems guaranteeing the width of the strip of analiticity for the solution, as provided by Kappeler and P\"oschel for KdV.

For higher order models, one obtains generalized KdV equations which, in general, are not integrable. Those are much less studied. Normal form construction, in the continuum limit for these models are computed in \cite{Bambusi-Libro} (see also Exercise \ref{exrc:BetaModelNF}).

\subsection{Higher dimensional lattices: FPUT, ETL and KP-II}

In Section \ref{sec:FPUT-Integrable} we showed that the FPUT phenomenon can be well described by the vicinity of the FPUT lattice to certain integrable systems. Then one meets the fact that despite a certain number of integrable systems are known in spatial dimension 1, very few (if any) in higher spatial dimension. Considering models on a lattice, in absence of particular symmetries, for dimension larger than one the reference integrable model should be the harmonic chain.

Since no analogue of Toda approximating the dynamics of two- or three- dimensional lattice exists, and no analogue of KdV in dimension higher than three exists, one may expect to observe a quick thermalization process.

On the other side, available analytical tools for higher spatial dimension are a few and this makes the problem particularly hard to treat. 

Numerically, the analysis of FPUT systems in higher spatial dimensions have been performed by \cite{Ooyama1969,Benettin2005-2D,Benettin2008-2D}. They found that no FPUT phenomenon occurs. However, thermalization is not as trivial as one might expect. Indeed, for all most of the models considered in \cite{Benettin2005-2D} and \cite{Benettin2008-2D}, the authors notice that -- for large $N$ -- the spectral entropy estimate of equipartition gives $T_{\mathrm{eq}}(\epsilon) \sim \frac{1}{\epsilon}$. The most unexpected result on these investigations is the fact that boundary conditions seem to have a relevant role in the equipartition time. Indeed, the models with periodic boundary conditions exhibit longer thermalization time, that could be guessed to be $T_{\mathrm{eq}}(\epsilon) \sim \epsilon^{-\frac{5}{4}}$. This latter possibility could be related to the presence of additional symmetries in the model. 

This latter observation is, in fact, true. In  \cite{Gallone2021-2D}, the authors analyzed an electric-transmission lattice adapting the techniques of \cite{Bambusi2006} (and quickly reviewed in Section \ref{sec:FPUT-Integrable}) for a two-dimensional lattice. The net result is that, when the system is considered on a lattice of size $N \times N^2$, then it is possible to approximate the dynamics using the integrable Kadomtsev--Petviashvili-II (KP-II) equation. Then, using that analytic initial data remains analytic along the flow of the KP-II equation, one is able to prove the existence of a FPUT-like scenario for this weakly two-dimensional lattice model.

Then, it is worth mentioning that KP-II equation has been shown to be a good approximation of long-wavelength small-amplitude initial conditions with a weak transverse modulation by Hristov, Pelinovski and Schneider \cite{Hristov2022,Pelinovsky2023}

Last, we mention a different formal approach, supported by numerical data with $N=32$ particles, that was developed in \cite{Onorato2015} (see also \cite{Onorato2023} for a review) using the theory of wave turbulence. This approach, which at present is only formal, is certainly worthy of attention from the perspective of the mathematical-physics community. 

The theory of wave turbulence is a growing research field in mathematical physics, both concerning the analysis of the so-called kinetic equation, and its derivation from hydrodynamic models with random initial data \cite{HaniDeng,Buckmaster2021} and with stochastic forcing \cite{GrandeHani}.

\subsection{Prethermalization in quantum systems}
In this short section we want to discuss, without any attempt of completeness, the mathematically rigorous results on prethermalization for quantum systems. The general setting, which is relevant for experimental applications, is a time-dependent perturbation of a time-independent many-body Hamiltonian.

Then, it is somehow expected that generic time-dependent systems \emph{do not} have any constant of motion and then, by unjustified used of the ergodic theorem \ref{thm:ErgodicEquivalence},  the system is expected to be ergodic or even, mixing on the whole phase space. 

Nevertheless, certain physical systems, when initialized far from equilibrium, quickly relax to long-living non-thermal states, and reach thermal equilibrium only after a much longer time-scale. The long-living state is often referred to as the  prethermal state. In general, the 
dynamics of a model in its prethermal phase is not necessarily described by an integrable
Hamiltonian. This is the case, for example, of certain systems subjected to external drivings, mathematically encoded in a time dependent Hamiltonian: in the cases when the system
prethermalizes, during the prethermal state the dynamics is often well-approximated by the
dynamics generated by a time independent Hamiltonian, but this in general guarantees an
approximate conservation of energy only. This approximate conservation of energy is enough to avoid quick thermalization and the system can -- at most -- explore a small region around the surface at constant energy for this scale of times. 

The fact that driven systems approach a long-standing prethermal state has been used in laboratory to produce  states with
certain controlled features given by external drivings. It has been predicted theoretically
that in absence of Many-Body Localization, the time dependent perturbation increases the
energy of the system, which eventually reaches thermal equilibrium in a featureless state,  often improperly referred to as a ``infinite-temperature Gibbs state'' \cite{DAlessio2014-PRET,Lazarides2014,Ponte2015}. In the particular case of a fast periodic driving, the 
heating rate of the system is suppressed and for a very long time, which is exponential in the
frequency \cite{Abanin2017}, the system remains in a prethermal state. This prethermal state often presents
interesting, non trivial features, and it is now customary to refer to the phases arising in this
way as Floquet phases of matter \cite{Eckhardt2022,Potter2016,Ye2021,Zhang2022,Zhang2022-BB}.

In the case of quasi-periodic drivings, recent theoretical works analyzed the possibility of
creating prethermal states which exhibit non-trivial thermodynamic features that are different
from the Floquet ones \cite{Lapierre2020,Long2022,Martin2017,Qi2021,
Martin2022,Zhao2021,Zhao2022}. Experimental works confirmed this
possibility \cite{Boyers2020,He2023,Malz2021}. 

Mathematically rigorous investigations on
this phenomenon have also been initiated, especially in the context of quantum many-body
systems in presence of external drivings \cite{Abanin2017,DeRoeckVerreet,Else2020,Gallone-Langella-2024}, and in the context of single-body quantum Hamiltonians with time-dependent external potential (see \cite{Gio_Zend,BL-QN,BL-QN2}).

These works adapt techniques from perturbation theory of classical mechanics to the quantum setting. As an example, let us briefly discuss the breakthrough result by Abanin and co-workers on fast driven Floquet systems.

The Hilbert space of a quantum spin chain is given in terms of the local Hilbert space on each site. If $\mathfrak{h}=\mathbb{C}^q$ for some $q\geq 2$ is the local Hilbert space, then the total Hilbert space is given by $\mathfrak{H}:= \otimes_{x \in \Lambda_L} \mathfrak{h}$, where $\Lambda_L$ is $\mathbb{Z}^d \cap [-L/2,L/2)^d$ for some $L>0$. Then, in mathematical physics, a certain important role is played by local Hamiltonians, which models lattice systems interacting with finite range potentials. We decide to state the result for the $d$-dimensional quantum Ising chain whose Hamiltonian is
\begin{equation}\label{eq:QuantumH}
	H \, = \, -J \sum_{x \in \Lambda_L} \sum_{j=1}^d \sigma_{x}^{(3)} \sigma_{x+e_j}^{(3)} + a \sum_{S \subset \Lambda} f(S) \prod_{x \in S} \sigma_x^{(3)}
\end{equation}
with $f(S)$ being a function supported only on connected sets with $|S| \leq M < \infty$.  This Hamiltonian is a model for quantum magnets with next-to-nearest neighbor interaction ruled by $J$ and local many body interaction of strength $a \in \mathbb{R}$. For the model described by the Hamiltonian \eqref{eq:QuantumH}, the total magnetization along the $(3)$-axis is defined as 
\begin{equation}
	M^{(3)} \, := \, \sum_{x \in \Lambda} \sigma_x^{(3)}
\end{equation}
is a constant of motion, i.e.~$[M^{(3)},H]=0$.

If we now add a time-dependent perturbation like a fast rotating magnetic field described by the operator
\begin{equation}
	W(t) \, := \, \cos(\omega t) \sum_{x \in \Lambda} \sigma_x^{(1)} + \sin(\omega t) \sum_{x \in \Lambda} \sigma_x^{(2)} \, ,
\end{equation}
we can ask what is the effect on the total magnetization. The time-evolution of total magnetization $M^{(3)}(t)$ is described by Heisenberg equations
\begin{equation}
	\frac{d M^{(3)}}{dt}(t) \, = \, [M^{(3)}(t),H+W(t)]=[M^{(3)}(t),W(t)] \, ,
\end{equation}
and an a-priori estimate as the one in Section \ref{sec:IntegrablePerturbations} yields 
\begin{equation}
	\frac{1}{|\Lambda|} \Vert M(t)-M(0) \Vert_{\mathrm{op}} \lesssim \frac{t}{\omega} \, ,
\end{equation}
(where $\Vert \cdot \Vert_{\mathrm{op}}$ denotes the operator norm.
In fact, the situation is more subtle and if $\omega$ is large enough, then it is possible to prove that 
\begin{equation}
	\frac{1}{|\Lambda|} \Vert M(t) - M(0) \Vert_{\mathrm{op}} \lesssim \omega^{-\frac{1}{2}} \, , \qquad |t| \lesssim e^{\omega} \, .
\end{equation}
That is, the magnetization of the system remains substantially unaltered for times that are exponentially long in the frequency of the driving. In other words, the quantum system will not thermalize earlier than $t \sim e^{\omega}$.

\vspace{1cm}

{\footnotesize
\noindent\textbf{Acknowledgments.} I would like to thank the organizers of the summer school ASIDE-15 who invited me to give a course on this subject and stimulated me to write these lecture notes. These notes would never have been written without the continuous discussions with my colleagues and collaborators on the subject: Ricardo Grande, Beatrice Langella, Alberto Maspero, Stefano Pasquali, Bob Rink, Stefano Ruffo and Giulio Ruzza. A special thanks goes to Giancarlo Benettin for having introduced me to the problem and for the continuous enthusiasm he puts in his research, and to Antonio Ponno for his invaluable contributions, insightful discussions, and unwavering support throughout our work.}

%\bibliography{mybib}{}
%\bibliographystyle{plain}

\textbf{} \\
\vspace{2cm}

\end{document}